\documentclass[12pt]{article}
\pdfoutput=1
\usepackage[margin=.6in]{geometry}
\usepackage{amsmath, amssymb, amsthm}
\usepackage{graphicx}
\usepackage[linesnumbered, ruled]{algorithm2e}
\usepackage{natbib}
\usepackage{bbold}
\usepackage{caption}
\usepackage{subcaption}
\usepackage{placeins}
\usepackage{rotating}
\usepackage{authblk}
\usepackage{comment}
\usepackage{xcolor}

\usepackage[hypertexnames=false]{hyperref}
\hypersetup{pdftex,colorlinks=true,allcolors=blue}
\usepackage{hypcap}

\SetArgSty{textnormal} 

\newtheorem{lemma}{Lemma}
\newtheorem{prop}{Proposition}

\newcommand\given{{\, | \,}}
\newcommand\giventh{{\,;\,}}
\newcommand\col{{\hspace{0.6mm}:\hspace{0.6mm}}}
\newcommand{\1}{\mathbb 1}
\renewcommand{\P}{\mathcal P}
\renewcommand{\S}{\mathcal S}
\newcommand{\T}{\mathcal T}
\newcommand{\R}{\mathcal R}

\newcommand{\sr}[1]{\,{#1}\,} 
\definecolor{green}{rgb}{0,0.5,0}

\title{Markov chain Monte Carlo algorithms with sequential proposals}
\author{Joonha Park\thanks{Email: \texttt{joonhap@bu.edu}}, Yves Atchad\'e}
\affil{Boston University}
\date{}

\begin{document}

\maketitle

\abstract{
  We explore a general framework in Markov chain Monte Carlo (MCMC) sampling where sequential proposals are tried as a candidate for the next state of the Markov chain.
  This sequential-proposal framework can be applied to various existing MCMC methods, including Metropolis-Hastings algorithms using random proposals and methods that use deterministic proposals such as Hamiltonian Monte Carlo (HMC) or the bouncy particle sampler.
  Sequential-proposal MCMC methods construct the same Markov chains as those constructed by the delayed rejection method under certain circumstances.
  In the context of HMC, the sequential-proposal approach has been proposed as extra chance generalized hybrid Monte Carlo (XCGHMC).
  We develop two novel methods in which the trajectories leading to proposals in HMC are automatically tuned to avoid doubling back, as in the No-U-Turn sampler (NUTS).
  The numerical efficiency of these new methods compare favorably to the NUTS.
  We additionally show that the sequential-proposal bouncy particle sampler enables the constructed Markov chain to pass through regions of low target density and thus facilitates better mixing of the chain when the target density is multimodal.
}

\section{Introduction}
Markov chain Monte Carlo (MCMC) methods are widely used to sample from distributions with analytically tractable unnormalized densities.
In this paper, we explore a MCMC framework in which proposals for the next state of the Markov chain are drawn sequentially.
We consider the objective of obtaining samples from a target distribution on a measurable space $(\mathbb X, \mathcal X)$ with density
\[ \bar{\pi}(x) := \frac{\pi(x)}{Z} \]
with respect to a reference measure denoted by $dx$, where $\pi(x)$ denotes an unnormalized density, and $Z$ denotes the corresponding normalizing constant.
MCMC methods construct Markov chains such that, given the current state of the Markov chain $X^{(i)}$, the next state $X^{(i+1)}$ is drawn from a kernel which has the target distribution $\bar\pi$ as its invariant distribution.
The widely used Metropolis-Hastings (MH) strategy constructs a kernel with a specified invariant distribution in the following two steps \citep{metropolis1953equation, hastings1970monte}.
First, a proposal $Y$ is drawn from a proposal kernel, and second, the proposal is accepted as $X^{(i+1)}$ with a certain probability.
When the proposal is not accepted, the next state of the chain is set equal to the current state $X^{(i)}$.
The acceptance probability depends on the target density and the proposal kernel density at $X^{(i)}$ and $Y$ in a way that ensures that $\bar\pi$ is a stationary density of the constructed Markov chain.

The typical size of proposal increments and the mean acceptance probability affect the rate of mixing of the constructed Markov chain and thus the numerical efficiency of the algorithm.
There is often a balance to be made between the size of proposal increments and the mean acceptance probability.
Theoretical studies on this trade-off have been carried out for several widely used algorithms, such as random walk Metropolis \citep{roberts1997weak}, Metropolis adjusted Langevin algorithm (MALA) \citep{roberts1998optimal}, or Hamiltonian Monte Carlo (HMC) \citep{beskos2013optimal}, in an asymptotic scenario where the target density is given by the product of $d$ identical copies of a one dimensional density and where $d$ tends to infinity.
These results suggest that the optimal balance can be made by aiming at a certain value of the mean acceptance probability which depends on the algorithm but not on the target density, provided that the marginal density satisfies some regularity conditions.

Alternative methods to the basic Metropolis-Hastings strategy have been proposed to improve the numerical efficiency beyond the optimal balance between the proposal increment size and the acceptance probability.
The multiple-try Metropolis method by \citet{liu2000multiple} makes multiple proposals given the current state of the Markov chain and select one of them as a candidate for the next state of the Markov chain.
\citet{calderhead2014general} proposed a different algorithm that makes multiple proposals and allows more than one of them to be taken as the samples in the Markov chain.
Since multiple proposals can be made independently in these methods, parallelization can increase computational efficiency.
These methods make a preset number of proposals conditional on the current state in the Markov chain at each iteration.

Developments in various other directions have been made to improve the numerical efficiency of MCMC sampling.
Adaptive MCMC methods use transition kernels that adapt over time using the information about the target distribution provided by the past history of the constructed chain \citep{haario2001adaptive, andrieu2008tutorial}.
The update scheme for the transition kernel is designed to induce a sequence of transition kernels that converges to one that is efficient for the target distribution.
The convergence of the law of the constructed chain and the rate of convergence have been studied under certain sets of conditions \citep{haario2001adaptive, atchade2005adaptive, andrieu2006ergodicity, andrieu2007efficiency, roberts2007coupling, atchade2010limit, atchade2012limit}.
Note however that the performance of an adaptive MCMC algorithm is limited by the efficiencies of the candidate transition kernels.
In a different approach, \citet{goodman2010ensemble} proposed using ensemble samplers that construct Markov chains that are equally efficient for all target distributions that are affine transformations of each other.
These methods draw information about the shape of the target distribution from parallel chains which jointly target the product distribution given by identical copies of the target density.

There also exist a class of methods that address difficulties in sampling from multimodal distributions using local proposals.
Methods in this class include parallel tempering \citep{geyer1991markov, hukushima1996exchange}, simulated tempering \citep{marinari1992simulated}, and the equi-energy sampler \citep{kou2006equi}.
In these methods, the mixing of the constructed Markov chain is aided by a set of other Markov chains that target alternative distributions for which the moves between separated modes happen more frequently.
The equi-energy sampler bears a similarity with the approach of slice sampling, where a new sample is obtained within a randomly chosen level set of the target density \citep{roberts1999convergence, mira2001perfect, neal2003slice}.

In this paper, we explore a novel approach where proposals are drawn sequentially conditional on the previous proposal in each iteration.
The proposal draws continue until a desired number of ``acceptable'' proposals are made, so the total number of proposals is variable.
A key element in this approach is that the decision of acceptance or rejection of proposals are coupled via a single uniform$(0,1)$ random variable drawn at the start of each iteration.
This feature makes a straightforward generalization of the Metropolis-Hastings acceptance or rejection strategy.
The approach is applicable to a wide range of commonly used MCMC algorithms, including ones that use proposal kernels with well defined densities and others that use deterministic proposal maps, such as Hamiltonian Monte Carlo \citep{duane1987hybrid} or the recently proposed bouncy particle sampler \citep{peters2012rejection, bouchard2018bouncy}.
We will demonstrate that the sequential-proposal approach is flexible; it is possible to make various modifications in order to develop methods that possess specific strengths.

The advantage of the sequential-proposal approach can be explained using Peskun-Tierney ordering \citep{peskun1973optimum,tierney1998note,andrieu2019peskun}.
Suppose two transition kernels $P_1$ and $P_2$ defined on $(\mathbb X, \mathcal X)$ are reversible with respect to $\bar\pi$:
\[
\int \mathbb 1_{A\times B}(x,y) P_j (x,dy)\pi(x) dx = \int \1_{B\times A}(x,y) P_j(x,dy)\pi(x)dx, \qquad \forall A,B\sr\in \mathcal X, ~ j\sr=1,2.
\]
The transition kernel $P_1$ is said to dominate $P_2$ off the diagonal if
\[
P_1(x, A\sr\setminus\{x\}) \geq P_2(x, A\sr\setminus\{x\}), \qquad \forall x \sr\in \mathbb X,~ \forall A\sr\in \mathcal X.
\]
For a $\mathcal X$-measurable function $f$ such that $\int f^2(x) \bar\pi(x) dx < \infty$, consider an estimator of $\int f(x)\bar\pi(x)dx$ given by $\hat I(f) := \frac{1}{M}\sum_{i=1}^M f(X^{(i)})$.
Define a scaled asymptotic variance of the estimator with respect to the kernels $P_j$, $j\sr=1,2$ by
\[
v(f,P_j) := \lim_{M\to\infty} M \text{Var}_{P_j}\left( \hat I(f) \right), \qquad j\sr=1,2,
\]
where $X^{(0)}$ is assumed to be drawn from the stationary density $\bar\pi$ and $X^{(i)}$ is drawn from $P_j(X^{(i-1)},\cdot)$, $i\sr\geq 1$.
Provided that $P_1$ dominates $P_2$ off the diagonal, we have $v(f,P_1) \leq v(f,P_2)$ \citep[Theorem 4]{tierney1998note}.
Suppose now we denote by $P_1(x,A\sr\setminus\{x\})$ the conditional probability that $X^{(i)}$ is in $A\sr\setminus\{x\}$ given $X^{(i-1)}\sr= x$ when a sequential-proposal method is used and by $P_2(x,A\sr\setminus\{x\})$ the same conditional probability when a standard MH method is used.
Then $P_1$ dominates $P_2$ off the diagonal, because the sequential-proposal method tries additional proposals when the first proposal is rejected.
Due to Peskun-Tierney ordering, the asymptotic variance of the estimate $\hat I(f)$ from the sequential-proposal method is always smaller than or equal to that from the standard MH.
In fact, a similar argument also motivated the development of the delayed rejection method \citep{tierney1999some, green2001delayed}, which we will show to be equivalent to the sequential-proposal approach under certain circumstances in terms of the law of the constructed Markov chains.

After the initial \textsf{arXiv} posting of this manuscript, we became aware of the extra chance generalized hybrid Monte Carlo (XCGHMC) of \citet{campos2015extra} that develops a method equivalent to the sequential-proposal approach in the context of HMC.
Our work departs from \citet{campos2015extra} in two important ways.
First, the current paper develops sequential-proposal MCMC more broadly, and as a general recipe for improving on MCMC algorithms.
Second, in applying the idea to HMC our emphasis differs from that in \citet{campos2015extra}.
One of the main challenges in using HMC in practice is the tuning of the number of leapfrog jumps.
Motivated by this issue, we use the sequential-proposal MCMC framework to develop two novel HMC methods that automatically tune the lengths of the leapfrog trajectories in such a way that the well-known double-backing issue in HMC is avoided, in the same spirit as the No-U-Turn sampler (NUTS) of \citet{hoffman2014no}.
Our numerical results on a multivariate normal distribution and a real data example show that the efficiencies of these new methods compare favorably to that of the NUTS.

Our paper is organized as follows.
In Section~\ref{sec:spMH} we explain the sequential-proposal approach in the case where proposal kernels have well defined densities.
We call the resulting methods as sequential-proposal Metropolis-Hastings algorithms.
An equivalence between between our approach and the delayed rejection method of \citet{mira2001metropolis} under certain settings is discussed in this section.
In Section~\ref{sec:sppdDescrip}, we explain how the sequential-proposal approach can be applied to a class of MCMC algorithms that use deterministic proposal maps.
Based on this formulation, we develop two variants of the NUTS algorithm in Section~\ref{sec:HMC}.
Section~\ref{sec:conclusion} gives a summarizing conclusion.
Proofs of most theoretical results are given in Appendices~\ref{sec:proof_spMH}--\ref{sec:proof_spNUTS}.
In Appendix~\ref{sec:BPS} we propose a novel discrete time bouncy particle sampler method based on the sequential-proposal approach and demonstrate some desirable numerical properties that enable faster mixing of the Markov chain for multimodal target distributions.

\section{Sequential-proposal Metropolis-Hastings algorithms}\label{sec:spMH}
\subsection{Sequential-proposal Metropolis algorithm}
We will first explain the sequential-proposal approach when the proposal kernel has well defined density with respect to the reference measure of the target density $\bar\pi$.
For a simpler presentation, we will first describe a sequential-proposal Metropolis algorithm, which uses a proposal kernel with symmetric density.
Various generalizations will be introduced in Section~\ref{sec:spMHgen}.
In standard Metropolis algorithms, given the current state $X^{(i)}\sr=x$ at the $i$-th iteration of the algorithm, the proposal $Y$ is drawn from a probability kernel with conditional density $q(y\given x)$ that is symmetric in the sense that $q(y\given x) = q(x\given y)$ for all $x,y \sr\in \mathbb X$.
The proposal $Y\sr=y$ is accepted with the probability
\[
\min\left(1, \frac{\pi(y)}{\pi(x)}\right).
\]
This is often implemented by drawing a uniform random variable $\Lambda\sim \text{unif}(0,1)$ and accepting the proposal by setting $X^{(i+1)} \gets Y$ if and only if $\Lambda < \frac{\pi(y)}{\pi(x)}.$
If $Y$ is not accepted, the algorithm sets $X^{(i+1)} \gets X^{(i)}$.

We will call $Y_1$ the first proposal drawn from $q(\cdot \given X^{(i)})$.
The proposal $Y_1$ is rejected if and only if a uniform random number $\Lambda \sim \text{unif}(0,1)$ is greater than or equal to $\pi(Y_1)/\pi(X^{(i)})$.
If rejected, a second proposal $Y_2$ is drawn from $q(\cdot \given Y_1)$.
The second proposal is accepted if and only if $\Lambda < \pi(Y_2)/\pi(X^{(i)})$ using the same value of $\Lambda$ used previously.
If accepted, the algorithm sets $X^{(i+1)} \gets Y_2$.
In the case where $Y_2$ is rejected, a third proposal is drawn from $q(\cdot \given Y_2)$ and checked for acceptability using the same type of criterion, $\Lambda < \pi(Y_3)/\pi(X^{(i)})$.
This procedure is repeated until an acceptable proposal is found or until a preset number $N$ of proposals are all rejected, whichever is reached sooner.
In the case where all $N$ proposals are rejected, the algorithm sets $X^{(i+1)}\gets X^{(i)}$.
A pseudocode for a sequential-proposal Metropolis algorithm is given in Algorithm~\ref{alg:spM}.
The algorithm reduces to a standard Metropolis algorithm if we set $N\sr=1$.

\begin{figure}[t]
  \centering
  \scalebox{.89}{\begin{minipage}{\textwidth}
\begin{algorithm}[H]
  \SetKwInOut{Input}{Input}\SetKwInOut{Output}{Output}
  \Input{Maximum number of proposals, $N$\\
  Symmetric proposal kernel, $q(y \given x)$\\
  Number of iterations, $M$}
  \vspace{1ex}
  \Output{A draw of Markov chain, $\left(X^{(i)}\right)_{i\in 1:M}$}
  \vspace{1ex}
  \textbf{Initialize:} Set $X^{(0)}$ arbitrarily
  
  \For {$i\gets 0\col M{-}1$}{
    Draw $\Lambda\sim \text{unif}(0,1)$\\
    Set $X^{(i+1)} \gets X^{(i)}$\\
    Set $Y_0 \gets X^{(i)}$\\
    \For {$n \gets 1\col N$} {
      Draw $Y_n \sim q(\cdot \given Y_{n-1})$\\
      \If{$\Lambda < \frac{\pi(Y_n)} {\pi(Y_0)}$}
      { Set $X^{(i+1)} \gets Y_n$\\
        \texttt{break}
      }
    }
  }
  \caption{A sequential-proposal Metropolis algorithm}
\label{alg:spM}
\end{algorithm}
  \end{minipage}}
  \end{figure}

We will now show that the sequential-proposal Metropolis algorithm just described constructs a reversible Markov chain with respect to the target distribution with density $\bar\pi$.
Throughout this paper, for two integers $n$ and $m$, we will denote by $n\col m$ the sequence $(n, n\sr+1, \dots, m)$ if $n \leq m$ and the sequence $(n, n\sr-1, \dots, m)$ if $n > m$.
Also, given a sequence $(a_n)_{n\geq 0} = (a_0, a_1, a_2, \dots)$, we will denote by $a_{n:m}$ the subsequence $(a_j)_{n\leq j \leq m}$.

\begin{prop}\label{prop:detailedbalance_spM}
  Algorithm~\ref{alg:spM} constructs a reversible Markov chain $\left(X^{(i)}\right)$ with respect to the target density $\bar\pi$. 
\end{prop}
\begin{proof}
  We will show that the detailed balance equation
  \[
  \mathcal P[X^{(i)}\sr\in A,\, X^{(i+1)}\sr\in B] = \mathcal P[X^{(i)}\sr\in B,\, X^{(i+1)}\sr\in A]
  \]
  holds for every pair of measurable subsets $A$ and $B$ of $\mathbb X$, provided that $X^{(i)}$ is distributed according to $\bar\pi$.
  We will write $Y_0 := X^{(i)}$, and the subsequent proposals as $Y_1, Y_2, \dots, Y_N$.
  The case where the $n$-th proposal $Y_n$ is taken for $X^{(i+1)}$ will be considered; then the claim of detailed balance will follow by combining the cases for $n$ in $1\col N$ and the case where all proposals are rejected.
  Under the assumption that $X^{(i)}$ is distributed according to $\bar\pi$, the probability that $X^{(i)}$ is in $A$ and the $n$-th proposal is in $B$ and taken as $X^{(i+1)}$ is given by
  \begin{equation}\begin{split}
    &\mathcal P[X^{(i)}\sr\in A,\, X^{(i+1)}\sr\in B,\, \text{the }n\text{-th proposal is taken as }X^{(i+1)}]\\
    &\qquad= \int \1_A(y_0) \1_B(y_n) \bar\pi(y_0) q(y_1\given y_0) \cdots q(y_n \given y_{n-1}) \\
    &\qquad\qquad\cdot \1\left[\Lambda \geq \frac{\pi(y_1)}{\pi(y_0)}\right] \cdots \1\left[\Lambda \geq \frac{\pi(y_{n-1})}{\pi(y_0)} \right] \1\left[\Lambda < \frac{\pi(y_n)}{\pi(y_0)}\right] \1[0<\Lambda<1] d\Lambda\, dy_0\, dy_1 \dots dy_n,
    \end{split}
    \label{eqn:prob_AB_nth_accep}
  \end{equation}
  where $\1_A$ denotes the indicator function for the set $A$ and $\1[\cdot]$ denotes the indicator function of the event specified between the brackets.
  The quantity
  \begin{equation}
  \1\left[\Lambda \geq \frac{\pi(y_1)}{\pi(y_0)}\right] \cdots \1\left[\Lambda \geq \frac{\pi(y_{n-1})}{\pi(y_0)} \right] \1\left[\Lambda < \frac{\pi(y_n)}{\pi(y_0)}\right] \1[0<\Lambda<1]
  \label{eqn:spM_indicator}
  \end{equation}
  is equal to unity if and only if
  \begin{equation}
    \Lambda \geq \max_{k\in 1:n-1} \frac{\pi(y_k)}{\pi(y_0)} \quad \text{and} \quad \Lambda < \min\left(1, \frac{\pi(y_n)}{\pi(y_0)}\right).
    \label{eqn:spM_lambda_cond}
  \end{equation}
  It can be readily observed that for real numbers $x$, $a$, and $b$, the conditions $x \geq a$ and $x < b$ are satisfied if and only if $x \in [ \min\{a,b\},\, b)$, where the interval length is given by $b\sr-\min(a,b)$.
  Thus the interval length corresponding to the conditions \eqref{eqn:spM_lambda_cond} is given by
  \[
  \min\left(1, \frac{\pi(y_n)}{\pi(y_0)}\right) - \min\left( 1, \frac{\pi(y_n)}{\pi(y_0)}, \max_{k\in 1:n-1} \frac{\pi(y_k)}{\pi(y_0)} \right),
  \]
  which gives the integral of \eqref{eqn:spM_indicator} over $\Lambda$.
  It follows that \eqref{eqn:prob_AB_nth_accep} is equal to
  \begin{equation}\begin{split}
      &\int \1_A(y_0) \1_B(y_n) \bar\pi(y_0) \prod_{k=1}^n q(y_k\given y_{k-1}) \cdot \left[ \min\left(1, \frac{\pi(y_n)}{\pi(y_0)}\right) - \min\left(1, \frac{\pi(y_n)}{\pi(y_0)}, \max_{k\in 1:n-1}\frac{\pi(y_k)}{\pi(y_0)} \right) \right] dy_{0:n}\\
      &= \frac{1}{Z} \int \1_A(y_0) \1_B(y_n) \prod_{k=1}^n q(y_k \given y_{k-1}) \cdot \left[ \min\{\pi(y_0), \pi(y_n)\} - \min\{\pi(y_0), \pi(y_n), \max_{k\in 1:n-1} \pi(y_k)\} \right] dy_{0:n}.
      \end{split}
    \label{eqn:prob_AB_integral}
  \end{equation}
  If we change the notation of the dummy variables by writing $y_0 \gets y_n$, $y_1 \gets y_{n-1}$, $\dots$, $y_n \gets y_0$, then \eqref{eqn:prob_AB_integral} is given by
  \begin{equation}
    \frac{1}{Z} \int \1_A(y_n)\1_B(y_0) \prod_{k=1}^n q(y_k\given y_{k-1}) \left[ \min \{ \pi(y_n), \pi(y_0) \} - \min \{ \pi(y_n), \pi(y_0), \max_{k\in 1:n-1} \pi(y_k) \} \right] dy_{0:n},
    \label{eqn:spM_reversed}
  \end{equation}
  where we have used the fact that the kernel density $q$ is symmetric; that is, $q(y_{k-1}\given y_k) = q(y_k \given y_{k-1})$, for $k\in 1\col n$.
  It is now obvious that \eqref{eqn:spM_reversed} is equal to the quantity obtained by swapping the positions of $A$ and $B$ in \eqref{eqn:prob_AB_nth_accep}.
  Thus we see that
  \begin{multline}
  \mathcal P[X^{(i)} \sr\in A, \, X^{(i+1)} \sr\in B, \text{the }n\text{-th proposal is taken as }X^{(i+1)}]\\
  = \mathcal P[X^{(i)} \sr\in B, \, X^{(i+1)} \sr\in A, \text{the }n\text{-th proposal is taken as }X^{(i+1)}].
  \label{eqn:spM_nth_accep}
  \end{multline}
  In the case where all $N$ proposals are rejected, the algorithms sets $X^{(i+1)} \gets X^{(i)}$.
  Thus,
  \begin{multline}
    \mathcal P[X^{(i)} \sr\in A, \, X^{(i+1)} \sr\in B, \text{all }N\text{ proposals are rejected}]\\
    = \mathcal P[X^{(i)} \sr\in A, \, X^{(i)} \sr\in B, \text{all }N\text{ proposals are rejected}],
    \label{eqn:spM_all_rejected}
  \end{multline}
  which is obviously unchanged under the swap of $A$ and $B$.
  Thus summing \eqref{eqn:spM_nth_accep} over all $n\sr\in 1\col N$ and adding \eqref{eqn:spM_all_rejected} gives
  \[
  \mathcal P[X^{(i)}\sr\in A,\, X^{(i+1)}\sr\in B] = \mathcal P[X^{(i)} \sr\in B, \, X^{(i+1)} \sr\in A].
  \]
\end{proof}

\subsection{Algorithm generalizations}\label{sec:spMHgen}
\begin{figure}[t]
  \centering
  \scalebox{.89}{\begin{minipage}{\textwidth}
\begin{algorithm}[H]
  \SetKwInOut{Input}{Input}\SetKwInOut{Output}{Output}
  \Input{Distribution for the maximum number of proposals and the number of accepted proposals, $\nu(N,L)$\\
  Possibly asymmetric proposal kernel, $q(y_n\given y_{n-1})$\\
  Number of iterations, $M$}
  \vspace{1ex}
  \Output{A draw of Markov chain, $\left(X^{(i)}\right)_{i\in 1:M}$}
  \vspace{1ex}
  \textbf{Initialize:} Set $X^{(0)}$ arbitrarily
  
  \For {$i\gets 0\col M{-}1$}{
    Draw $(N,L) \sim \nu(\cdot, \cdot)$\\
    Draw $\Lambda\sim \text{unif}(0,1)$\\
    Set $X^{(i+1)} \gets X^{(i)}$\\
    Set $Y_0 \gets X^{(i)}$ and $n_a \gets 0$\\
    \For {$n \gets 1\col N$} {
      Draw $Y_n \sim q(\cdot \given Y_{n-1})$\\
      \textbf{if} {$\Lambda < \frac{\pi(Y_n) \prod_{j=1}^n q(Y_{j-1} \given Y_j) } { \pi(Y_0) \prod_{j=1}^n q(Y_j \given Y_{j-1}) }$} \textbf{then} $n_a \gets n_a + 1$\\
      \If {$n_a = L$} {
        Set $X^{(i+1)} \gets Y_n$\\
        \texttt{break}
      }
    }
  }
  \caption{A sequential-proposal Metropolis-Hasting algorithm}
\label{alg:spMH}
\end{algorithm}
  \end{minipage}}
  \end{figure}

The sequential-proposal Metropolis algorithm described in the previous subsection can be generalized in various ways.
Firstly, the algorithm may use proposal kernels with asymmetric density.
The $n$-th proposal $Y_n$ is drawn from a probability kernel with density $q$ which may not satisfy $q(y\given x) = q(x \given y)$, $\forall x, y \sr \in \mathbb X$.
A proposed value $Y_n$ is deemed acceptable if and only if
\begin{equation}
  \Lambda < \frac{\pi(Y_n) \prod_{j=1}^n q(Y_{j-1} \given Y_j) } { \pi(Y_0) \prod_{j=1}^n q_j(Y_j \given Y_{j-1}) }.
  \label{eqn:accep_n_asym}
\end{equation}
Here $Y_0$ denotes the current state of the Markov chain.
Clearly, \eqref{eqn:accep_n_asym} reduces to, if the proposal density $q$ is symmetric, the acceptance probability $\pi(Y_n)/\pi(Y_0)$ in Algorithm~\ref{alg:spM}.
We call a sequential-proposal MCMC algorithm that uses a proposal kernel that has possibly asymmetric density a sequential-proposal Metropolis-Hastings algorithm.

A sequential-proposal Metropolis-Hastings algorithm can be further generalized by taking the $L$-th acceptable proposal instead as the next state of the Markov chain for general $L\geq 1$.
The algorithms previously described correspond to the case where $L\,{=}\,1$.
A pseudocode for this generalized Metropolis-Hastings algorithm is given in Algorithm~\ref{alg:spMH}.
The algorithmic parameters $N$ and $L$ may be randomly selected at each iteration, provided that they are independent of the proposals $\{Y_n \giventh n\sr\geq 1\}$ and $\Lambda$.
If there are less than $L$ acceptable proposals in the first $N$ proposals, the Markov chain stays at its current position.
The proof that Algorithm~\ref{alg:spMH} constructs a reversible Markov chain with respect to the target density $\bar\pi$ is given in Appendix~\ref{sec:proof_spMH}.

A sequential-proposal Metropolis-Hastings algorithm can also employ proposal kernels that depend on the sequence of previous proposals.
Suppose that proposals are sequentially drawn in such a way that the $k$-th candidate $Y_k$ is drawn from a proposal kernel with density $q_k(\cdot\given Y_{k-1}, \dots, Y_0)$, where $Y_{k-1}$, $\dots$, $Y_1$ denote the previous proposals and $Y_0$ denotes the current state $X^{(i)}$ in the Markov chain at the $i$-th iteration.
The candidate $Y_k$ is deemed acceptable if
\begin{equation}
  \Lambda < \frac{\pi(Y_k) \prod_{j=1}^k q_j(Y_{k-j} \given Y_{k-j+1:k})}{\pi(Y_0) \prod_{j=1}^k q_j(Y_j \given Y_{j-1:0}) }.
  \label{eqn:accep_n_path_dep}
\end{equation}
Proposals are sequentially drawn until $L$ acceptable proposals are found.
If there are less than $L$ acceptable proposals among the first $N$ proposals, the next state in the Markov chain is set to the current state, $X^{(i+1)} \sr\gets X^{(i)}$.
Suppose now that the $L$-th acceptable state is obtained by the $n$-th proposal $Y_n$ for some $n \sr \leq N$.
In the case where the proposal kernel depends on the sequence of previous proposals, in order to take $Y_n$ as the next state of the Markov chain, an additional condition needs to be checked, namely that there are exactly $L\sr-1$ numbers $k\sr\in 1\col n{-}1$ that satisfy
\begin{equation}
\Lambda < \frac{\pi(Y_k) \prod_{j=1}^{n-k} q_j(Y_{k+j} \given Y_{k+j-1:k}) \prod_{j=n-k+1}^n q_j(Y_{n-j} \given Y_{n-j+1:n})} { \pi(Y_0) \prod_{j=1}^n q_j(Y_j \given Y_{j-1:0}) }.
\label{eqn:spMH_dual_condition}
\end{equation}
If this additional condition is satisfied, $Y_n$ is taken as the next state of the Markov chain, that is $X^{(i+1)} \sr\gets Y_n$.
Otherwise, the next state is set to the current state in the Markov chain, $X^{(i+1)}\sr\gets X^{(i)}$.
A pseudocode for sequential-proposal Metropolis-Hastings algorithms that employ kernels dependent on the sequence of previous proposals are given in Appendix~\ref{sec:spMH_path_depen_kernel}.
The role of the additional condition \eqref{eqn:spMH_dual_condition} is to establish detailed balance between $X^{(i)}$ and $X^{(i+1)}$ by creating a symmetry between the sequence of proposals $Y_0\sr\to Y_1 \sr\to \cdots \sr\to Y_n$ and the reversed sequence $Y_n \sr\to Y_{n-1} \sr\to \cdots \sr\to Y_0$.
To see this, we note that the candidate $Y_n$ can be taken as the next state of the Markov chain only when there are exactly $L\sr- 1$ acceptable proposals among $Y_1$, $\dots$, $Y_{n-1}$.
The additional symmetry condition accounts for a mirror case where there are $L-1$ acceptable proposals among $Y_{n-1}$, $\dots$, $Y_1$, assuming that these proposals are sequentially drawn in the reverse order starting from $Y_n$.
A proof of detailed balance for this algorithm is also given in Appendix~\ref{sec:spMH_path_depen_kernel}.
This algorithm reduces to Algorithm~\ref{alg:spMH} in the case where the proposal kernel is dependent only on the most recent proposal.

We note that sequential-proposal Metropolis-Hastings algorithms in the case where $L\,{=}\,1$ construct the same Markov chains as those constructed by delayed rejection methods \citep{tierney1999some,mira2001metropolis,green2001delayed} when the proposal kernel depends only on the most recent proposal.
A brief description of the delayed rejection method, following \citet{mira2001metropolis}, is given as follows.
Given the current state of the Markov chain $y_0$, the first candidate value $y_1$ is drawn from $q(\cdot \given y_0)$ and accepted with probability
$$\alpha_1(y_0,y_1) = 1\land \frac{\pi(y_1)q(y_0\given y_1)}{\pi(y_0) q(y_1\given y_0)},$$
where $a\land b := \min(a,b)$.
If $y_1$ is rejected, a next candidate value $y_2$ is drawn from $q(\cdot \given y_1)$.
The acceptance probability for $y_2$ is given by
$$\alpha_2(y_0, y_1, y_2) = 1 \land \frac{\pi(y_2)q(y_1\given y_2) q(y_0 \given y_1)\{1- \alpha_1(y_2, y_1)\}}{\pi(y_0) q(y_1\given y_0) q(y_2 \given y_1)\{1- \alpha_1(y_0, y_1)\}}.$$
If $y_1,\dots, y_{n-1}$ are rejected, $y_n$ is drawn from $q(\cdot \given y_{n-1})$ and accepted with probability
$$\alpha_n(y_{0:n}) = 1\land \frac{\pi(y_n) \prod_{j=1}^n q(y_{j-1}\given y_j) \prod_{j=1}^{n-1}\{1-\alpha_j(y_{n:n-j})\}} {\pi(y_0) \prod_{j=1}^n q(y_j\given y_{j-1})  \prod_{j=1}^{n-1}\{1-\alpha_j(y_{0:j})\}}.$$
If all proposals are rejected up to a certain number $N$, the next state of the Markov chain is set to the current state $y_0$.
We show in Appendix~\ref{sec:equiv_delayedrej} the equivalence between the delayed rejection method and the sequential-proposal Metropolis-Hastings algorithm with $L\,{=}\,1$ under the case where each proposal is made depending only on the most recent proposal.
The delayed rejection method can also use proposal kernels dependent on the sequence of previous proposals to construct a reversible Markov chain with respect to the target distribution.
In this case, however, the law of the constructed Markov chain will be different from that by a sequential-proposal Metropolis-Hastings algorithm.

In our view, there are several advantages sequential-proposal Metropolis-Hastings algorithms have over the delayed rejection method:
\begin{enumerate}
\item Sequential-proposal Metropolis-Hastings algorithms are more straightforward to implement than the delayed rejection method.
  The evaluation of $\alpha_n(y_{0:n})$ in delayed rejection involves the evaluation of a sequence of reversed acceptance probabilities $\{\alpha_j(y_{n:n-j}) \giventh j \in 1\col n{-}1 \}$.
  This involves computation of a total of $\mathcal O(n^2)$ acceptance probabilities.
  In comparison, sequential-proposal Metropolis-Hastings algorithms only compare the ratio in \eqref{eqn:accep_n_asym} to a uniform random number $\Lambda$ for the same task of checking the acceptability of $y_n$.
  The algorithmic simplicity of sequential-proposal Metropolis-Hastings facilitates the use of a large number of proposals in each iteration.
  Moreover, one may choose to take the $L$-th acceptable proposal for the next state of the Markov chain for a large $L\sr>1$.
\item The sequential-proposal MCMC framework can be readily applied to MCMC algorithms using deterministic maps for proposals, as explained in Section~\ref{sec:sppdDescrip}.
  In particular, the sequential-proposal MCMC framework applies to Hamiltonian Monte Carlo and the bouncy particle sampler methods, leading to improved the numerical efficiency.
  Applications to these algorithms are discussed in Section~\ref{sec:HMC} and Appendix~\ref{sec:BPS}.
  We note that the delayed rejection method has been generalized to algorithms using deterministic maps in \citet{green2001delayed}, although only the case for the second proposal was discussed.
\item The conceptual simplicity of the sequential-proposal MCMC framework allows for various generalizations and modifications.
  For example, in Section~\ref{sec:spNUTS}, we develop sequential-proposal No-U-Turn sampler algorithms (Algorithms~\ref{alg:spNUTS1} and \ref{alg:spNUTS2}) that automatically adjust the lengths of trajectories leading to proposals in HMC, similarly to the No-U-Turn sampler algorithm proposed by \citet{hoffman2014no}.
  The proofs of detailed balance for these algorithms can be obtained by making minor modifications to the proof for the sequential-proposal Metropolis-Hastings algorithms.
\end{enumerate}

\section{Sequential-proposal MCMC algorithms using deterministic kernels}\label{sec:sppdDescrip}
The sequential-proposal MCMC framework can be applied to algorithms that use deterministic proposal kernels.
MCMC algorithms that employ deterministic proposal kernels often target a distribution on an extended space $\mathbb{X} \sr \times \mathbb{V}$ whose the marginal distribution on $\mathbb{X}$ is equal to the original target distribution $\bar\pi$.
An additional variable $V$ drawn from a distribution on $\mathbb V$ serves as a parameter for the deterministic proposal kernel.
In this section, we will explain a general class of MCMC algorithms using deterministic proposal kernels and show how the sequential-proposal scheme can be applied to these algorithms.
Applications to specific algorithms, such as HMC or the bouncy particle sampler (BPS), are discussed in subsequent sections (Section~\ref{sec:HMC} and Appendix~\ref{sec:BPS}).

\begin{figure}[t]
  \centering
  \scalebox{.89}{\begin{minipage}{\textwidth}
\begin{algorithm}[H]
  \SetKwInOut{Input}{Input}\SetKwInOut{Output}{Output}
  \Input{
    Distribution of the maximum number of proposals and the number of accepted proposals, $\nu(N,L)$\\
    Time step length distribution, $\mu(d\tau)$\\
    Velocity distribution density, $\psi(v \giventh x)$\\
    Time evolution operators, $\{\S_\tau\}$\\
    Velocity reflection operators, $\{\R_x\}$\\
    Velocity refreshment probability, $p^{\text{ref}}(x)$\\
    Number of iterations, $M$}
  \vspace{1ex}
  \Output{A draw of Markov chain, $\left(X^{(i)}\right)_{i\in 1:M}$}
  \vspace{1ex}
  \textbf{Initialize:} Set $X^{(0)}$ arbitrarily and draw $V^{(0)}\sim \psi(\,\cdot\giventh X^{(0)}).$
  
  \For {$i\gets 0\col M{-}1$}{
    Draw $N,L\sim \nu(\cdot, \cdot)$\\
    Draw $\tau \sim \mu(\cdot)$\\
    Draw $\Lambda\sim \text{unif}(0,1)$\\
    Set $X^{(i+1)} \gets X^{(i)}$ and $V^{(i+1)} \gets \R_{X^{(i)}}V^{(i)}$\\
    Set $n_a \gets 0$\\
    Set $(Y_0, W_0)\gets (X^{(i)}, V^{(i)})$\\
    \For {$n \gets 1\col N$} {
      Set $(Y_n, W_n) \gets \S_\tau(Y_{n-1}, W_{n-1})$\\
      \textbf{if} {$\displaystyle \Lambda < \frac{\pi(Y_n)\psi(W_n\giventh Y_n)}{\pi(Y_0)\psi(W_0\giventh Y_0)} \left| \det D\S_\tau^n(Y_0,W_0) \right|$} \textbf{then} $n_a \gets n_a + 1$\\
      \If {$n_a = L$} {
        Set $(X^{(i+1)}, V^{(i+1)}) \gets (Y_n, W_n)$\\
        \texttt{break}
      }
    }
    With probability $p^{\text{ref}}(X^{(i+1)})$, refresh $V^{(i+1)}\sim \psi(\,\cdot \giventh X^{(i+1)})$ 
  }
  \caption{A sequential-proposal MCMC using a deterministic kernel}
  \label{alg:sppd}
\end{algorithm}
  \end{minipage}}
\end{figure}

We suppose that the extended target distribution on $\mathbb X \times \mathbb V$ has density $\Pi(x,v)$ with respect to a reference measure denoted by $dx\, dv$.
We further assume that the original target density $\bar\pi$ equals the marginal density of $\Pi$, such that $\Pi(x,v) = \bar\pi(x) \psi(v\giventh x)$ for some $\psi(v\giventh x)$, the conditional density of $v$ given $x$.
We define a collection of deterministic maps $\S_\tau : \mathbb{X}\times\mathbb{V} \to \mathbb{X} \times \mathbb{V}$ for possibly various values of $\tau$.
In HMC and the BPS, $\S_\tau$ has an analogy with the evolution of a particle in a physical system for a time duration $\tau$.
In this analogy, the variable $x \in \mathbb X$ is considered as the position of a particle in the system and the variable $v \in \mathbb{V}$ as the velocity of the particle.
The point $\S_\tau(x,v)$ then represents the final position-velocity pair of a particle that moves with initial position $x$ and initial velocity $v$ for time $\tau$.
We suppose that the map $\S_\tau$ for each $\tau$ satisfies the following condition: 

\newcommand{\id}{\mathcal I}
\vspace{2ex}

\noindent\parbox{\linewidth}{
  \noindent \textbf{Reversibility condition.}
  \itshape
There exists a velocity reflection operator $\R_x: \mathbb{V} \to \mathbb{V}$ defined for every point $x \in \mathbb{X}$ such that
\begin{equation}
  \R_x \circ \R_x = \id, \label{eqn:RR}
\end{equation}
holds for every $x \sr\in \mathbb X$ and
\begin{equation}
  \frac{\psi(\R_x v \giventh x)}{\psi(v \giventh x)} \left| \frac{\partial \R_x v}{\partial v} \right| = 1
  \label{eqn:psiR}
\end{equation}
holds for almost every $(x,v)\sr\in\mathbb X\times \mathbb V$ with respect to the reference measure $dx\, dv$.
Furthermore, if we define a map $\T:\mathbb{X}\times \mathbb{V} \to \mathbb{X} \times \mathbb{V}$ as $\T(x,v) := (x, \R_x v)$, we have 
\begin{equation}
  \T\circ \S_\tau \circ \T\circ \S_\tau = \id. \label{eqn:TSTS}
\end{equation}
}

\noindent
Similar sets of conditions appear routinely in the literature on MCMC \citep{fang2014compressible, vanetti2017piecewise} and on Hamiltonian dynamics \citep[Section~4.3]{leimkuhler2004simulating}.
In \eqref{eqn:RR} and \eqref{eqn:TSTS}, $\id$ denotes the identity map in the corresponding space $\mathbb{V}$ or $\mathbb{X}\times\mathbb{V}$, and the symbol $\circ$ denotes function composition.
In \eqref{eqn:psiR}, $\left| \frac{\partial \R_x v}{\partial v} \right|$ denotes the absolute value of the Jacobian determinant of the map $\R_x$ at $v$.
The condition \eqref{eqn:psiR} is equivalent to the condition that
\begin{equation}
  \int_A \Pi(x,v) dv = \int_{\R_x(A)} \Pi(x,v) dv
\end{equation}
for every measurable subset $A$ of $\mathbb V$ and for almost every $x\sr\in\mathbb X$, due to the change of variable formula.
The condition \eqref{eqn:TSTS} can be understood as an abstraction of a property in Hamiltonian dynamics that if we reverse the velocity of a particle and advance in time, the particle traces back its past trajectory.

Given $X^{(i)}\sr =x$ and $V^{(i)}\sr=v$ at the start of the $i$-th iteration, a MCMC algorithm can make a deterministic proposal $\S_\tau(x,v)$, which is accepted with probability
\[
\min\left( 1, \frac{\Pi(\S_\tau(x,v))}{\Pi(x,v)} \left| \det D\S_\tau (x,v) \right| \right),
\]
where $D$ denotes the differential operator (i.e., $D\S_\tau(x,v) = \frac{\partial \S_\tau(x,v)}{\partial (x,v)}$).
In algorithms such as HMC or the BPS, the extended target density $\Pi(x,v)$ is often taken as a product of independent densities, $\bar\pi(x) \psi(v)$, where a common choice for $\psi(v)$ is a multivariate normal density.
The map $\S_\tau$ is often taken to preserve the reference measure, such that it has unit Jacobian determinant (i.e., $\left| \det D\S_\tau(x,v) \right| = 1$, for all $(x,v)$).

The sequential-proposal framework can be used to generalize MCMC algorithms using deterministic kernels in a similar way that it is applied to Metropolis-Hastings algorithms.
A pseudocode of a sequential-proposal MCMC algorithm using a deterministic kernel is shown in Algorithm~\ref{alg:sppd}.
Proposals are obtained sequentially as $(Y_n, W_n) \gets \S_\tau(Y_{n-1},W_{n-1})$, where we write $(Y_0,W_0) := (X^{(i)}, V^{(i)})$.
The pair $(Y_n,W_n)$ is deemed acceptable if
\[
\Lambda < \frac{\Pi(Y_n,W_n)}{\Pi(Y_0,W_0)} \left| \det  D\S_\tau^n(Y_0,W_0) \right|,
\]
where $\S_\tau^n = \S_\tau \circ \cdots \circ \S_\tau$ denotes a map obtained by composing $\S_\tau$ $n$ times.
If there are less than $L$ acceptable proposals in the sequence of $L$ proposals, the next state of the Markov chain is set to $(X^{(i+1)}, V^{(i+1)}) \gets (X^{(i)}, \R_{X^{(i)}}V^{(i)}).$
The velocity $V^{(i+1)}$ may be refreshed at the end of the iteration by drawing from $\psi(\,\cdot\giventh X^{(i+1)})$ with a certain probability $p^{\text{ref}}(X^{(i+1)})$ that may depend on $X^{(i+1)}$.
The parameter $\tau$ for the evolution map $\S_\tau$ can be drawn randomly.
The pseudocode in Algorithm~\ref{alg:sppd} shows the case where $\tau$ is drawn once per iteration and the same value is used for all $n\sr\in 1\col N$, but $\tau$ can also be drawn separately for each $n$, provided that the draws are independent of each other and of all other random draws in the algorithm.

We state the following result for Algorithm~\ref{alg:sppd}.
The proof is given in Appendix~\ref{sec:proof_SPPD_invariance}.
\begin{prop}\label{prop:SPPD_invariance}
  The extended target distribution with density $\Pi(x,v)$ is a stationary distribution for the Markov chain $\left( X^{(i)}, V^{(i)} \right)_{i\in1:M}$ constructed by Algorithm~\ref{alg:sppd}.
  Furthermore, the Markov chain $\left(X^{(i)}\right)_{i\in1:M}$ constructed by Algorithm~\ref{alg:sppd}, marginally for the $x$-component, is reversible with respect to the target distribution $\bar\pi(x)$.
\end{prop}

\section{Connection to Hamiltonian Monte Carlo methods}\label{sec:HMC}
\subsection{Sequential-proposal Hamiltonian Monte Carlo} \label{sec:spHMC}
In this section, we consider applications of the sequential-proposal approach described in Section~\ref{sec:sppdDescrip} to Hamiltonian Monte Carlo algorithms and discuss the numerical efficiency.
We first briefly summarize basic features of HMC algorithms.
A function on $\mathbb X \times \mathbb V$, called the Hamiltonian, is defined as the negative log density of the extended target density:
\begin{equation}
  H(x,v) := -\log \Pi(x,v) = -\log \bar\pi(x) - \log \psi(v\giventh x).\label{eqn:HamSys}
\end{equation}
We assume both $\mathbb X$ and $\mathbb V$ are equal to the $d$ dimensional Euclidean space $\mathbb R^d$.
The velocity distribution $\psi(v\giventh x)$ is often taken as a multivariate normal density independent of $x$,
\begin{equation*}
\psi(v\giventh x) \equiv \psi_C(v) := \frac{1}{\sqrt{(2\pi)^d \left| \det C \right|}} \exp\left\{ - \frac{v^T C^{-1} v}{2}\right\}.
\label{eqn:mvndensity}
\end{equation*}
An analogy with a physical Hamiltonian system is drawn by interpreting the first term $-\log \bar\pi(x)$ as the static potential energy of a particle and the second term $-\log \psi(v)$ as the kinetic energy.
In this analogy, the covariance matrix $C$ can be interpreted as the inverse of the mass of the particle.
Hamiltonian dynamics is defined as a solution to the Hamiltonian equation of motion (HEM):
\begin{equation}\begin{split}
  \frac{dx}{dt} &= C\frac{\partial H}{\partial v}\\
  \frac{dv}{dt} &= -C\frac{\partial H}{\partial x}.
  \end{split}\label{eqn:HEM}\end{equation}
If we denote the solution to the HEM as $(x(t), v(t))$, the exact Hamiltonian flow $S_\tau^*$ defined by $\S_\tau^*(x(0),v(0)) := (x(\tau),v(\tau))$ satisfies the reversibility condition \eqref{eqn:psiR} and \eqref{eqn:TSTS} when the velocity reflection operator is given by $\R_x(v) \sr= {-}v$ for all $x \sr\in \mathbb X$ and $v\sr\in\mathbb V$.
The map $\S_\tau^*$ preserves the Hamiltonian, that is, $H(x,v) = H(\S_\tau^*(x,v))$ for all $x\sr\in\mathbb X$, $v\sr\in\mathbb V$, and $\tau\sr\geq 0$.
The map $\S_\tau^*$ also preserves the reference measure $dx \, dv$, that is, $\left| \det D\S_\tau^*(x,v) \right|=1$ for all $x \in \mathbb X$, $v\in\mathbb V$ and $\tau\geq 0$, which is known as Liouville's theorem \citep{liouville1838note}.
\begin{figure}[t]
\centering
\scalebox{.89}{\begin{minipage}{\textwidth}
\begin{algorithm}[H]
\SetKwInOut{Input}{Input}\SetKwInOut{Output}{Output}
\Input{Leapfrog step size, $\epsilon$\\
  Number of leapfrog jumps, $l$\\
  Covariance of the velocity distribution, $C$}
  \vspace{1ex}
  \Output{A draw of Markov chain, $\left(X^{(i)}\right)_{i\in 1:M}$}
  \vspace{1ex}
  Run Algorithm~\ref{alg:sppd} with $\Pi(x,v) = \bar\pi(x) \psi_C(v)$, $p^\text{ref}(x) = 1$, $\tau := (\epsilon, l)$, $\S_\tau(x,v) = \texttt{Leapfrog}(x,v,\epsilon,l,C)$, and $\R_x = -\id$.

\vspace{1ex}
\SetKwProg{Fn}{Function}{}{end}
\Fn{\emph{\texttt{Leapfrog}($x,v,\epsilon,l,C$)}} {
  $v \gets v + \frac{\epsilon}{2}\cdot C \cdot \nabla \log\pi(x)$\\
  $x \gets x + \epsilon v$\\
  Set $j\gets 1$\\
  \While {$j < l$} {
    $v \gets v + \epsilon \cdot C\cdot \nabla \log \pi(x)$\\
    $x \gets x + \epsilon v$\\
    Set $j \gets j+1$
  }
  $v \gets v + \frac{\epsilon}{2} \cdot C \cdot \nabla \log \pi(x)$
}
\caption{Sequential-proposal HMC and leapfrog jump function}
\label{alg:leapfrog}
\end{algorithm}
\end{minipage}}
\end{figure}
A commonly used numerical approximation method for solving the HEM is called the leapfrog method \citep{duane1987hybrid, leimkuhler2004simulating}.
One iteration of the leapfrog method approximates time evolution of a Hamiltonian system for duration $\epsilon$ by alternately updating the velocity and position $(x,v)$ as follows:
\begin{equation}\begin{split}
    v &\gets v + \frac{\epsilon}{2} \cdot C \cdot \nabla \log \pi(x)\\
    x &\gets x + \epsilon v\\
    v &\gets v + \frac{\epsilon}{2} \cdot C \cdot  \nabla \log \pi(x). \label{eqn:leapfrog}
\end{split}\end{equation}
We call the time increment $\epsilon$ the leapfrog step size.

A standard Hamiltonian Monte Carlo algorithm is a specific instance of MCMC algorithms using deterministic kernels described in Section~\ref{sec:sppdDescrip}, where the extended target density $\Pi(x,v)$ is given by $\bar\pi(x)\psi_C(v)$ and the proposal map $\S_\tau$ is given by $l$ leapfrog jumps with step size $\epsilon$, such that the time duration parameter $\tau$ can be understood as the pair $(\epsilon, l)$.
The reversibility condition \eqref{eqn:RR}--\eqref{eqn:TSTS} is satisfied by this $\S_\tau$ with $\R_x \sr= {-}\id$ for all $x \sr\in \mathbb X$.
Each step in the leapfrog method \eqref{eqn:leapfrog} preserves the reference measure $dx\,dv$, so we have $|\det D \S_\tau | \equiv 1$.
It is common to refresh the probability at every iteration (i.e., $p^\text{ref}(x) \equiv 1$).

Sequential-proposal HMC (Algorithm~\ref{alg:leapfrog}) is obtained as a specific case of sequential-proposal MCMC algorithms using deterministic kernels (Algorithm~\ref{alg:sppd}) under the same setting, $\Pi(x,v)=\bar\pi(x)\psi_C(v)$ and $\S_\tau=\S_{(\epsilon,l)}$.
In other words, a proposal $(Y_1,W_1)$ is made by making $l$ leapfrog jumps of size $\epsilon$ starting from $(Y_0,W_0)$, and if the proposal is rejected, a new proposal $(Y_2,W_2)$ is made by making $l$ leapfrog jumps from $(Y_1, W_1)$.
The procedure is repeated until $L$ acceptable proposals are found, or until $N$ proposals have been tried, whichever comes sooner.
The leapfrog jump size $\epsilon$ and the unit number of jumps $l$ may be re-drawn at every iteration or for every new proposal.
As mentioned earlier, \citet{campos2015extra} has proposed extra chance generalized hybrid Monte Carlo (XCGHMC), which is identical to the sequential-proposal approach, except possibly in the way the velocity is refreshed at the end of each iteration.
In generalized HMC \citep{horowitz1991generalized}, the velocity is partially refreshed by setting
\[
V^{(i+1)} = \sin \theta \,V + \cos \theta \,U,
\]
where $V$ is the velocity before refreshement, $U$ is an independent draw from $\mathcal N(0,C)$, and $\theta$ is an arbitrary real number.
It was shown in \citet{campos2015extra} that Markov chains constructed by XCGHMC have the same law as those constructed by Look Ahead Hamiltonian Monte Carlo (LAHMC) developed by \citet{sohl2014hamiltonian}.

A major advantage of HMC algorithms over random walk based algorithms such as random walk Metropolis or Metropolis adjusted Langevin algorithms is that HMC can make a global jump in one iteration \citep{neal2011mcmc}.
The leapfrog method is able to build long trajectories that are numerically stable, provided that the target distribution satisfies some regulatory conditions and the leapfrog step size is less than a certain upper bound \citep{leimkuhler2004simulating}.
Since the solution to the HEM preserves the Hamiltonian, proposals obtained by a numerical approximation to the solution can be accepted with reasonably high probabilities.
Given a fixed length of leapfrog trajectory, the number of leapfrog jumps is inversely proportional to the leapfrog jump size.
Thus an increase in the leapfrog step size leads to a reduced number of evaluations of the gradient of the target density.
On the other hand, decreasing the leapfrog step size tends to increase the mean acceptance probability.
As $\epsilon\to 0$, the average increment in the Hamiltonian at the end of the leapfrog trajectory scales as $\epsilon^4$ \citep{leimkuhler2004simulating}.
In an asymptotic scenario where the target distribution is given by a product of $d$ independent, identical low dimensional distributions and $d$ tends to infinity, the increment in the Hamiltonian converges in distribution to a normal distribution with mean $\mu \epsilon^4 d$ and variance $2 \mu \epsilon^4 d$ for some constant $\mu\sr>0$ dependent on the target density $\bar\pi$ \citep{gupta1990acceptance, neal2011mcmc}.
\citet{beskos2013optimal} showed under some mild regulatory conditions on the target density that as $\epsilon = \epsilon_0 d^{-1/4}$ and $d\to\infty$, the mean acceptance probability tends to $a(\epsilon_0) := 2\Phi\left( - \epsilon_0^2 \sqrt{\mu/2} \right)$ where $\Phi(\cdot)$ denotes the cdf of the standard normal distribution.
The computational cost for obtaining an accepted proposal that is fixed distance away from the current state in HMC is approximately given by
\[
\frac{1}{\epsilon_0 a(\epsilon_0)},
\]
which is minimized when $a(\epsilon_0) = 0.651$ to three decimal places \citep{beskos2013optimal, neal2011mcmc}.
Empirical results also support targeting the mean acceptance probability of around 0.65 \citep{sexton1992hamiltonian, neal1994improved}.
HMC using sequential proposals can improve on the numerical efficiency by increasing the probability that the constructed Markov chain makes a nonzero move at each iteration.
A numerical study in Section~\ref{sec:HMCnumerical} shows that HMC with sequential proposals leads to higher effective sample sizes per computation time compared to the standard HMC on a toy model.

\subsection{Sequential-proposal No-U-Turn sampler algorithms}\label{sec:spNUTS}
\begin{figure}[t!]
\centering
\scalebox{.89}{\begin{minipage}{\textwidth}
\begin{algorithm}[H]
\SetKwInOut{Input}{Input}\SetKwInOut{Output}{Output}
\Input{Leapfrog step size, $\epsilon$}
\vspace{1ex}
\Output{A draw of Markov chain, $\left(X^{(i)}\right)_{i\in 1:M}$}
\vspace{1ex}
\textbf{Initialize:} Set $X^{(0)}$ arbitrarily

\For {$i\gets 0\col M{-}1$}{
  Draw $\Lambda\sim \text{unif}(0,1)$ and $V\sim \mathcal N(0,I_d)$\\
  Start with an initial tree $T^0 :=\{(X^{(i)}, V)\}$ having a single leaf\\
  \For {$j \geq 1$}{
    Draw $\sigma_j \sim \text{unif}(\{-1,1\})$\\
    Make $2^{j-1}$ leapfrog jumps either forward or backward depending on $\sigma_j$, forming a new binary tree $T'$ of the same size as $T^{j-1}$.\\
    \uIf {every sub-binary trees of $T'$ is such that the two leaves on the opposite sides do not satisfy the U-turn condition \eqref{eqn:U_cond_NUTS}} {
      Set $T^j \gets T^{j-1} \cup T'$
    }
    \Else {
      \texttt{break}
    }
    \If {the two opposite leaves of $T^j$ satisfies the U-turn condition} {
      \texttt{break}
    }
  }
  Let $T^{j_0}$ be the final binary tree constructed\\
  \emph{Naive NUTS (Algorithm~2 in \citet{hoffman2014no})}: Take for $X^{(i+1)}$ one of the leaf nodes of $T^{j_0}$ that are acceptable, i.e.,  $\frac{\Pi(x,v)}{\Pi(X^{(i)},V)} > \Lambda$, uniformly at random\\
  \emph{Efficient NUTS (Algorithm~3 in \citet{hoffman2014no})}: Denote by $n_a(T)$ the number of acceptable leaf nodes in a binary tree $T$, and \\
  \For {$j \gets j_0\col 0$} {
    With probability $1 \land \frac{n_a(T^j\sr\setminus T^{j-1})}{n_a(T^{j-1})}$, take for $X^{(i+1)}$ one of the acceptable leaf nodes of $T^j\sr\setminus T^{j-1}$ uniformly at random, and break out from \textbf{for} loop
  }
}
\caption{The No-U-Turn samplers by \citet{hoffman2014no}}
\label{alg:NUTS}
\end{algorithm}
\end{minipage}}
\end{figure}

\begin{figure}
  \centering
  \includegraphics[width=.55\textwidth]{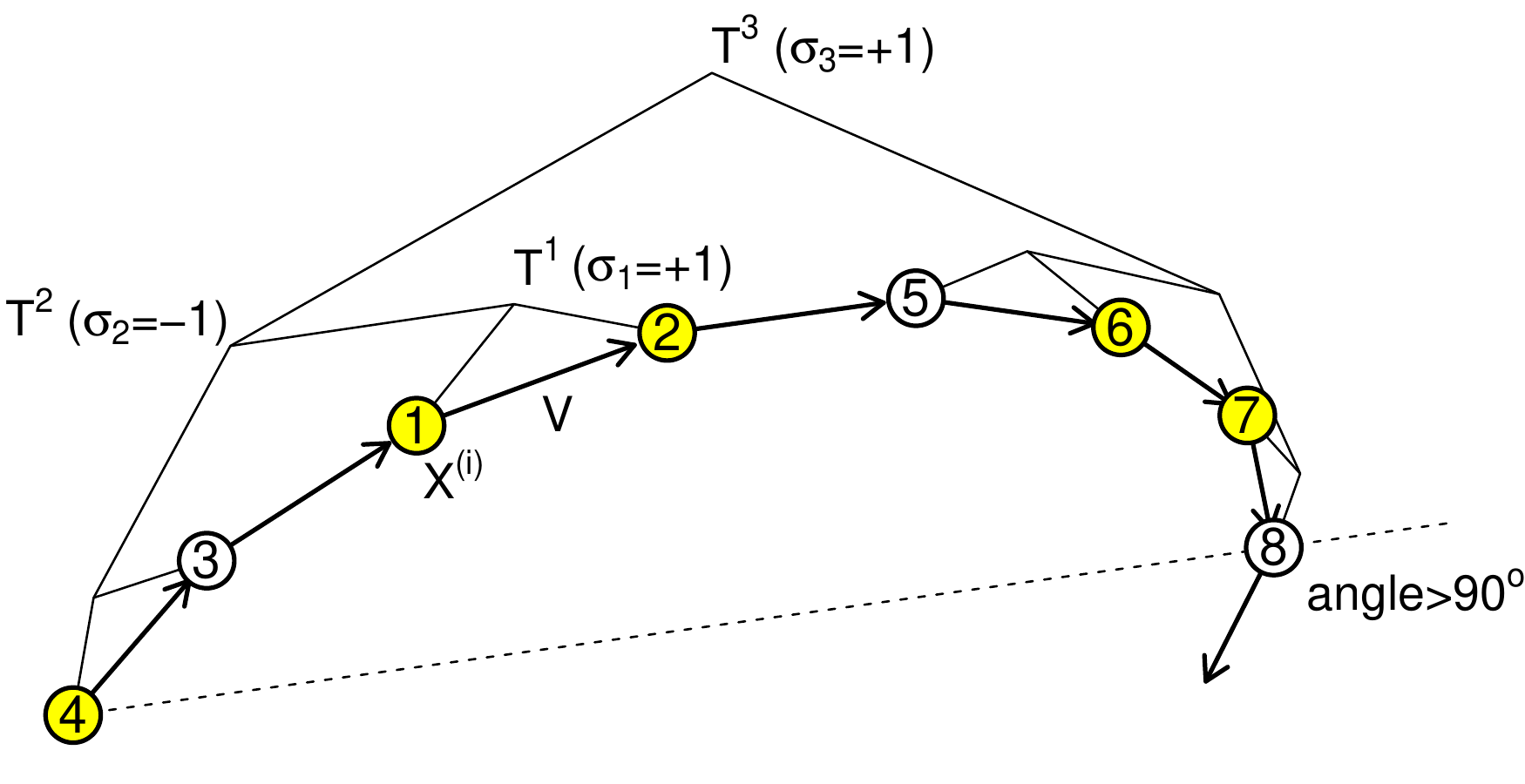}
  \caption{An example diagram of a final binary tree constructed in an iteration of the NUTS algorithm by \citet{hoffman2014no}. The numbered circles indicate the points along a leapfrog trajectory in the order they are added. The binary tree stops expanding at $T^3$ because there is a U-turn between leaf nodes 4 and 8. The next state of the Markov chain is selected randomly among the acceptable states, colored in yellow.}
  \label{fig:NUTSdiagram}
\end{figure}

As previously mentioned, a key advantage of HMC over random walk based methods comes from its ability to make long moves.
If the number of leapfrog jumps is too small, the Markov chain from HMC may essentially behave like a random walk because the velocity is randomly refreshed before a long leapfrog trajectory is built.
Conversely, if the number of leapfrog jumps is too large, the trajectory may double back on itself, since the solution to the Hamiltonian equation of motion is confined to a level set of the Hamiltonian.
However, simply stopping the leapfrog jumps when the trajectory starts doubling back on itself generally destroys the detailed balance of the Markov chain with respect to the target distribution.
In order to solve this issue, \citet{hoffman2014no} proposed the No-U-Turn sampler (NUTS).
In this section, we will briefly explain the NUTS algorithm and discuss the connection with sequential-proposal framework.
In addition, we will propose two new algorithms that address the same issue of trajectory doubling.

In the No-U-Turn sampler, leapfrog trajectories are repeatedly extended twice in size in either forward or backward direction in the form of binary trees, until a ``U-turn'' is observed (see Figure~\ref{fig:NUTSdiagram}).
The binary tree starts from the initial node $(X^{(i)},V)$, where $V$ is the velocity drawn from the standard multivariate normal distribution at the beginning of the $i$-th iteration.
The direction of binary tree expansion is determined by a sequence of $\text{unif}(\{-1,1\})$ variables denoted by $(\sigma_j)_{j\geq 0}$.
The expansion of the binary tree stops if a U-turn is observed between the two leaf nodes on opposite sides of any of the sub-binary trees of the current tree.
A position-velocity pair $(x,v)$ and another pair $(x',v')$ that is ahead of $(x,v)$ on a leapfrog trajectory are said to satisfy the U-turn condition if either
\begin{equation}
(x'\sr-x) \cdot v' \leq 0 \quad \text{or} \quad (x'\sr-x) \cdot v \leq 0,
\label{eqn:U_cond_NUTS}
\end{equation}
where $\cdot$ denotes the inner product in Euclidean spaces.
If there is a U-turn within the lastly added half of the current binary tree, the other half without a U-turn is taken as the final binary tree.
On the other hand, if a U-turn is only observed between the two opposite leaf nodes of the current binary tree but not within any of the sub-binary trees, the current binary tree is taken as the final binary tree.
The next state of the Markov chain $X^{(i+1)}$ is set to one of the acceptable leaf nodes in the final binary tree.
A leaf node $(x,v)$ is deemed acceptable if $\frac{\Pi(x,v)}{\Pi(X^{(i)},V)} > \Lambda$.
Here $\Pi(x,v) := \bar\pi(x) \psi_{I_d}(v)$, where $\psi_{I_d}$ denotes the density of the $d$ dimensional standard normal distribution.
\citet{hoffman2014no} gives two versions of the NUTS algorithm.
The naive version selects the next state of the Markov chain uniformly at random among the acceptable leaf nodes in the final binary tree.
The efficient version preferentially selects a random leaf node in sub-binary trees that are added later.
A pseudocode for these NUTS algorithms is given in Algorithm~\ref{alg:NUTS}.
By construction, for every leaf node in the final binary tree, it is possible to build the same final binary tree starting from that leaf node using a unique sequence of directions.
Since each direction is drawn from $\text{unif}(\{-1,1\})$, the probability of constructing the final binary tree is the same when started from any of its leaf nodes.
This symmetric relationship ensures that the constructed Markov chain is reversible with respect to the target distribution.

The NUTS algorithm shares with the sequential-proposal MCMC framework the key feature that the decisions of acceptance or rejection of proposals are mutually coupled via a single uniform$(0,1)$ random variable drawn at the start of each iteration.
Furthermore, the naive version of the algorithm (Algorithm~2 in \citet{hoffman2014no}) can be viewed as a specific case of the sequential-proposal MCMC algorithm as follows.
At each iteration a binary tree starting from $(X^{(i)},V)$ is expanded until a U-turn is observed, as described above.
Proposals are made sequentially by selecting one of the leaf nodes of the final binary tree uniformly at random.
The first proposal that is acceptable is taken as the next state of the Markov chain.
Since the next state of the Markov chain is then selected uniformly at random among the \emph{acceptable} leaf nodes in the final binary tree, this sequential-proposal approach is equivalent to the naive NUTS.

There are two features of the NUTS algorithm that may, unfortunately, compromise the numerical efficiency.
First, the point chosen for the next state of the Markov chain is generally not the farthest point on the leapfrog trajectory from the initial point.
The NUTS typically constructs a leapfrog trajectory that is longer than the distance between the initial point and the point selected for the next state of the Markov chain due to the requirement of detailed balance.
Second, the NUTS evaluates the log target density at every point on the constructed leapfrog trajectory to determine the acceptability.
This can result in a substantial overhead if the computational cost of evaluating the log target density is at least comparable to that of evaluating the gradient of the log target density.
We propose two alternative No-U-Turn sampling algorithms, which we call spNUTS1 and spNUTS2, addressing these two issues.

\newcommand\cosangle{\text{cosAngle}}
\begin{figure}[t!]
\centering
\scalebox{.89}{\begin{minipage}{\textwidth}
\begin{algorithm}[H]
\SetKwInOut{Input}{Input}\SetKwInOut{Output}{Output}
\Input{Leapfrog step size, $\epsilon$\\
  Unit number of leapfrog jumps, $l$\\
  Covariance of velocity distribution, $C$\\
  Scheduled checkpoints for a U-turn, $(b_j)_{j\in 1\col j_{\max}}$\\
  Distribution for the stopping value of cosine angle, $\zeta$ \\
  Maximum number of proposals tried, $N$}
\vspace{1ex}
\Output{A draw of Markov chain, $\left(X^{(i)}\right)_{i\in 1:M}$}
\vspace{1ex}
\textbf{Initialize:} Set $X^{(0)}$ arbitrarily

\For {$i\gets 0\col M{-}1$}{
  Set $Y_0 \gets X^{(i)}$ and draw $W_0 \sim \mathcal N(0,C)$\\
  Set $X^{(i+1)} \gets X^{(i)}$\\
  Draw $\Lambda\sim \text{unif}(0,1)$ and set $H_\text{max} \gets -\log\pi(Y_0) + \frac{1}{2}\Vert W_0\Vert_C^2 - \log \Lambda$\\
  \For {$n \gets 1\col N$} {
    Draw $c \sim \zeta(\cdot)$\\
    $(Y_n, W'_n) \gets \texttt{spNUTS1Kernel}(Y_{n-1}, W_{n-1}, c$)\\
    \If {$-\log\pi(Y_n) + \frac{1}{2}\Vert W'_n\Vert_C^2 < H_\text{max}$} {
      Set $X^{(i+1)} \gets Y_n$\\
      \texttt{break}
   }
    Draw $U\sim \mathcal N(0,C)$ and set $W_n \gets U \cdot \frac{\Vert W'_n \Vert_C}{\Vert U \Vert_C}$
  }
}

\vspace{1ex}
\SetKwProg{Fn}{Function}{}{end}
\Fn{\emph{\texttt{spNUTS1Kernel}($x_0, v_0, c$)}}{
  \For {$k \gets 1\col b_1$} {
    $(x_k,v_k) \gets \texttt{Leapfrog}(x_{k-1},v_{k-1},\epsilon,l,C)$
  }
  Set $j\gets 1$\\
  \While {$\cosangle(x_{b_j}{-}x_0, v_0\giventh C) > c$ and $\cosangle(x_{b_j}{-}x_0, v_{b_j}\giventh C) >c$ and $j<j_\text{max}$} {
    Set $j\gets j+1$\\
    \For {$k \gets b_{j-1}{+}1\col b_j$} {
      $(x_k,v_k) \gets \texttt{Leapfrog}(x_{k-1},v_{k-1},\epsilon,l,C)$
    }
  }
  \uIf {$\cosangle (x_{b_j}{-}x_{b_j-b_{j'}}, v_{b_j} \giventh C) >c$ and $\cosangle(x_{b_j}{-}x_{b_j-b_{j'}}, v_{b_j-b_{j'}} \giventh C) >c$ for all $j' \in 1\col j{-}1$} {
    \texttt{return} $(x_{b_j},v_{b_j})$
  }
  \Else {
    \texttt{return} $(x_0,v_0)$
  }
}
\caption{Sequential-proposal No-U-Turn sampler---Type 1 (spNUTS1)}
\label{alg:spNUTS1}
\end{algorithm}
\end{minipage}}
\end{figure}

In spNUTS1, leapfrog trajectories are extended in one direction according to a given length schedule until a U-turn is observed, and only the endpoint of the trajectory is checked for acceptability.
If the endpoint is not acceptable, a new trajectory is started from that point with a refreshed velocity.
A pseudocode of spNUTS1 is given in Algorithm~\ref{alg:spNUTS1}.
At the start of each iteration, a velocity vector is drawn from a multivariate normal distribution $\mathcal N(0,C)$ where $C$ is a $d\sr\times d$ positive definite matrix.
A leapfrog trajectory started from the current state of the Markov chain and the drawn velocity vector, denoted by $(x_0,v_0)$, is repeatedly extended in units of $l$ jumps.
The position-velocity pair after $lk$ leapfrog jumps is denoted by $(x_k,v_k)$.
We note that the leapfrog updates \eqref{eqn:leapfrog} should also use the same matrix $C$ as the covariance of the velocity distribution.
At preset checkpoints determined by a finite increasing sequence $(b_j)_{j\in 1:j_{\max}}$, the algorithm calculates the angles between the displacement $x_{b_j}{-}x_0$ and the velocities $v_0$ and $v_{b_j}$.
In order to take into account the given covariance structure $C$, we define a $C$-norm of a vector $x \in \mathbb R^d$ as
\[
\Vert x \Vert_C := \sqrt{x^T C^{-1} x},
\]
and the cosine of the angle between two vectors $x$ and $x'$ as
\begin{equation}
  \cosangle(x,x'\giventh C) := \frac{x^TC^{-1}x'}{\Vert x\Vert_C \cdot \Vert x'\Vert_C}.
  \label{eqn:cosangle}
\end{equation}
The leapfrog trajectory stops at $(x_{b_j},v_{b_j})$ if either of the following inequalities hold for a given $c$:
\begin{equation}
  \cosangle(x_{b_j}{-}x_0,v_0\giventh C) \leq c \quad \text{ or } \quad \cosangle(x_{b_j}{-}x_0, v_{b_j} \giventh C)\leq c.
  \label{eqn:stopping_cond_spNUTS1}
\end{equation}
The value of $c$ can be fixed at a constant value or randomly drawn for each trajectory.
Algorithm~\ref{alg:spNUTS1} describes a case where $c$ is randomly drawn from a distribution denoted by $\zeta$.
If the above stopping condition \eqref{eqn:stopping_cond_spNUTS1} is not satisfied until $j\,{=}\,j_{\max}$, the trajectory stops at $(x_{b_{j_{\max}}},v_{b_{j_{\max}}})$.
\begin{figure}
  \centering
  \includegraphics[width=.8\textwidth]{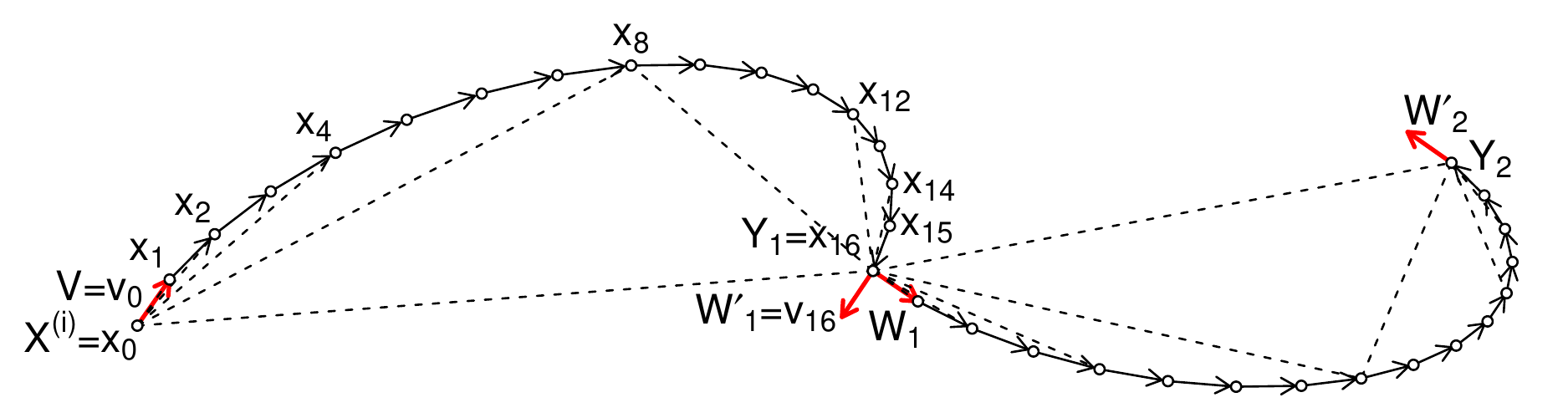}
  \caption{An example diagram for an iteration in spNUTS1 where $b_j\sr=2^{j-1}$. The first proposal $(Y_1,W'_1) \sr= (x_{16},v_{16})$ was rejected, and the second trajectory was started with a refreshed velocity $W_1$. The pairs of points for which the U-turn condition is checked are connected by dashed line segments.}
  \label{fig:spNUTS1diagram}
\end{figure}
The final state at the stopped trajectory $(x_{b_j},v_{b_j})$ makes the first proposal $(Y_1,W'_1)$.
It is taken as the next state in the Markov chain if the following two conditions are met.
First, the state $(x_{b_j},v_{b_j})$ has to be acceptable by satisfying
\begin{equation}
  \log\Lambda < \log\pi(x_{b_j})+\log\psi_C(v_{b_j})-\log\pi(x_0)-\log\psi_C(v_0).
  \label{eqn:acceptability_cond_spNUTS1}
\end{equation}
Since the Hamiltonian of the state $(x_{b_j},v_{b_j})$ is given by $-\bar\pi(x_{b_j}){-}\log\psi_C(v_{b_j})$, the acceptability criterion \eqref{eqn:acceptability_cond_spNUTS1} can be interpreted as that the increase in the Hamiltonian compared to the initial state $(x_0,v_0)$ is at most $-\log\Lambda$.
The second required condition is that 
\begin{equation}
  \cosangle(x_{b_j}\sr-x_{b_j-b_{j'}}, v_{b_j}\giventh C) > c ~ \text{ and } ~ \cosangle(x_{b_j}\sr-x_{b_j-b_{j'}}, v_{b_j-b_{j'}}\giventh C) > c ~~ \text{ for all }1\sr\leq j' \sr\leq j{-}1.
  \label{eqn:symmetry_cond_spNUTS1}
\end{equation}
Since the trajectory has been extended to $(x_{b_j}, v_{b_j})$, the stopping condition \eqref{eqn:stopping_cond_spNUTS1} was not satisfied between the initial state $(x_0,v_0)$ and any of the previously visited states $\{(x_{b_{j'}},v_{b_{j'}})\giventh 1 \sr\leq j' \sr< j\}$.
When the trajectory is viewed in the reverse order, an analogous situation is that the final state $(x_{b_j}, v_{b_j})$ and the intermediate states $\{(x_{b_j-b_{j'}}, v_{b_j-b_{j'}}) \giventh 1\sr\leq j' \sr< j\}$ satisfy \eqref{eqn:symmetry_cond_spNUTS1}.
This symmetry condition is necessary to establish detailed balance of the Markov chain.
If the symmetry condition is not satisfied, the next state of the Markov chain is set to the current state $x_0$.
In the case where both the acceptability condition \eqref{eqn:acceptability_cond_spNUTS1} and the symmetry condition \eqref{eqn:symmetry_cond_spNUTS1} are satisfied, $x_{b_j}$ is taken as the next state of the Markov chain.
If the symmetry condition is satisfied but the acceptability condition is not, the algorithm starts a new leapfrog trajectory from $x_{b_j}$ with a new velocity $W_1$.
The new velocity $W_1$ is obtained by drawing a random vector from the velocity distribution $\psi_C$ and rescaling it such that the $C$-norm is preserved: $\Vert W_1 \Vert_C = \Vert W'_1 \Vert_C$, or $\log\psi_C(W_1) = \log\psi_C(W'_1)$.
The new trajectory is extended until the stopping condition \eqref{eqn:stopping_cond_spNUTS1} with respect to the starting pair $(Y_1,W_1)$ is satisfied.
This procedure is repeated for subsequent proposals $(Y_n,W'_n)$, $n\sr\geq 2$.
If all $N$ proposals $\{(Y_n,W'_n) \giventh n \sr \in 1\col N\}$ are unacceptable, the next state of the Markov chain is set to the current state $X^{(i)}$.
The Markov chain also stays at its current state if the symmetry condition is not satisfied by at least one of the constructed leapfrog trajectories.

Algorithm~\ref{alg:spNUTS1} can be considered as a special case of sequential-proposal MCMC algorithms using deterministic kernels (Algorithm~\ref{alg:sppd}) where $L\sr=1$ and the proposal map $\S_\tau$ is given by a function that maps the starting state $(x_0,v_0)$ to the final state $(x_{b_j},v_{b_j})$ of the stopped trajectory, but with a differing feature that the direction of the velocity is randomly refreshed between proposals.
In practice, extending the trajectory length at an exponential rate by setting, for example, $b_j \sr= 2^{j-1}$ can lead to a high probability that the symmetry condition is satisfied.
\citet{hoffman2014no} also remarked (based on Figure~5 in their paper) that as the size of binary trees were repeatedly doubled, the U-turn condition was satisfied most of the time only by the two opposite leaf nodes of the final binary tree but not by the opposite leaf nodes of any of the sub-binary trees, for all examples they considered including ones arising from practical applications.
The choice $b_j\sr= 2^{j-1}$ also makes it easy to predict the checkpoints for the symmetry condition, $\{b_j\sr-b_{j'}\giventh j'\sr\leq j{-}1\}$.
We note that the numerical efficiency of Algorithm~\ref{alg:spNUTS1}, as well as that of the original NUTS algorithm, can be improved by tuning the covariance $C$ (see the numerical results in Section~\ref{sec:HMCnumerical}).

\renewcommand{\topfraction}{.85}
\begin{figure}[t!]
\centering
\scalebox{.89}{\begin{minipage}{\textwidth}
\begin{algorithm}[H]
\SetKwInOut{Input}{Input}\SetKwInOut{Output}{Output}
\Input{Leapfrog step size, $\epsilon$\\
  Unit number of leapfrog jumps, $l$\\
  Covariance of velocity distribution, $C$\\
  Maximum number of proposals, $N$\\
  Scheduled checkpoints for a U-turn, $(b_j)_{j\in 1\col j_{\max}}$\\
  Distribution for the stopping value of cosine angle, $\zeta$
}
\vspace{1ex}
\Output{A draw of Markov chain, $\left(X^{(i)}\right)_{i\in 1:M}$}
\vspace{1ex}
\textbf{Initialize:} Set $X^{(0)}$ arbitrarily

\For {$i\gets 0\col M{-}1$}{
  Draw $V \sim \mathcal N(0,C)$\\
  Draw $c \sim \zeta(\cdot)$\\
  Draw $\Lambda \sim \text{unif}(0,1)$ and set $\Delta \gets - \log \Lambda$\\
  Set $\left(X^{(i+1)}, V^{(i+1)}\right) \gets \texttt{spNUTS2Kernel}( X^{(i)}, V, \Delta, \epsilon, C, c)$\\
}

\vspace{1ex}
\SetKwProg{Fn}{Function}{}{end}
\Fn{\emph{\texttt{spNUTS2Kernel}($x_0, v_0, \Delta, \epsilon, C, c$)}} {
  Set $H_\text{max} \gets -\log \pi(x_0) + \frac{1}{2} \Vert v_0 \Vert_C^2 + \Delta$ \\
  \For {$k \gets 1\col b_1$} {
    $(x_k, v_k, \mathbf f) \gets \texttt{FindNextAcceptable}(x_{k-1}, v_{k-1}, \epsilon, H_\text{max}, C)$\\
    \textbf {if} $\mathbf f=0$ \textbf{then} \texttt{return} $(x_0,v_0)$  \qquad // the case where no acceptable states were found \\
  }
  Set $j\gets 1$\\
  \While {$\cosangle(x_{b_j}{-}x_0, v_0 \giventh C) > c$ and $\cosangle(x_{b_j}{-}x_0, v_{b_j} \giventh C) > c$ and $j< j_\text{max}$}{
    Set $j \gets j\,{+}\,1$\\
    \For {$k \gets b_{j-1}{+}1\col b_j$} {
      $(x_k, v_k, \mathbf f) \gets \texttt{FindNextAcceptable}(x_{k-1}, v_{k-1}, \epsilon, H_\text{max}, C)$\\
      \textbf {if} $\mathbf f=0$ \textbf{then} \texttt{return} $(x_0,v_0)$\\
    }
  }
  \uIf {$\cosangle(x_{b_j}{-}x_{b_j-b_{j'}}, v_{b_j}\giventh C) >c$ and $\cosangle(x_{b_j}{-}x_{b_j-b_{j'}}, v_{b_j-b_{j'}}\giventh C) > c$ for all $j' \in 1\col j{-}1$} {
    \texttt{return} $(x_{b_j},v_{b_j})$
  }
  \Else {
    \texttt{return} $(x_0, v_0)$
  }
}

\vspace{1ex}
\Fn{\emph{\texttt{FindNextAcceptable}($x, v, \epsilon, H_\text{max}, C$)}} {
  Set $(x_\text{try}, v_\text{try}) \gets (x,v)$\\
  \For {$n \gets 1\col N$} {
    $(x_\text{try},v_\text{try}) \gets \texttt{Leapfrog}(x_\text{try},v_\text{try},\epsilon, l,C)$\\
    \textbf{if} {$-\log \pi(x_\text{try}) + \frac{1}{2}\Vert v_\text{try}\Vert_C^2 < H_\text{max}$} \textbf{then} \texttt{return} $(x_\text{try},v_\text{try}, 1)$
  }
  \texttt{return} $(x, v, 0)$
}
\caption{Sequential-proposal No-U-Turn sampler---Type 2 (spNUTS2)}
\label{alg:spNUTS2}
\end{algorithm}
\end{minipage}}
\end{figure}

The computational efficiency of Algorithm~\ref{alg:spNUTS1} may be compared favorably to that of the NUTS algorithm.
Some numerical results are given in Section~\ref{sec:HMCnumerical}.
Suppose that the average computational cost of evaluating the log target density is denoted by $c_\pi$ and that of evaluating the gradient of the log target density by $c_\triangledown$.
Assume also that in the NUTS and spNUTS1 (Algorithm~\ref{alg:spNUTS1}), the computational cost of checking the U-turn condition such as \eqref{eqn:stopping_cond_spNUTS1} is denoted by $c_U$.
If a leapfrog trajectory stops after making $\xi$ sets of $l$ unit jumps on average, the NUTS algorithm evaluates the log target density $\xi$ times and checks the U-turn condition $\xi$ times on average.
Thus the average computational cost for one iteration of the NUTS algorithm is given by $(lc_\triangledown+c_\pi+c_U) \cdot \xi$.
In comparison, spNUTS1 evaluates the log target density once and checks the U-turn condition $2\log_2 \xi\sr+ 1$ times if $b_j \sr= 2^{j-1}$ for $j\sr\in 1\col j_{\max}$.
The average computational cost of obtaining a proposal in spNUTS1 is given by $lc_{\triangledown} \xi + c_\pi + c_U(2\log_2\xi \sr+1)$, and the average cost of finding a new state for the Markov chain different from the current state is roughly given by $\frac{1}{a \cdot \tilde a} \big(lc_{\triangledown} \xi \sr+ c_\pi \sr+ c_U(2\log_2\xi\sr+1)\big)$, where $a$ denotes the mean acceptance probability of a proposal and $\tilde a$ denotes the average probability that the symmetry condition is satisfied.
Both $a$ and $\tilde a$ can be made close to unity in practice, so there is a computational gain in using spNUTS1 over the original NUTS if $\xi$ is large and $c_\pi$ is at least comparable to $c_\triangledown$.
The number $l$ can be chosen to one unless there is an issue of numerical instability of leapfrog trajectories.
We note that the increase in the cost by a factor of $\frac{1}{a}$ can be partially negated, in terms of the overall numerical efficiency, due to the fact if a proposal is deemed unacceptable, the next proposal can be further away from the initial state $Y_0$.
The average distance between two consecutive states in the constructed Markov chain is a measure widely used to evaluate the numerical efficiency of a MCMC algorithm \citep{sherlock2010random}.

The proof of the following proposition is given in Appendix~\ref{sec:proof_spNUTS}.
\begin{prop}  \label{prop:detailedbalance_spNUTS1}
  The Markov chain $\left(X^{(i)}\right)_{i\in1:M}$ constructed by the sequential-proposal No-U-Turn sampler of type 1 (spNUTS1, Algorithm~\ref{alg:spNUTS1}) is reversible with respect to the target density $\bar\pi$.
\end{prop}

\begin{figure}
  \centering
  \includegraphics[width=0.6\textwidth]{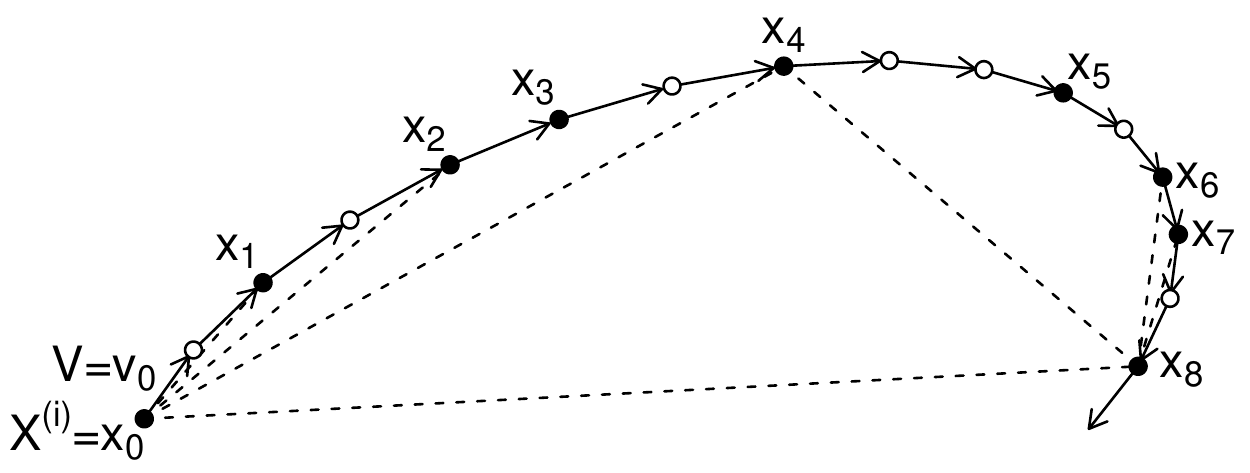}
  \caption{An example diagram for an iteration in spNUTS2 where $b_j \sr= 2^{j-1}$. Acceptable states are marked by filled circles and unacceptable ones by empty circles. The pairs of states for which the U-turn condition is checked are indicated by dashed line segments. The eighth acceptable state $x_8$ is taken as the next state of the Markov chain.}
  \label{fig:spNUTS2diagram}
\end{figure}
Another algorithm that automatically tunes the lengths of leapfrog trajectories, called spNUTS2, is given in Algorithm~\ref{alg:spNUTS2}.
Unlike spNUTS1, spNUTS2 applies the sequential-proposal scheme within one trajectory.
The spNUTS2 algorithm takes the endpoint of the constructed leapfrog trajectory as a candidate for the next state of the Markov chain, as in spNUTS1.
However, it evaluates the log target density at every point on the trajectory like the original NUTS.
Starting from the current state of the Markov chain $X^{(i)} \sr= x_0$ and a velocity vector $v_0$ randomly drawn from $\psi_C$, the algorithm extends a leapfrog trajectory in units of $l$ leapfrog jumps.
We will denote by $(x_1,v_1)$ the first \emph{acceptable} state along the trajectory that is a multiple of $l$ leapfrog jumps away from the initial state.
Here, $(x_1,v_1)$ is acceptable if 
\begin{equation*}
  \Lambda < \frac{\pi(x_1)\psi_C(v_1)}{\pi(x_0)\psi_C(v_0)}.
  \label{eqn:acceptability_cond_spNUTS2}
\end{equation*}
In order to avoid indefinitely extending the trajectory when the leapfrog approximation is numerically unstable, the algorithm ends the attempt to find the next acceptable state if $N$ consecutive states at intervals of $l$ leapfrog jumps are all unacceptable.
In this case, the next state of the Markov chain is set to $(x_0,v_0)$.
For $k\,{\geq}\, 2$, the state $(x_k,v_k)$ is likewise found as the first acceptable state along the leapfrog trajectory that is a multiple of $l$ jumps from $(x_{k-1},v_{k-1})$.
If for any $k\,{\geq}\,1$ the next acceptable state is not found in $N$ consecutive states visited after $(x_{k-1},v_{k-1})$, the next state in the Markov chain is also set to $(x_0,v_0)$.
In practice, however, this situation can be avoided by taking the leapfrog step size $\epsilon$ reasonably small to ensure numerical stability and $N$ large enough.
The algorithm takes a preset increasing sequence of integers $(b_j)_{j\in 1 \col j_{\max}}$ and checks if the angles between the displacement vector $x_{b_j}\sr-x_0$ and the initial and the last velocity vectors $v_0$ and $v_{b_j}$ are below a certain level $c$.
The trajectory is stopped at $(x_{b_j},v_{b_j})$ if either
\begin{equation}
\cosangle(x_{b_j}\sr-x_0,v_0\giventh C) \leq c \quad \text{or} \quad \cosangle(x_{b_j}\sr-x_0,v_{b_j}\giventh C) \leq c.
\label{eqn:stopping_cond_spNUTS2}
\end{equation}
Upon reaching $(x_{b_{j_{\max}}},v_{b_{j_{\max}}})$, however, the trajectory stops regardless of whether \eqref{eqn:stopping_cond_spNUTS2} is satisfied for $j\,{=}\,j_{\max}$.
As in spNUTS1, a symmetry condition is checked to ensure detailed balance.
That is, the state $(x_{b_j},v_{b_j})$ is taken as the next state in the Markov chain if and only if
\begin{equation}
  \cosangle(x_{b_j} \sr-x_{b_j-b_{j'}}, v_{b_j}\giventh C) > c ~ \text{ and } ~ \cosangle(x_{b_j}\sr- x_{b_j-b_{j'}}, v_{b_j-b_{j'}}\giventh C) > c
  ~\text{ for all } 1\,{\leq}\, j' \,{\leq}\, j{-}1.
  \label{eqn:symmetry_cond_spNUTS2}
\end{equation}
If the symmetry condition is not satisfied, the next state of the Markov chain is set to $(x_0,v_0)$.
As in spNUTS1, the choice of $b_j \,{=}\, 2^{j-1}$, $j\in 1\col j_{\max}$, allows the symmetry condition in spNUTS2 to be satisfied with high probability and makes the checkpoints for the symmetry condition, $\{b_j\sr-b_{j'}\giventh j'\sr\leq j{-}1\}$, readily predictable.

When $c_\pi$, $c_\triangledown$, and $c_U$ denote the same average computational costs as before and $b_j \sr= 2^{j-1}$, the average computational cost of finding a distinct sample point for the Markov chain using spNUTS2 is roughly given by $\frac{1}{\tilde a}\{(lc_\triangledown + c_\pi) \cdot \xi + c_U (2\log_2(a\xi)\sr+1)\}$, where $\xi$ denotes the average length of stopped trajectories in units of $l$ leapfrog jumps, $\tilde a$ the average probability that the symmetry condition is satisfied, and $a$ the mean acceptance probability.
Since $\tilde a$ can be close to unity and $c_U$ is often smaller than $c_\pi$ or $c_\triangledown$ in practice, the computational cost of spNUTS2 per distinct sample is comparable to that of the NUTS.
However, the overall numerical efficiency of spNUTS2 can be higher because the average distance between the current and the next state of the Markov chain can be larger.

The proof of the following proposition is also given in Appendix~\ref{sec:proof_spNUTS}.
\begin{prop}  \label{prop:detailedbalance_spNUTS2}
  The Markov chain $\left(X^{(i)}\right)_{i\in1:M}$ constructed by the sequential-proposal No-U-Turn sampler of type 2 (spNUTS2, Algorithm~\ref{alg:spNUTS2}) is reversible with respect to the target density $\bar\pi$.
\end{prop}

\subsection{Adaptive tuning of parameters in HMC}\label{sec:adaptiveHMC}
Adaptively tuning parameters in MCMC algorithms using the history of the Markov chain can often lead to enhanced numerical efficiency \citep{haario2001adaptive, andrieu2008tutorial}.
Here we discuss adaptive tuning of some parameters in HMC algorithms.
As discussed in Section~\ref{sec:spHMC}, tuning the leapfrog step size $\epsilon$ is one of the critical decisions to make in running HMC.
Numerical efficiency of HMC algorithms can be increased by targeting an average acceptance probability that is away from both zero and one \citep{beskos2013optimal}.
Since the mean acceptance probability tends to increase with decreasing step size, we use the following recursive formula to update the step size,
\begin{equation}
  \log\epsilon_{i+1} \gets \log\epsilon_i + \frac{\lambda}{i^\alpha} (a_i - a^*),
  \label{eqn:eps_adaptive_HMC}
\end{equation}
where $\epsilon_i$ and $a_i$ denote the leapfrog step size and the acceptance probability of a proposal at the $i$-th iteration, and $a^*$ the target mean acceptance probability.
We follow a standard approach for the sequence of adaptation sizes by taking $\alpha \sr\in (0, 1]$ and $\lambda \sr> 0$ \citep{andrieu2008tutorial}.

\begin{figure}[t]
  \centering
  \begin{subfigure}[b]{.4\textwidth}
    \includegraphics[width=\textwidth]{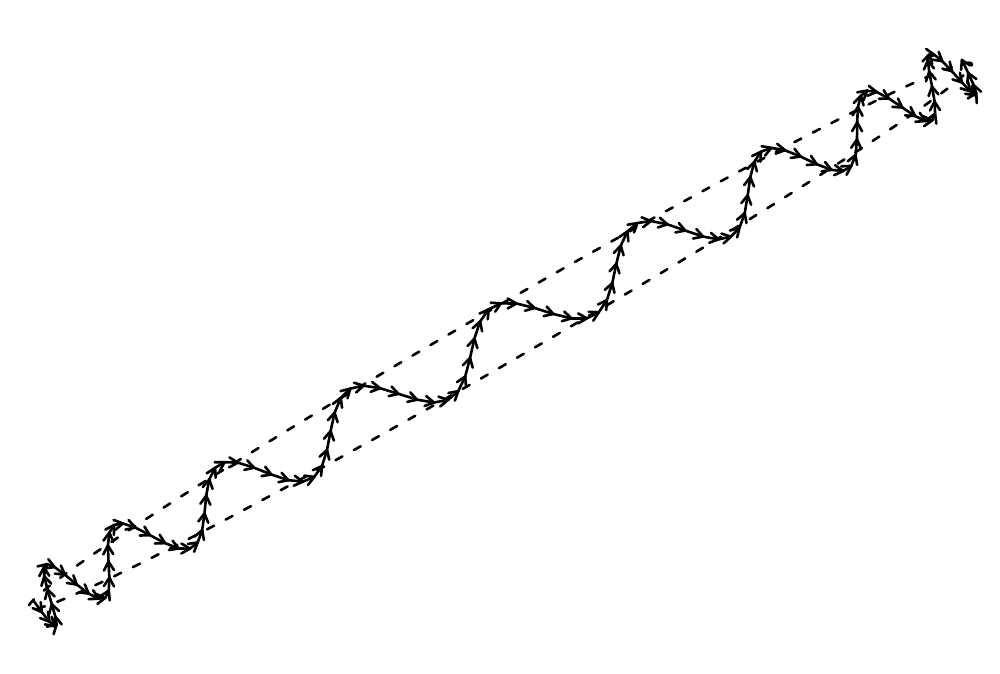}
    \caption{120 leapfrog jumps with $C=I$.}
    \label{fig:leapfrog_1}
  \end{subfigure}
  \quad
  \begin{subfigure}[b]{.4\textwidth}
    \includegraphics[width=\textwidth]{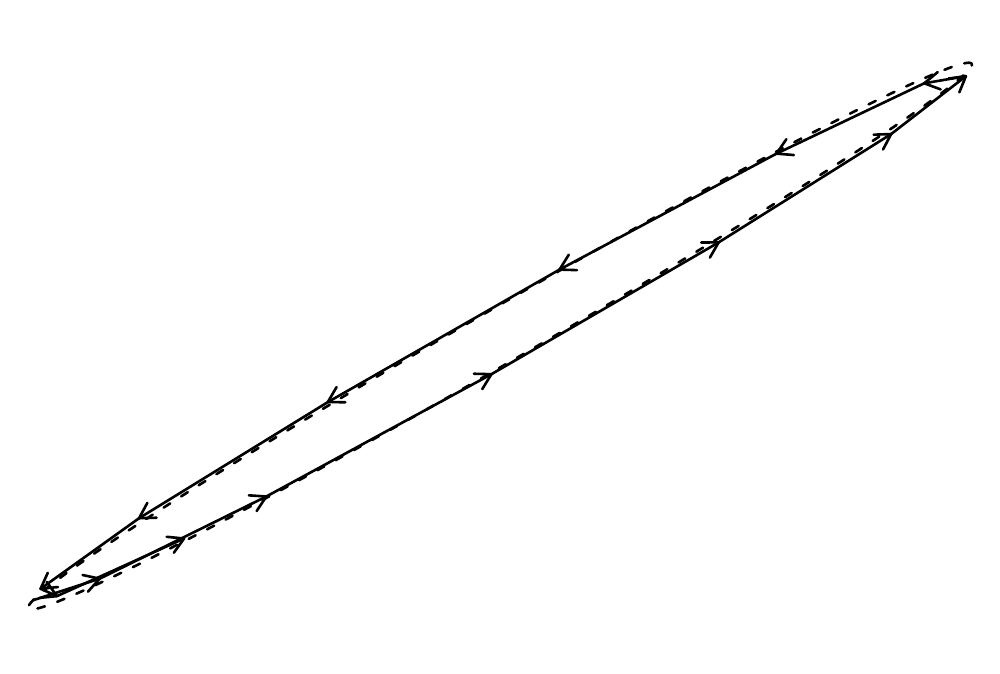}
    \caption{15 leapfrog jumps with $C=\Sigma$.}
    \label{fig:leapfrog_C}
  \end{subfigure}
  \caption{Two leapfrog trajectories for an ill conditioned Gaussian distribution with covariance $\Sigma$ for two different choices of $C$.
    In both cases, the leapfrog jump size $\epsilon$ was 0.5.
    A level set of the target density is shown as a dashed ellipsoid.}
  \label{fig:leapfrog_trajectories}
\end{figure}
Tuning the covariance $C$ of the velocity distribution can also increase the numerical efficiency of HMC algorithms.
If the marginal distributions for the target density $\bar\pi$ along different directions have orders of magnitude differences in standard deviation, the size of leapfrog jumps should typically be on the order of the smallest standard deviation in order to avoid numerical instability \citep{neal2011mcmc}.
In this case a large number of leapfrog jumps are needed to make a global move in the direction having the largest standard deviation.
Figure~\ref{fig:leapfrog_1} shows an example of leapfrog trajectory when the target distribution is $\mathcal N(0,\Sigma)$, where the standard deviation of one principal component of $\Sigma$ is twenty times larger than the other.
In this diagram, the leapfrog trajectory takes about 120 jumps to move across the less constrained direction from one end of a level set to the other.
Choosing a covariance $C$ for the velocity distribution close to the covariance of the target distribution can substantially reduce the number of leapfrog jumps needed to explore the sample space in every direction \citep[Section~5.4.1]{neal2011mcmc}.
Figure~\ref{fig:leapfrog_C} shows that a leapfrog trajectory can loop around the level set with only fifteen jumps when the covariance of the velocity distribution $C$ is equal to $\Sigma$.
We note that the covariance $C$ affects not only the velocity distribution and the leapfrog updates, but also the U-turn condition \eqref{eqn:stopping_cond_spNUTS2} for NUTS-type algorithms (the original NUTS, spNUTS1, and spNUTS2) via the $\cosangle$ function.

For adaptive tuning, the covariance $C_i$ used at the $i$-th iteration can be set equal to the sample covariance of the Markov chain sampled up to the previous iteration.
During initial iterations, a fixed covariance $C_0$ can be used to avoid numerical instability \citep{haario2001adaptive}:
\[
C_i \gets \left\{ \begin{array}{ll}
  C_0 & i \sr \leq i_0 \\
  \text{sample covariance of }\{X^{(j)} \giventh j\sr\leq i{-}1\} & i\sr>i_0.
\end{array}\right.
\]
It is possible to take $C_i$ as a diagonal matrix whose diagonal entries are given by the sample marginal variances of each component of the Markov chain \citep{haario2005componentwise}.
This approach is effective when the components have different scales.
The computational cost can be substantially reduced by using a diagonal covariance matrix when the target distribution is high dimensional, because operations such as the Cholesky decomposition of $C_i$ can be avoided.
Marginal sample variances can be updated with little overhead at each iteration using a recursive formula.

\subsection{Numerical examples}\label{sec:HMCnumerical}
\subsubsection{Multivariate normal distribution}
We used two examples to study the numerical efficiency of various algorithms discussed in this paper.
We first considered a one hundred dimensional normal distribution $\mathcal N(0,\Sigma)$ where the covariance matrix $\Sigma$ is diagonal and the marginal standard deviations form a uniformly increasing sequence from 0.01 to 1.00.

\begin{figure}[t]
  \captionsetup{width=0.88\linewidth}
  \centering
  \includegraphics[width=.9\linewidth]{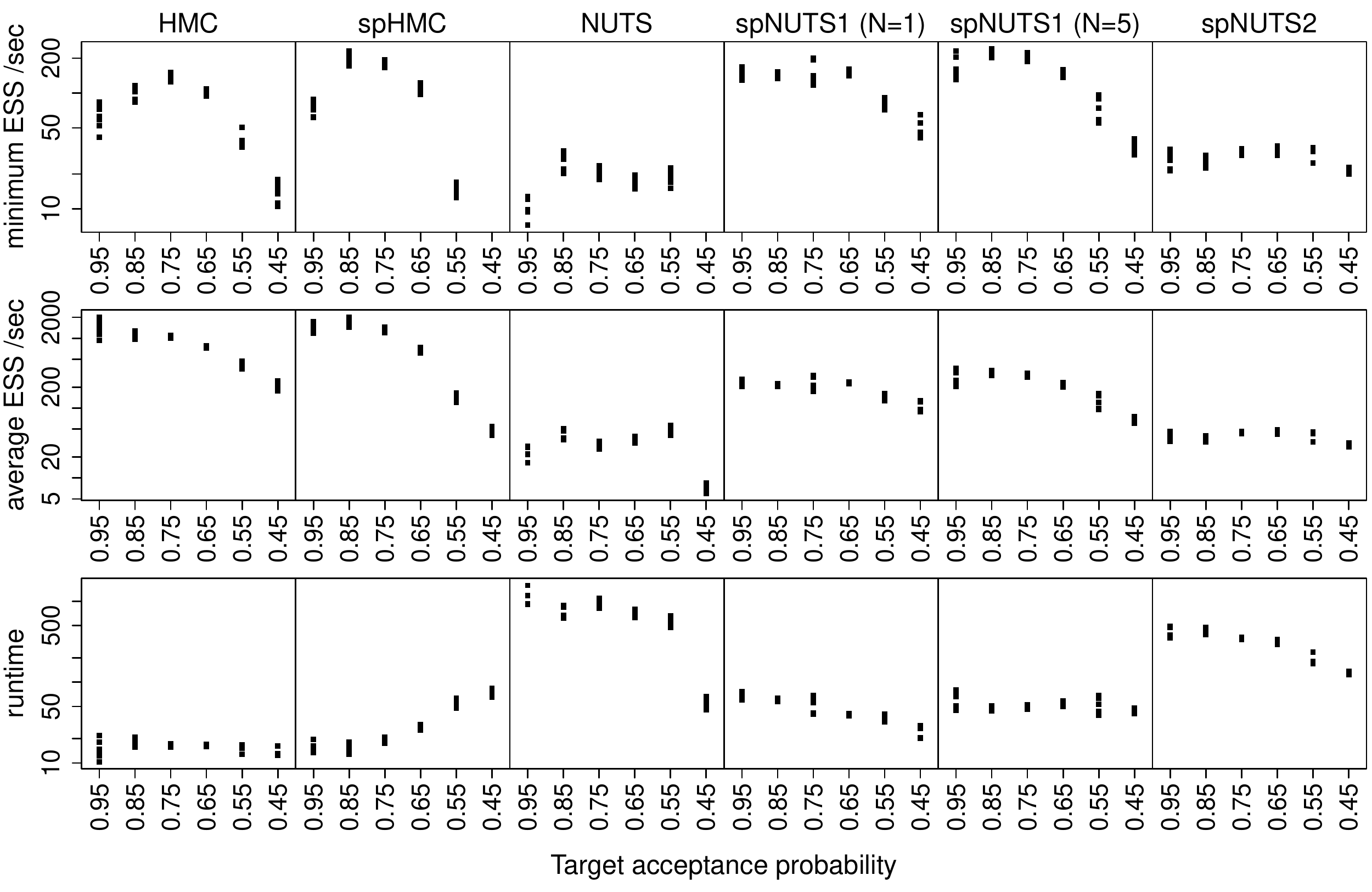}
  \caption{The minimum and average effective sample sizes of constructed Markov chains across $d\,{=}\,100$ variables per second of runtime for the target distribution $\mathcal N(0,\Sigma)$ when the covariance $C$ of the velocity distribution was fixed. The target acceptance probabilities are shown on the $x$-axis. The runtime in seconds are shown in the bottom row of plots. All $y$-axes are in logarithmic scales.}
  \label{fig:ESSperSec_astar_algo_mvnorm_fixedmass}
\end{figure}

We compared the numerical efficiency of the following five algorithms, first without adaptively tuning the covariance of the velocity distribution $C$: the standard HMC, HMC with sequential proposals (abbreviated as spHMC, which is equivalent to XCGHMC in Algorithm~\ref{alg:leapfrog}), the NUTS algorithm by \citet{hoffman2014no} (the efficient version), spNUTS1 (Algorithm~\ref{alg:spNUTS1}), and spNUTS2 (Algorithm~\ref{alg:spNUTS2}).
All experiments were carried out using the implementation of the algorithms in \textsf{R} \citep{R}.
The source codes are available at \url{https://github.com/joonhap/spMCMC}.
The covariance matrix $C$ of the velocity distribution was set equal to the one hundred dimensional identity matrix.
The leapfrog step size $\epsilon$ was adaptively tuned using \eqref{eqn:eps_adaptive_HMC} with $\alpha\sr=0.7$ and target acceptance probabilities $a^*$ varying from 0.45 to 0.95.
The adaptation started from the one hundredth iteration.
The acceptance probability at the $i$-th iteration $a_i$ was computed using the state that was one leapfrog jump away from the current state of the Markov chain to ensure that the leapfrog jump size $\epsilon$ converges to the same value for the same target acceptance probability across the various algorithms.
When running HMC, spHMC, spNUTS1, and spNUTS2, the leapfrog step size was randomly perturbed at each iteration by multiplying to $\epsilon_i$ a uniform random number between 0.8 and 1.2.
Randomly perturbing the leapfrog step size can improve the mixing of the Markov chain constructed by HMC algorithms \citep{neal2011mcmc}.
For the NUTS, we found perturbing the leapfrog step size did not improve numerical efficiency and thus used $\epsilon_i$.
In HMC and spHMC, each proposal was obtained by making fifty leapfrog jumps.
In spHMC, a maximum of $N\,{=}\,10$ proposals were tried in each iteration and the first acceptable proposal was taken as the next state of the Markov chain (i.e., $L\,{=}\,1$).
In spNUTS1 and spNUTS2, the stopping condition was checked according to the schedule $b_j\,{=}\,2^{j-1}$ for $j \in 1\col j_{\max}$ with $j_{\max}\,{=}\,15$, and the unit number of leapfrog trajectories was set to one (i.e., $l\,{=}\,1$ in Algorithms~\ref{alg:spNUTS1} and \ref{alg:spNUTS2}).
The value $c$ in the stopping condition \eqref{eqn:stopping_cond_spNUTS2} was randomly drawn from a uniform$(0,1)$ distribution for each trajectory, as we found randomizing $c$ yielded better numerical results than fixing it at zero.
For the NUTS algorithm, randomizing $c$ did not improve the numerical efficiency, so each trajectory was stopped when the cosine angle fell below zero (i.e., $c\,{=}\,0$), as was in \citet{hoffman2014no}. 
In spNUTS1, the maximum number of proposals $N$ in each iteration was set to either one or five.
In spNUTS2, a maximum of $N\,{=}\,20$ consecutive states on a leapfrog trajectory were tried in each attempt to find an acceptable state.
Every algorithm ran for $M\sr=20{,}200$ iterations.
As a measure of numerical efficiency, the effective sample size (ESS) of each component of the Markov chain was computed using an estimate of the spectral density at frequency zero via the \texttt{effectiveSize} function in \textsf{R} package \texttt{coda} \citep{plummer2006coda}.
The first two hundred states of the Markov chains were discarded when computing the effective sample sizes.
Each experiment was independently repeated ten times.
All computations were carried out using the Boston University Shared Computing Cluster.

Figure~\ref{fig:ESSperSec_astar_algo_mvnorm_fixedmass} shows both the minimum and the average effective sample size for the one hundred variables divided by the runtime in seconds when the covariance $C$ of the velocity distribution was fixed at the identity matrix.
We observed that there were large variations in the effective sample sizes among the $d\sr= 100$ variables for the Markov chains constructed by HMC and spHMC, resulting in minimum ESSs much smaller than average ESSs.
This happened due to the fact that for some variables the leapfrog trajectories with fifty jumps consistently tended to return to states close to the initial positions.
The Markov chains mixed slowly in these variables.
On the other hand, the leapfrog trajectories tended to reach the opposite side of the level set of Hamiltonian for some variables, for which the autocorrelation at lag one was close to $-1$.
For these variables, the effective sample size was greater than the length of the Markov chain $M$.
There were much variations in the effective sample size among the variables for the Markov chains constructed by spNUTS1 and spNUTS2 when the stopping cosine angle $c$ was fixed at zero, but variations diminished when $c$ was varied uniformly in the interval $(0,1)$.

The highest value of the minimum ESS per second achieved by spHMC among various values of the target acceptance probability was about fifty percent higher than that by the standard HMC.
For this multivariate normal distribution, the number of leapfrog jumps $l\sr= 50$ for HMC and spHMC was within the range of the average number of jumps in the leapfrog trajectories constructed by the NUTS, spNUTS1, and spNUTS2 algorithms.
Thus the effective sample sizes by HMC and spHMC were comparable to those by the other three algorithms, but the runtimes tended to be shorter.
The highest minimum ESS per second by spNUTS1 with $N\sr= 5$ was 7.6 times higher than that by the NUTS and 6.9 times higher than that by spNUTS2.
The runtimes of the NUTS were more than ten times longer than those of spNUTS1 and twice longer than those of spNUTS2.
This happened because the evaluation of the gradient of the log target density took much less computation time than the evaluation of the log target density for this example.
The highest minimum ESS per second by spNUTS1 when up to five sequential proposals were made (i.e., $N\sr= 5$) was twenty percent higher than when only one proposal was made ($N\sr= 1$).

\begin{figure}[t]
  \captionsetup{width=0.88\linewidth}
  \centering
  \includegraphics[width=.9\linewidth]{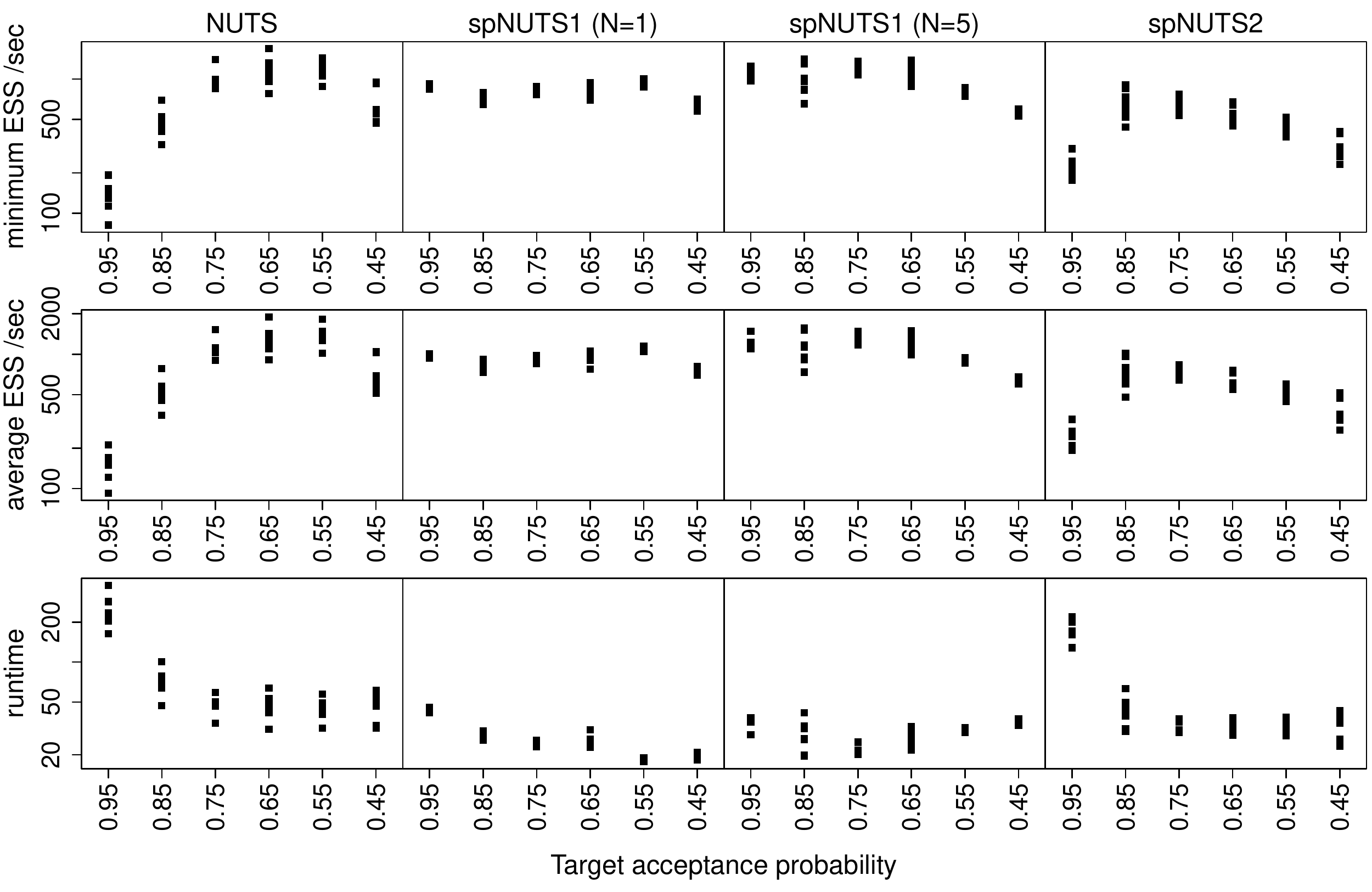}
  \caption{The minimum and average effective sample sizes per second of runtime for the target distribution $\mathcal N(0,\Sigma)$ when the covariance $C$ of the velocity distribution was adaptively tuned. The target acceptance probabilities are shown on the $x$-axis.}
  \label{fig:ESSperSec_astar_algo_mvnorm}
\end{figure}

Next we ran the NUTS, spNUTS1, and spNUTS2 algorithms for the same target distribution $\mathcal N(0,\Sigma)$ but with adaptively tuning the covariance $C$ of the velocity distribution.
The covariance $C_i$ at the $i$-th iteration for $i\sr\geq 100$ was set to a diagonal matrix whose diagonal entries were given by the marginal sample variances of the Markov chain constructed up to that point.
We did not test HMC or spHMC when $C$ was adaptively tuned, because the leapfrog step size, and thus the total length of the leapfrog trajectory with a fixed number of jumps, varied depending on the tuned values for $C$.
Figure~\ref{fig:ESSperSec_astar_algo_mvnorm} shows the minimum and average ESS among $d\sr= 100$ variables divided by the runtime in seconds.
The highest minimum ESS per second improved more than fifty times, compared to when the covariance $C$ was fixed, for the NUTS.
There was more than a five-fold improvement for spNUTS1, and more than a 25-fold improvement for spNUTS2.
The highest minimum ESS per second by the NUTS was $19\%$ higher than that by spNUTS1 ($N\sr= 5$) and $86\%$ higher than that by spNUTS2.
The NUTS was relatively more efficient when $C$ was adaptively tuned because the trajectories were built using fewer leapfrog jumps.
The computational advantage of spNUTS1 that the log target density is not evaluated at every leapfrog jump is relatively small when there are only few jumps per trajectory.
When $C$ is close to $\Sigma$, the sampling task is essentially equivalent to sampling from the standard normal distribution, in which case larger leapfrog step sizes may be used.
The trajectories were made of five to eight leapfrog jumps at the most efficient target acceptance probability when $C$ was adaptively tuned.
In comparison, the number of leapfrog jumps in a trajectory was between 80 and 250 when $C$ was not adaptively tuned.

\subsubsection{Bayesian logistic regression model}\label{sec:num_logit}
\begin{figure}[t]
  \captionsetup{width=0.88\linewidth}
  \centering
  \includegraphics[width=.9\linewidth]{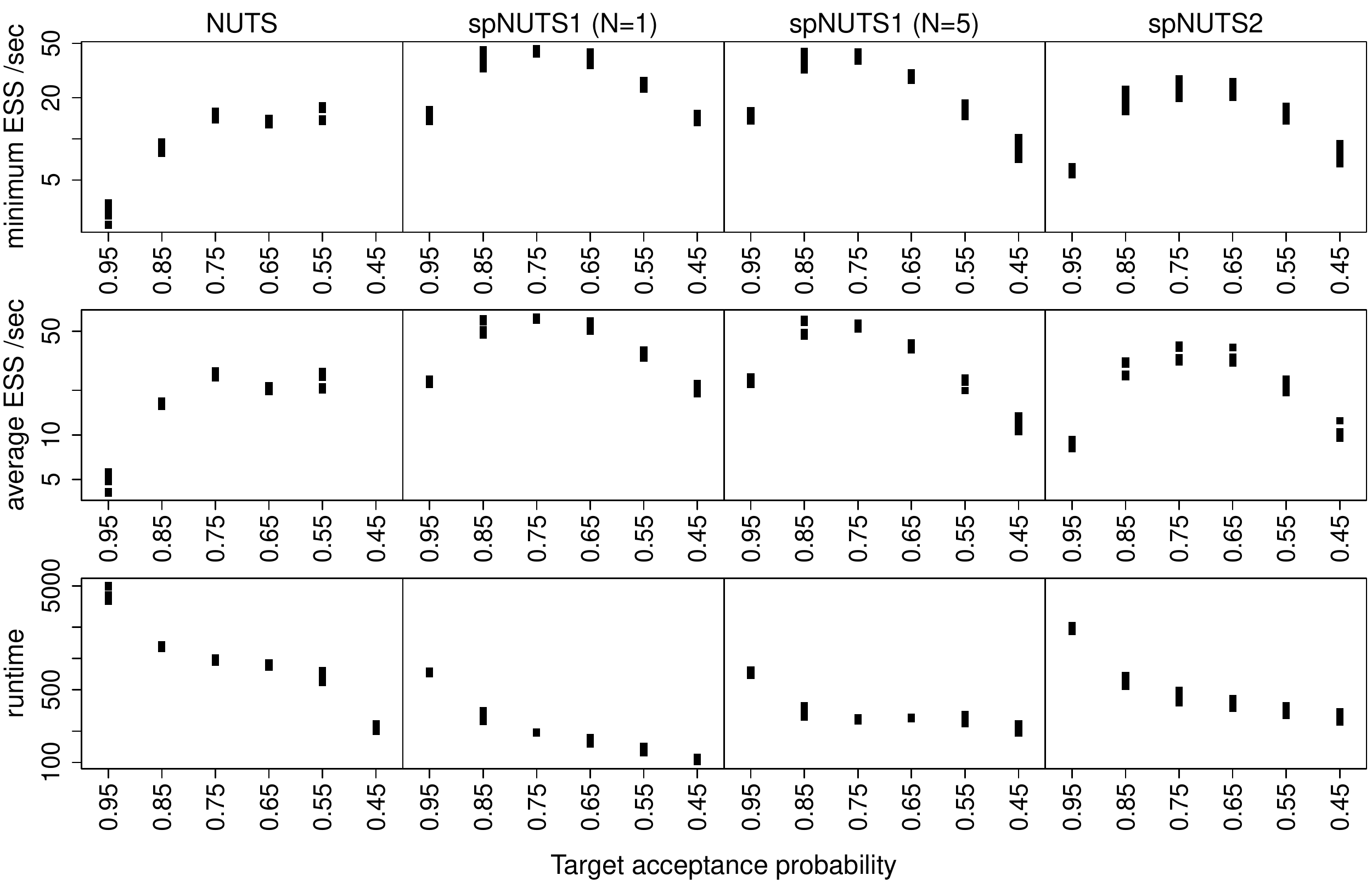}
  \caption{The minimum and average effective sample sizes per second of runtime across $d\sr= 25$ variables for the posterior distribution for the Bayesian logistic regression model in Section~\ref{sec:num_logit}.}
  \label{fig:ESSperSec_astar_algo_logit}
\end{figure}

We also experimented the numerical efficiency of the NUTS, spNUTS1, and spNUTS2 using the posterior distribution for a Bayesian logistic regression model.
The Bayesian logistic regression model and the data we used are identical to those considered by \citet{hoffman2014no}.
The German credit dataset from the UCI repository \citep{dua2019} consists of twenty four attributes of individuals and one binary variable classifying those individuals' credit.
The posterior density is proportional to
\begin{equation*}
  \pi(\alpha,\beta \given x,y)
  \propto \exp\left\{-\sum_{i=1}^{1000} \log(1+\exp\{-y_i(\alpha+x_i\cdot \beta\}) - \frac{\alpha^2}{200} - \frac{\Vert \beta \Vert_2^2}{200}\right\},
\end{equation*}
where $x_i$ denotes the twenty four dimensional covariate vector for the $i$-th individual and $y_i$ denotes the classification result taking a value from $\pm 1$.
We did not normalize the covariates to zero mean and unit variance as in \citet{hoffman2014no}, because we let $C$ be adaptively tuned.
The covariance $C$ was set to a diagonal matrix having as its diagonal entries the marginal sample variances of the constructed Markov chain up to the previous iteration.
All algorithms were run under the same settings as those used for the multivariate normal distribution example.

Figure~\ref{fig:ESSperSec_astar_algo_logit} shows the minimum and average ESS across $d\sr= 25$ variables per second of runtime.
The minimum ESS per second by spNUTS1 at the most efficient target acceptance probability was 2.6 times higher than that by the NUTS and 1.7 times higher than that by spNUTS2.
The differences in the numerical efficiency was led mostly by the differences in the runtime.
The numbers of leapfrog jumps in stopped trajectories tended to be larger than those for the normal distribution example due to the correlations between the variables in this Bayesian logistic regression model; the numbers of leapfrog jumps were about fifty for the NUTS, twenty seven for spNUTS1, and twenty two for spNUTS2.

\section{Conclusion}\label{sec:conclusion}
The sequential-proposal MCMC framework is readily applicable to a wide range of MCMC algorithms.
The flexibility and simplicity of the framework allow for various adjustments to the algorithms and offer possibilities of developing new ones.
In this paper, we showed that the numerical efficiency of MCMC algorithms can be improved by using sequential proposals.
In particular, we developed two novel NUTS-type algorithms, which showed higher numerical efficiency than the original NUTS by \citet{hoffman2014no} on two examples we examined.
In Appendix~\ref{sec:BPS}, we apply the sequential-proposal framework to the bouncy particle sampler (BPS) and demonstrate an advantageous property that the sequential-proposal BPS can readily make jumps between multiple modes.
The possibilities of other applications of the sequential-proposal MCMC framework can be explored in future research.

{\footnotesize
\paragraph{Acknowledgement}
This work was supported by National Science Foundation grants DMS-1513040 and DMS-1308918.
The authors thank Edward Ionides, Aaron King, and Stilian Stoev for comments on an earlier draft of this manuscript.
The authors also thank Jes\'us Mar\'ia Sanz-Serna for informing us about related references.
}

\appendix
\section{Proof of detailed balance for Algorithm~\ref{alg:spMH} (sequential-proposal Metropolis-Hastings algorithm)}\label{sec:proof_spMH}

Here we give a proof that Algorithm~\ref{alg:spMH} constructs a reversible Markov chain with respect to the target density $\bar\pi$.
In what follows, we denote the $l$-th rank of a given finite sequence $a_{n:m}$ by $r_l(a_{n:m})$; that is, if we reorder the sequence $a_{n:m}$ as $a_{(1)} \geq a_{(2)} \geq \cdots \geq a_{(m-n+1)}$, then $r_l(a_{n:m}) = a_{(l)}$.
If $l$ is greater than the length of the sequence $a_{n:m}$, we define $r_l(a_{n:m}) := 0$.
We also define $r_0(a_{n:m}) := \infty.$

\begin{prop}\label{prop:detailedbalanceMH}
  The Markov chain $\left(X^{(i)}\right)_{i\in 1:M}$ constructed by Algorithm~\ref{alg:spMH} is reversible with respect to the target density $\bar\pi$. 
\end{prop}
\begin{proof}
  It suffices to show the claim for fixed $N$ and $L$.
  The general case immediately follows by considering a mixture over $N$ and $L$ according to $\nu(N,L)$.

  We will show that for a given $n\sr\in 1\col N$, the probability density of taking $y_n$ as the next state of the Markov chain starting from the current state $y_0$ after rejecting a sequence of proposals $y_{1:n-1}$ is the same as the probability density of taking $y_0$ starting from $y_n$ after going through a reversed sequence of proposals $y_{n-1:1}$.
  The case for $n\sr=1$ coincides with a standard Metropolis-Hastings algorithm.
  We now fix $n\sr\geq 2$.
  Denoting the uniform$(0,1)$ random variable drawn at the beginning of the iteration by $\Lambda$, the $k$-th proposal $y_k$ is considered acceptable if and only if
  \[
  \Lambda < \frac{\pi(y_k) \prod_{j=1}^k q(y_{j-1}\given y_j)} {\pi(y_0) \prod_{j=1}^{k} q(y_j \given y_{j-1})}.
  \]
  Multiplying both the numerator and the denominator by $\prod_{j=k+1}^{n} q(y_j \given y_{j-1})$, we see that the above condition is equivalent to
  \begin{equation}
  \Lambda < \frac{\pi(y_k) \prod_{j=1}^k q(y_{j-1}\given y_j) \prod_{j=k+1}^{n} q(y_j \given y_{j-1})} {\pi(y_0) \prod_{j=1}^{n} q(y_j \given y_{j-1})}.
  \label{eqn:proof_spMH_accep_cri}
  \end{equation}
  For $k\sr \in 0\col n$, we define the following quantities
  \[
  p_k(y_0,y_1,\dots,y_n) := \pi(y_k) \prod_{j=1}^k q(y_{j-1}\given y_j) \prod_{j=k+1}^{n} q(y_j \given y_{j-1}),
  \]
  such that the condition \eqref{eqn:proof_spMH_accep_cri} can be concisely written as
  \[
  \Lambda < \frac{p_k(y_0,y_1,\dots,y_n)}{p_0(y_0,y_1,\dots,y_n)}.
  \]
  In what follows, $p_k(y_0,y_1,\dots,y_n)$ will be denoted by $p_k$ for brevity.
  The proposal $y_n$ is taken as the next state of the Markov chain if and only if it is the $L$-th acceptable proposal among the sequence of proposals $y_{1:n}$.
  This happens if and only if $\Lambda < p_n/p_0$ and there are exactly $L {-} 1$ proposals among $y_{1:n-1}$ such that $\Lambda < p_k/p_0$.
  The latter condition is satisfied if and only if $\Lambda$ is less than the $L{-} 1$-th largest number among $p_{1:n-1}/p_0$ but greater than or equal to the $L$-th largest number among the same sequence, that is,
  \[
  r_L\left(\frac{p_{1:n-1}}{p_0}\right) \leq \Lambda < r_{L-1}\left(\frac{p_{1:n-1}}{p_0}\right).
  \]
  Under the assumption that $X^{(i)}$ is distributed according to the target density $\bar\pi$, the probability that the current state of the Markov chain is in a set $A \sr\in \mathcal X$ and the $n$-th proposal, which is in a set $B \sr\in \mathcal X$, is taken as the next state of the Markov chain is given by
  \begin{equation}
  \begin{split}
    &\int \1_A(y_0) \1_B(y_n) \bar\pi(y_0) \prod_{j=1}^n q(y_j \given y_{j-1}) \1\left[\Lambda \geq r_L\left( \frac{p_{1:n-1}}{p_0} \right) \right] \1\left[\Lambda < r_{L-1} \left( \frac{p_{1:n-1}}{p_0} \right) \right] \1\left[ \Lambda < \frac{p_n}{p_0} \right] \\
    & \hspace{70ex}\cdot \1\left[ 0 < \Lambda < 1\right] d\Lambda\, dy_{0:n}\\
    &= \int \1_A(y_0) \1_B(y_n)\cdot \frac{p_0}{Z} \cdot \1\left[\Lambda \geq r_L\left( \frac{p_{1:n-1}}{p_0} \right) \right] \cdot \1\left[\Lambda < \min\left\{ r_{L-1} \left( \frac{p_{1:n-1}}{p_0} \right), \frac{p_n}{p_0}, 1 \right\} \right] d\Lambda\, dy_{0:n}\\
    &= \int \1_A(y_0) \1_B(y_n) \cdot \frac{p_0}{Z} \cdot \left( \min\left\{ \frac{r_{L-1}(p_{1:n-1})}{p_0}, \frac{p_n}{p_0}, 1\right\} - \min\left\{ \frac{r_L(p_{1:n-1})}{p_0}, \frac{p_n}{p_0}, 1\right\} \right) d\Lambda \, dy_{0:n}\\
    &= \frac{1}{Z} \int \1_A(y_0) \1_B(y_n) \cdot \big( \min\{r_{L-1}(p_{1:n-1}), p_n, p_0\} - \min\{r_L(p_{1:n-1}), p_n, p_0\} \big) d\Lambda\, dy_{0:n}.
  \end{split}
  \label{eqn:proof_spMH_int}
  \end{equation}
  We will change the notation of dummy variables by writing $y_0 \gets y_n$, $y_1 \gets y_{n-1}$, $\dots$, $y_n \gets y_0$.
  But note that $p_k(y_n,y_{n-1},\dots, y_0)$ can be expressed as
  \[
  \pi(y_{n-k}) \prod_{j=1}^{k} q(y_{n-j+1}\given y_{n-j}) \prod_{j=k+1}^n q(y_{n-j} \given y_{n-j+1})
  = \pi(y_{n-k}) \prod_{j=n-k+1}^n q(y_j \given y_{j-1}) \prod_{j=1}^{n-k} q(y_{j-1} \given y_j),
  \]
  which is the same as an expression for $p_{n-k}(y_0,y_1,\dots, y_n)$.
  Thus under the change of notation, \eqref{eqn:proof_spMH_int} can be re-written as
  \[
  \frac{1}{Z} \int \1_A(y_n) \1_B(y_0) \cdot \big[ \min\{r_{L-1}(p_{n-1:1}), p_0, p_n\} - \min\{r_L(p_{n-1:1}), p_0, p_n\} \big] d\Lambda\, dy_{n:0}.
  \]
  where $p_k$ denotes $p_k(y_0, y_1, \dots, y_n)$ for $k \sr\in 0\col n$.
  The above integral is equal to \eqref{eqn:proof_spMH_int} with the sets $A$ and $B$ interchanged.
  Thus we have proved that the probability that the current state of the Markov chain is in $A$ and the $n$-th proposal, which is in $B$, is taken as the next state of the Markov chain is equal to the probability that the current state is in $B$ and the $n$-th proposal, which is in $A$, is taken as the next state.
  Summing the established equality over all $n\sr\in 1\col N$ and finally noting that the next state in the Markov chain is set equal to the current state in the case where there were less than $L$ proposals found in the first $N$ proposals, we reach the conclusion that under the assumption that $X^{(i)}$ is distributed according to $\bar\pi$, the probability that the current state of Markov chain is in $A$ and the next state is in $B$ is the same as the probability that the current state is in $B$ and the next state is in $A$.
  This finishes the proof of detailed balance for Algorithm~\ref{alg:spMH}.
\end{proof}

\section{Sequential-proposal Metropolis-Hastings algorithms with proposal kernels dependent on previous proposals}\label{sec:spMH_path_depen_kernel}

\begin{figure}[t]
  \centering
  \scalebox{.89}{\begin{minipage}{\textwidth}
\begin{algorithm}[H]
  \SetKwInOut{Input}{Input}\SetKwInOut{Output}{Output}
  \Input{Distribution for the maximum number of proposals and the number of accepted proposals, $\nu(N,L)$\\
  Path-dependent proposal kernels, $\{q_j(\cdot \given x_{j-1}, \dots, x_0) \giventh j \sr \geq 1\}$\\
  Number of iterations, $M$}
  \vspace{1ex}
  \Output{A draw of Markov chain, $\left(X^{(i)}\right)_{i\in 1:M}$}
  \vspace{1ex}
  \textbf{Initialize:} Set $X^{(0)}$ arbitrarily
  
  \For {$i\gets 0\col M{-}1$}{
    Draw $(N,L) \sim \nu(\cdot, \cdot)$\\
    Draw $\Lambda\sim \text{unif}(0,1)$\\
    Set $X^{(i+1)} \gets X^{(i)}$\\
    Set $Y_0 \gets X^{(i)}$ and $n_a \gets 0$\\
    \For {$n \gets 1\col N$} {
      Draw $Y_n \sim q_n(\cdot \given Y_{n-1:0})$\\
      \textbf{if} {$\displaystyle \Lambda < \frac{\pi(Y_n) \prod_{j=1}^n q_j(Y_{n-j} \given Y_{n-j+1:n}) } { \pi(Y_0) \prod_{j=1}^n q_j(Y_j \given Y_{j-1:0}) }$} \textbf{then} $n_a \gets n_a + 1$\\
      \If {$n_a = L$} {
        \If {there exist exactly $L\sr -1$ cases among $k\sr\in 1\col n{-}1$ such that $\displaystyle \Lambda < \frac{\pi(y_{k}) \prod_{j=1}^{n-k} q_j(y_{k+j} \given y_{k+j-1:k}) \prod_{j=n-k+1}^n q_j(y_{n-j} \given y_{n-j+1:n})} { \pi(y_0) \prod_{j=1}^n q_j(y_j \given y_{j-1:0}) }$ } {
          Set $X^{(i+1)} \gets Y_n$
        }
        \texttt{break}
      }
    }
  }
  \caption{A sequential-proposal Metropolis Hasting algorithm using a path-dependent proposal kernel}
\label{alg:spMHgen}
\end{algorithm}
  \end{minipage}}
  \end{figure}

In Section~\ref{sec:spMHgen} we presented a generalization of Algorithm~\ref{alg:spMH} in which the proposal kernel can depend on previous proposals made in the same iteration.
A pseudocode for this generalized version is given in Algorithm~\ref{alg:spMHgen}.
A proof that Algorithm~\ref{alg:spMHgen} constructs a reversible Markov chain with respect to the target density $\bar\pi$ is given below.
\begin{prop}\label{prop:spMHgen_detailed_balance}
  Algorithm~\ref{alg:spMHgen} constructs a reversible Markov chain with respect to the target density $\bar\pi$.
\end{prop}
\begin{proof}
  Again we consider fixed $N$ and $L$, because the general case can easily follow by considering a mixture over $N$ and $L$.
  Let $\Lambda$ denote the uniform$(0,1)$ number drawn at the start of the iteration.
  We will denote the value of current state of the Markov chain by $y_0$, and the values of a sequence of proposals up to the $n$-th proposal as $y_1$, $\dots$, $y_n$.
  The $n$-th proposal $y_n$ is taken as the next state of the Markov chain if and only if
  \[
  \Lambda < \frac{\pi(y_n) \prod_{j=1}^n q_j(y_{n-j} \given y_{n-j+1:n})}{\pi(y_0) \prod_{j=1}^n q_j(y_j \given y_{j-1:0}) },
  \]
  and there are exactly $L-1$ numbers $k$ among $1\col n{-}1$ that satisfy
  \begin{equation}
  \Lambda < \frac{\pi(y_k) \prod_{j=1}^k q_j(y_{k-j} \given y_{k-j+1:k})}{\pi(y_0) \prod_{j=1}^k q_j(y_j \given y_{j-1:0}) },
  \label{eqn:spMHgen_k_accep}
  \end{equation}
  and also there exist exactly $L-1$ numbers $k'$ among $1\col n{-}1$ that satisfy
  \begin{equation}
    \Lambda < \frac{\pi(y_{k'}) \prod_{j=1}^{n-k'} q_j(y_{k'+j} \given y_{k'+j-1:k'}) \prod_{j=n-k'+1}^n q_j(y_{n-j} \given y_{n-j+1:n})} { \pi(y_0) \prod_{j=1}^n q_j(y_j \given y_{j-1:0}) }.
    \label{eqn:spMHgen_dual}
  \end{equation}
  The inequality \eqref{eqn:spMHgen_k_accep} can be expressed as
  \[
  \Lambda < \frac{\pi(y_k) \prod_{j=1}^k q_j(y_{k-j} \given y_{k-j+1:k}) \prod_{k+1}^n q_j(y_j \given y_{j-1:0})} {\pi(y_0) \prod_{j=1}^n q_j(y_j \given y_{j-1:0}) }.
  \]
  We note that the numerator in the expression above is the probability density of drawing a sequence of proposals in the order $y_k \sr\to y_{k-1} \sr\to \cdots \sr\to y_0 \sr\to y_{k+1} \sr\to y_{k+2} \sr\to \dots \sr\to y_n$, where the value $y_j$ for $j\sr\geq k\sr+1$ is drawn from a proposal density $q_j(\cdot \given y_{j-1}, y_{j-2}, \dots, y_0)$.
  We denote this probability density by
  \[
  p_k(y_0,y_1,\dots,y_n) := \pi(y_k) \prod_{j=1}^k q_j(y_{k-j} \given y_{k-j+1:k}) \prod_{k+1}^n q_j(y_j \given y_{j-1:0}).
  \]
  We also denote the numerator in \eqref{eqn:spMHgen_dual} by
  \[
  \overline p_k(y_0,y_1,\dots,y_n) := \pi(y_{k}) \prod_{j=1}^{n-k} q_j(y_{k+j} \given y_{k+j-1:k}) \prod_{j=n-k+1}^n q_j(y_{n-j} \given y_{n-j+1:n}),
  \]
  which gives the probability density of drawing proposals in the order $y_k \to y_{k+1} \to \cdots y_n \to y_{k-1} \to \cdots \to y_0$, where $y_j$ for $j \sr \leq k\sr- 1$ is drawn from $q_{n-j}(\cdot \given y_{j+1}, \dots, y_n)$.
  One can easily check the following relations:
  \begin{equation}
    \begin{aligned}
      p_n(y_{0:n}) = p_0(y_{n:0}),& \qquad p_0(y_{0:n}) = p_n(y_{n:0}), \qquad\text{and}\\
      p_k(y_{0:n}) = \overline p_{n-k}(y_{n:0}),& \qquad p_{k}(y_{n:0}) = \overline p_{n-k}(y_{0:n}) \qquad \text{for } k \in 0\col n,\\
    \end{aligned}
    \label{eqn:spMHgen_sym_p}
  \end{equation}
  where we remind the reader of our notation $y_{0:n} := (y_0, y_1, \dots, y_n)$ and $y_{n:0} := (y_n, y_{n-1}, \dots, y_0)$.
  Now \eqref{eqn:spMHgen_k_accep} and \eqref{eqn:spMHgen_dual} can be concisely expressed as
  \[
  \Lambda < \frac{p_k(y_{0:n})}{p_0(y_{0:n})}, \quad \text{ and } \quad \Lambda < \frac{\overline p_k(y_{0:n})}{p_0(y_{0:n})}
  \]
  respectively.
  The conditions required for taking $y_n$ as the next state of the Markov chain can be summarized by the following inequalities:
  \begin{gather*}
    \Lambda \geq r_L\left( \frac{p_{1:n-1}}{p_0}(y_{0:n}) \right), \qquad \Lambda < r_{L-1}\left( \frac{p_{1:n-1}}{p_0}(y_{0:n}) \right),\\
    \Lambda \geq r_L\left( \frac{\overline p_{1:n-1}}{\overline p_0}(y_{0:n}) \right), \qquad \Lambda < r_{L-1}\left( \frac{\overline p_{1:n-1}}{\overline p_0}(y_{0:n}) \right),\\
    \text{and }\quad \Lambda < \frac{p_n}{p_0},
  \end{gather*}
  where $r_L$ denotes the function returning the $L$-rank as defined in Section~\ref{sec:proof_spMH}, and $\frac{p_{1:n-1}}{p_0}(y_{0:n})$ denotes the sequence of values $\big(\frac{p_1(y_{0:n})}{p_0(y_{0:n})}, \dots, \frac{p_{n-1}(y_{0:n})}{p_0(y_{0:n})}\big)$.
  In what follows, $p_k(y_{0:n})$ and $\overline p_k(y_{0:n})$ will be written as $p_k$ and $\overline p_k$ for brevity.
  Under the assumption that at the current iteration the state of the Markov chain is distributed according to $\bar\pi$, the probability that the current state is in $A$ and the $n$-th proposal, which is in $B$, is taken as the next state of the Markov chain is given by
  \begin{equation}\begin{split}
  &\int \1_A(y_0) \1_B(y_n) \bar\pi(y_0) \prod_{j=1}^n q_j(y_j \given y_{j-1:0}) \1\left[ \Lambda < \min\left\{ 1, \frac{p_n}{p_0}, r_{L-1}\left(\frac{p_{1:n-1}}{p_0}\right), r_{L-1}\left(\frac{\overline p_{1:n-1}}{p_0}\right) \right\} \right] \\
  &\hspace{45ex}\cdot \1\left[ \Lambda \geq \max\left\{r_L\left( \frac{p_{1:n-1}}{p_0} \right), r_L \left( \frac{\overline p_{1:n-1}}{p_0} \right) \right\} \right] d\Lambda \, dy_{0:n}\\
  &= \int \1_A(y_0) \1_B(y_n) \frac{p_0}{Z} \left[ \min\left\{ 1, \frac{p_n}{p_0}, r_{L-1}\left(\frac{p_{1:n-1}}{p_0}\right), r_{L-1}\left(\frac{\overline p_{1:n-1}}{p_0}\right) \right\} \right.\\
  &\hspace{8ex} \left.- \min\left\{ 1, \frac{p_n}{p_0}, r_{L-1}\left(\frac{p_{1:n-1}}{p_0}\right), r_{L-1}\left(\frac{\overline p_{1:n-1}}{p_0}\right), \max\left\{r_L\left(\frac{p_{1:n-1}}{p_0}\right), r_L\left( \frac{\overline p_{1:n-1}}{p_0} \right) \right\} \right\} \right] dy_{0:n}\\
  &=\frac{1}{Z} \int \1_A(y_0) \1_B(y_n) \big[ \min\{ p_0, p_n, r_{L-1}(p_{1:n-1}), r_{L-1}(\overline p_{1:n-1}) \} \\
  &\hspace{22ex} - \min\{ p_0, p_n, r_{L-1}(p_{1:n-1}), r_{L-1}(\overline p_{1:n-1}), \max\{ r_L(p_{1:n-1}), r_L(\overline p_{1:n-1})\} \} \big] dy_{0:n}
  \end{split}
  \label{eqn:spMHgen_prob}
  \end{equation}
  We now change the notation of dummy variables by writing $y_0 \gets y_n$, $y_1 \gets y_{n-1}$, $\dots$, $y_n \gets y_0$, and noting the relations \eqref{eqn:spMHgen_sym_p}, we may rewrite \eqref{eqn:spMHgen_prob} as
  \[
  \begin{split}
  &\frac{1}{Z} \int \1_A(y_n) \1_B(y_0) \big[ \min\{ p_n, p_0, r_{L-1}(\overline p_{n-1:1}), r_{L-1}(p_{n-1:1}) \} \\
  &\hspace{22ex} - \min\{ p_n, p_0, r_{L-1}(\overline p_{n-1:1}), r_{L-1}(p_{n-1:1}), \max\{ r_L(\overline p_{n-1:1}), r_L(p_{n-1:1})\} \} \big] dy_{n:0}
  \end{split}
  \]
  But the above display is equal to what is obtained when the sets $A$ and $B$ are interchanged in \eqref{eqn:spMHgen_prob}.
  Thus we have proved that, denoting the current state of the Markov chain as $X^{(i)}$ and the next state as $X^{(i+1)}$ and assuming that $X^{(i)}$ is distributed according to $\bar\pi$,
  \begin{multline*}
  \mathcal P[X^{(i)} \in A, X^{(i+1)} \in B, \text{the }n\text{-th proposal is taken as }X^{(i+1)}]\\
  = \mathcal P[X^{(i)} \in B, X^{(i+1)} \in A, \text{the }n\text{-th proposal is taken as }X^{(i+1)}].
  \end{multline*}
  Summing the above equation for $n \sr \in 1\col N$ and considering that $X^{(i+1)}$ is set equal to $X^{(i)}$ for all scenarios except when a proposal among $y_0$, $\dots$, $y_N$ is taken as the next state of the Markov chain, we reach the conclusion that 
  \[
  \mathcal P[X^{(i)} \in A, X^{(i+1)} \in B] = \mathcal P[X^{(i)} \in B, X^{(i+1)} \in A],
  \]
  which shows the desired detailed balance of the Markov chain with respect to $\bar\pi$.
\end{proof}

\section{Equivalence between the sequential-proposal Metropolis-Hastings algorithm and the delayed rejection method (only in the case where proposals are not path-dependent)}\label{sec:equiv_delayedrej}
Before showing the equivalence between sequential-proposal Metropolis-Hastings algorithms for $L\sr=1$ and the delayed rejection method when the proposal kernel is not path-dependent, we briefly check that the target density $\bar\pi$ is invariant in the delayed rejection method.
It suffices to check that the detailed balance equation holds for each $n\sr\in 1\col N$.
The probability density that $y_0$ is drawn from $\bar\pi$, the proposals $y_1, \dots, y_n$ are drawn sequentially, and $y_n$ is the first accepted value is given by
\begin{equation*}\begin{split}
    &\bar\pi(y_0) \prod_{j=1}^n q(y_j \given y_{j-1}) \prod_{j=1}^{n-1} \{1-\alpha_j(y_{0:j})\} \cdot \alpha_n(y_{0:n})\\
    &= \left[ \bar\pi(y_0) \prod_{j=1}^n q(y_j \given y_{j-1}) \prod_{j=1}^{n-1} \{1-\alpha_j(y_{0:j})\}\right] \land \left[\bar\pi(y_n) \prod_{j=1}^n q(y_{j-1}\given y_j) \prod_{j=1}^{n-1}\{1-\alpha_j(y_{n:n-j})\}\right].
  \end{split}\end{equation*}
Since the above quantity is symmetric with respect to reversing the order of sequence from $y_{0:n}$ to $y_{n:0}$, it also equals the probability density that starting from $y_n$, the proposals $y_{n-1}, \dots, y_1, y_0$ are drawn and $y_0$ becomes the first accepted value, which is given by
\[
\bar\pi(y_n) \prod_{j=1}^n q(y_{j-1}\given y_j) \prod_{j=1}^{n-1}\{1-\alpha_j(y_{n:n-j})\} \cdot \alpha_n(y_{n:0}).
\]
Combining the case for $n \sr \in 1\col N$, we see that the Markov chain constructed by the delayed rejection method is reversible with respect to the target density $\bar\pi$.
We now prove the following proposition.

\begin{prop}\label{prop:equiv_delayedrej}
  Sequential-proposal Metropolis-Hastings algorithm (Algorithm~\ref{alg:spMH}) for $L\,{=}\,1$ and fixed $N$ constructs a Markov chain that has the same law as that constructed by the delayed rejection method.
\end{prop}
\begin{proof}
  For both Algorithm~\ref{alg:spMH} and the delayed rejection method, let the state of the Markov chain at the start of a certain iteration be denoted by $y_0$.
  In both algorithms, the probability density of drawing $y_0$ from $\bar\pi$ and a sequence of proposals $y_{1:N}$ using a proposal kernel with density $q$ is given by
  \[
  \bar\pi(y_0)\prod_{j=1}^n q(y_j \given y_{j-1}).
  \]
  Given the current state $y_0$ and a sequence of proposals $y_{1:n}$, the probability that $y_n$ is taken as the next state of the Markov chain is obtained by subtracting the probability that all of $y_{1:n}$ are rejected from the probability that $y_{1:n-1}$ are rejected.
  Thus it suffices to show that for an arbitrary $n \sr\in 1\col N$ and a given sequence of proposals $y_{1:n}$, the probability that all of the proposals are rejected is the same in both algorithms.
  In what follows we assume $y_{1:n}$ are drawn and fixed.
  We define for $k\in 1\col N$,
  \[
  c_k(y_{0:n}) := \frac{\pi(y_k) \prod_{j=1}^k q(y_{j-1}\given y_j)} {\pi(y_0) \prod_{j=1}^k q(y_j \given y_{j-1})}.
  \]
  For brevity, $c_k(y_{0:n})$ will be denoted simply by $c_k$ in cases where the argument $y_{0:n}$ can be clearly understood.
  In Algorithm~\ref{alg:spMH}, all $y_{1:n}$ are rejected if and only if $\Lambda \geq c_k$ for all $k\sr \in 1\col n$ where $\Lambda \sim \text{unif}(0,1)$ is the random number drawn at the start of the current iteration, so the probability that all $y_{1:n}$ are rejected is given by
  \[
  \int_0^1 \1[\Lambda \geq \max(c_{1:n})] d\Lambda
  = 1 - \max(c_{1:n}) \land 1.
  \]
  For the delayed rejection method, we denote
  \[
  \beta_k := \alpha_k(y_{0:k}), \quad \bar\beta_k^n := \alpha_k(y_{n:n-k}),  \quad k\sr\in 1\col n{-}1.
  \]
  Since the probability that all $y_{1:n}$ are rejected in an implementation of the delayed rejection method is given by $\prod_{k=1}^n (1-\beta_k)$, our goal is to show that
  \begin{equation}\prod_{k=1}^n (1-\beta_k) = 1 - \max(c_{1:n}) \land 1, \label{eqn:equiv_goal}\end{equation}
  which we will prove by induction.
  The case for $n=1$ is obvious from the definition of $\beta_1$.
  Suppose we have
  \begin{equation}
    \prod_{k=1}^{n-1}(1-\beta_k) = 1-\max(c_{1:n-1})\land 1.
    \label{eqn:equiv_delay_spMH_induc_hypo}
  \end{equation}
  We denote for $k \sr\in 1\col n{-}1$ and $n \sr\in 1\col N$
  \[
  \bar c^n_k := \frac{\pi(y_{n-k}) \prod_{j=1}^{k} q(y_{n-j+1} \given y_{n-j})}{\pi(y_n) \prod_{j=1}^{k} q(y_{n-j} \given y_{n-j+1})}.
  \]
  In \eqref{eqn:equiv_delay_spMH_induc_hypo}, both $\beta_{1:n-1}$ and $c_{1:n-1}$ are functions of $y_{0:n}$.
  If we change the notation by writing $y_0\gets y_n$, $y_1\gets y_{n-1}$, $\dots$, $y_n \gets y_0$, the equation \eqref{eqn:equiv_delay_spMH_induc_hypo} becomes
  \begin{equation}
    \prod_{k=1}^{n-1}(1-\bar \beta_k^n) = 1-\max(\bar c_{1:n-1}^n) \land 1.
    \label{eqn:induc_hypo_rev}
  \end{equation}
  It can also be easily checked from the definitions of $c_k$ and $\bar c_k^n$ that $c_n \bar c^n_k = c_{n-k}$ for $k\sr\in 1\col n{-}1$.
  Since the acceptance probability $\beta_k$ in the delayed rejection method is given by
  \[
  \beta_k = c_n \frac{\prod_{k=1}^{n-1} (1-\bar\beta_k^n)}{\prod_{k=1}^{n-1} (1-\beta_k)} \land 1,
  \]
  we observe that
  \begin{equation}\begin{split}
      \prod_{k=1}^n(1-\beta_k) &= \prod_{k=1}^{n-1}(1-\beta_k) \cdot \left(1- c_n \frac{\prod_{k=1}^{n-1} (1-\bar\beta_k^n)}{\prod_{k=1}^{n-1} (1-\beta_k)} \land 1\right)\\
      &= \{1-\max(c_{1:n-1})\land 1\} - [c_n \{1-\max(\bar c^n_{1:n-1}) \land 1\}] \land \{1-\max(c_{1:n-1})\land 1\}\\
      &= \{1-\max(c_{1:n-1})\land 1\} - \{c_n - \max(c_{n-1:1})\land c_n\} \land \{1-\max(c_{1:n-1})\land 1\}
      \label{eqn:equiv_delay_spMH_all_rejec}
  \end{split}\end{equation}
  using the induction hypothesis \eqref{eqn:equiv_delay_spMH_induc_hypo} and \eqref{eqn:induc_hypo_rev}.
  For real numbers $u$, $v$, and $w$, the following relation holds:
  \[
  (u\lor v) \land w \equiv u\land w + v \land w - u\land v \land w.
  \]
  We see that
  \begin{equation*}\begin{split}
      1 - \max(c_{1:n})\land 1 &= 1 - (\max(c_{1:n-1}) \lor c_n) \land 1\\
      &= 1 - \max(c_{1:n-1})\land 1 - c_n \land 1 + \max(c_{1:n-1})\land c_n \land 1.
  \end{split}\end{equation*}
  Thus from \eqref{eqn:equiv_delay_spMH_all_rejec}, showing \eqref{eqn:equiv_goal} reduces to checking
  \begin{equation}
    \{c_n - \max(c_{n-1:1})\land c_n\} \land \{1-\max(c_{1:n-1})\land 1\} = c_n \land 1 - \max(c_{1:n-1})\land c_n \land 1.
    \label{eqn:equiv_delay_spMH_c}
  \end{equation}
  However, one can check the following relation holds:
  \[(u - w\land u) \land (v - w\land v) \equiv u\land v - w\land u\land v,\]
  from which \eqref{eqn:equiv_delay_spMH_c} follows by letting $u=c_n$, $w=\max(c_{1:n-1})$, and $v=1$.
\end{proof}

\section{Proof of Proposition~\ref{prop:SPPD_invariance}}\label{sec:proof_SPPD_invariance}
In this section, we will show that sequential-proposal MCMC algorithms using deterministic kernels (Algorithm~\ref{alg:sppd}) construct reversible Markov chains with respect to the target density $\bar\pi$.
We first present the following lemma.
\begin{lemma}\label{lem:TSnTSn}
  Suppose \eqref{eqn:RR} and \eqref{eqn:TSTS} hold.
  Define recursively $\S_\tau^n := \S_\tau^{n-1} \circ \S_\tau$ where $\S_\tau^1 = \S_\tau$.
  Then for any $n \geq 1$, we have $\T \circ \S_\tau^n \circ \T \circ \S_\tau^n = \id.$
  Moreover, $\S_\tau$ is a bijective map.
\end{lemma}
\begin{proof}
  From \eqref{eqn:RR}, we have $\T\circ \T = \id$.
Thus from \eqref{eqn:TSTS}, we have $\T = \T\circ \T \circ \S_\tau\circ \T \circ \S_\tau = \S_\tau \circ \T \circ \S_\tau$.
Thus, we can see that $\S_\tau^n \circ \T$ is a self-inverse for any $n\geq 1$ from induction
$$\S_\tau^n \circ \T \circ \S_\tau^n \circ \T = \S_\tau^{n-1} \circ \S_\tau\circ \T\circ \S_\tau\circ \S_\tau^{n-1} \circ \T = \S_\tau^{n-1} \circ \T \circ \S_\tau^{n-1} \circ \T.$$
It also follows that $\S_\tau\circ \T \circ \S_\tau \circ \T = \T\circ \T = \id$.
Thus, since $f\circ g = \id$ implies that function $f$ is surjective and $g$ is injective, the relation $\S_\tau \circ (\T \circ \S_\tau \circ \T) = \id$ implies that $\S_\tau$ is surjective and $(\T \circ \S_\tau \circ \T) \circ \S_\tau = \id$ implies that $\S_\tau$ is injective.
\end{proof}

\renewcommand{\theprop}{\ref{prop:SPPD_invariance}}
\begin{prop}
  The extended target distribution with density $\Pi(x,v)$ is a stationary distribution for the Markov chain $\left( X^{(i)}, V^{(i)} \right)_{i\in1:M}$ constructed by the sequential-proposal MCMC algorithm using a deterministic kernel (Algorithm~\ref{alg:sppd}).
  Furthermore, the Markov chain $\left(X^{(i)}\right)_{i\in1:M}$ constructed by Algorithm~\ref{alg:sppd}, marginally for the $x$-component, is reversible with respect to the target distribution $\bar\pi(x)$.
\end{prop}
\addtocounter{prop}{-1}
\renewcommand\theprop\originaltheprop
\begin{proof}
  We will prove the claim for the case where $N$, $L$, and $\tau$ are fixed, since the general case can easily follow by considering a mixture over these parameters.
  In Algorithm~\ref{alg:sppd}, the $n$-th proposal $(Y_n,W_n)$ is obtained as $\S_\tau^n (Y_0,W_0)$, and if there are less than $L$ acceptable proposals in the first $N$ proposals, the next state of the Markov chain $(X^{(i+1)}, V^{(i+1)})$ is set to $(X^{(i)}, \R_{X^{(i)}} V^{(i)})$.
  Each iteration of Algorithm~\ref{alg:sppd} can thus be understood a composition of two operations, where the first operation is simply reflecting the velocity component from $(X^{(i)}, V^{(i)})$ to $(X^{(i)}, \R_{X^{(i)}} V^{(i)})$, and the second operation proposes a sequence of proposals $(Y_n,W_n) = \S_\tau^n \circ \T (X^{(i)}, \R_{X^{(i)}} V^{(i)})$ until $L$ acceptable proposals are found or until $N$ proposals have been made.
  The reason that we view the algorithm this way is to use the fact that both $\T$ and $\S_\tau^n \circ \T$ are self-inverse maps.
  If both the first and the second operations preserve $\Pi$ as an invariant density, the Markov chain constructed by Algorithm~\ref{alg:sppd} preserves $\Pi$ as an invariant density.
  In fact, we will show that both the first and the second operations satisfy detailed balance with respect to $\Pi$.
  However, we note that this does not imply that the constructed Markov chain satisfy detailed balance with respect to $\Pi$, because carrying out the first and then the second operation is not the same as carrying out the second and then the first.
  
  It is rather straightforward to see that velocity reflection operation $\T$ establishes detailed balance with respect to $\Pi$.
  Supposing that $(X,V)\sim \Pi$, we have for $A$, $B$ measurable in $\mathbb X\times \mathbb V$,
  \begin{equation*}
      \P[ (X,V)\sr\in A,~ (X,\R_X V)\sr\in B]
      = \int \1_A(x,v) \1_B(x,R_x v) \Pi(x,v) dx \, dv
  \end{equation*}
  Upon denoting $v' := R_x v$, we can express the right hand side as
  \begin{equation*}\begin{split}
      &\int \1_A(x, \R_x v') \1_B(x,v') \Pi(x, \R_x v') \left| \frac{\partial \R_x v'}{\partial v'} \right| dx \, dv' 
      = \int \1_A(x, \R_x v') \1_B(x, v') \Pi(x,v') dx \, dv' \\
      &= \P[(X,V)\in B, (X,\R_X V) \in A],
  \end{split}\end{equation*}
  where we have used the condition \eqref{eqn:psiR}.
  This shows that $\T$ establishes detailed balance with respect to $\Pi$.
   
  Now we will show that the second operation also establishes detailed balance with respect to $\Pi$.
  If the current state in the Markov chain $(X^{(i)}, V^{(i)})$ is denoted by $(Y_0,W_0)$, the second operation starts at $(Y_0, \R_{Y_0} W_0)$ since the velocity was reflected by the first operation.
  We will show that for arbitrary $n \sr\in 1\col N$ and for measurable subsets $A$, $B$ of $\mathbb X \times \mathbb V$, 
  \begin{multline*}
    \P[(Y_0, \R_{Y_0} W_0) \sr\in A,~ (Y_n, W_n) \sr\in B,~ (Y_n,W_n) \text{ is the $L$-th acceptable proposal}]\\
    =\P[(Y_0, \R_{Y_0} W_0) \sr\in B,~ (Y_n, W_n) \sr\in A,~ (Y_n,W_n) \text{ is the $L$-th acceptable proposal}],
  \end{multline*}
  provided that $(Y_0, \R_{Y_0} W_0)$ is distributed according to $\Pi$.
  Then by combining the cases for $n \sr \in 1\col N$, we can establish detailed balance for the second operation.
  For notational convenience, for $(y_0, w_0) \sr \in \mathbb X \times \mathbb V$, we will write $(y_k, w_k) = \S_\tau^k (y_0, w_0)$ and $\bar w_k = \R_{y_k} w_k$ for $k \sr \in 0 \col n$.
  We will write
  \[
  p_k(y_0, \bar w_0) := \Pi\{\S_\tau^k\circ \T(y_0, \bar w_0)\} \left| \frac{\partial \S_\tau^k \circ \T(y_0, \bar w_0)}{\partial (y_0, \bar w_0)}\right| = \Pi(y_k,w_k) \left| \frac{\partial(y_k,w_k)}{\partial (y_0,\bar w_0)}\right|
  \]
  for $k \sr \in 0\col n$.
  Note that this definition leads to
  \[
  p_0(y_0, \bar w_0) = \Pi(y_0, w_0) \left| \frac{\partial (y_0,w_0)}{\partial(y_0,\bar w_0)} \right| = \Pi(y_0, \bar w_0)
  \]
  due to \eqref{eqn:psiR}.
  Also, since $\S_\tau^{n-k}\circ \T (y_k, \bar w_k) = (y_n, w_n)$ and $\S_\tau^{n-k}\circ \T$ is a self-inverse map, we have $S_\tau^{n-k}\circ T(y_n, w_n) = (y_k, \bar w_k)$.
  This leads to 
  \begin{equation} \begin{split}
      p_{n-k}(y_n,w_n) &= \Pi \{\S_\tau^{n-k}\circ \T(y_n, w_n)\} \left| \frac{\partial \S_\tau^{n-k} \circ \T(y_n,w_n)}{\partial (y_n, w_n)} \right|\\
      &= \Pi(y_k, \bar w_k) \left| \frac{\partial (y_k, \bar w_k)}{\partial (y_n, w_n)}\right|\\
      &= \Pi(y_k, w_k) \left| \frac{\partial (y_k,w_k)}{\partial(y_n,w_n)} \right| \qquad \text{due to \eqref{eqn:psiR}}\\
      &= p_k(y_0, \bar w_0) \left| \frac{\partial (y_0, \bar w_0)}{\partial(y_n, w_n)} \right|.
      \label{eqn:sppd_pnk}
  \end{split}\end{equation}
  The following steps are similar to corresponding steps in the proof of Proposition~\ref{prop:detailedbalanceMH}.
  We have
  \begin{equation}\begin{split}
      &\P[(Y_0, \R_{Y_0} W_0) \in A,~ (Y_n, W_n) \in B,~ (Y_n,W_n) \text{ is the $L$-th acceptable proposal}]\\
      &= \int \1_A(y_0,\bar w_0) \1_B(y_n,w_n) \Pi(y_0, \bar w_0) \1\left[\Lambda \geq r_L\left\{\frac{p_{1:n-1}}{p_0}(y_0,\bar w_0)\right\}\right] \1\left[\Lambda < r_{L-1}\left\{\frac{p_{1:n-1}}{p_0}(y_0,\bar w_0)\right\}\right] \\
      &\hspace{71ex}\cdot \1\left[ \Lambda < \frac{p_n}{p_0} \land 1 \right] d\Lambda \, dy_0 \, d\bar w_0\\
      &= \int \1_A(y_0, \bar w_0) \1_B(y_n,w_n) p_0 \left\{ \frac{p_n}{p_0} \land 1 \land r_{L-1}\left(\frac{p_{1:n-1}}{p_0}\right) - \frac{p_n}{p_0} \land 1 \land r_L\left( \frac{p_{1:n-1}}{p_0}\right) \right\}(y_0,\bar w_0) dy_0\, d\bar w_0,
      \label{eqn:sppd_PAB}
  \end{split}\end{equation}
  where all functions $p_k$, $k\sr\in 0\col n$, in the above display take the argument $(y_0, \bar w_0)$.
  Using \eqref{eqn:sppd_pnk}, the above equation is equal to
  \[ \begin{split}
    &\int \1_A(y_0,\bar w_0) \1_B(y_n, w_n) \left\{p_0 \land p_n \land r_{L-1}(p_{n-1:1}) - p_0 \land p_n \land r_L(p_{n-1:1}) \right\} (y_n,w_n) \left| \frac{\partial (y_n,w_n)}{\partial(y_0,\bar w_0)}\right| dy_0\, d\bar w_0\\
    &= \int \1_A(y_0, \bar w_0) \1_B(y_n, w_n) \left\{ p_0 \land p_n \land r_{L-1}(p_{n-1:1}) - p_0 \land p_n \land r_L(p_{n-1:1})\right\} (y_n, w_n) dy_n \, dw_n
  \end{split}\]
  We change the dummy variables by writing $(y_0, \bar w_0) \gets (y_n, w_n)$.
  Since $\S_\tau^n \circ \T(y_0,\bar w_0) = (y_n, w_n)$, we can also write $(y_n, w_n) \gets (y_0, \bar w_0)$.
  The above display can be re-written as 
  \[
  \int \1_A(y_n,w_n) \1_B(y_0, \bar w_0) \left\{ p_0 \land p_n \land r_{L-1}\left(p_{n-1:1}\right) - p_0 \land p_n \land r_L\left( p_{n-1:1} \right) \right\}(y_0,\bar w_0) dy_0\, d\bar w_0,
  \]
  which is equal to \eqref{eqn:sppd_PAB} where the sets $A$ and $B$ are interchanged.
  Thus we have proved 
  \begin{multline*}
    \P[(Y_0, \R_{Y_0} W_0) \in A,~ (Y_n, W_n) \in B,~ (Y_n,W_n) \text{ is the $L$-th acceptable proposal}]\\
    =\P[(Y_0, \R_{Y_0} W_0) \in B,~ (Y_n, W_n) \in A,~ (Y_n,W_n) \text{ is the $L$-th acceptable proposal}].
  \end{multline*}
  By adding the cases for $n\sr\in 1\col N$, we can conclude the proof of detailed balance for the second operation with respect to $\Pi$.
  Since both the first and the second operations preserves $\Pi$ as an invariant density, Algorithm~\ref{alg:sppd} preserves $\Pi$ as an invariant density.
  Finally, refreshing the velocity $V^{(i+1)}$ from $\psi(\cdot \giventh X^{(i+1)})$ at the end of the iteration with an arbitrary probability $p^\text{ref}(X^{(i+1)})$ clearly preserves the invariant density $\Pi(x,v) = \bar\pi(x) \psi(v \giventh x)$.

  In order to prove the claim that the marginally for the $x$-component, the Markov chain $\left(X^{(i)}\right)_{i\in1:M}$ constructed by Algorithm~\ref{alg:sppd} is reversible with respect to the target distribution $\bar\pi(x)$, we denote the position-velocity pair taken as the next state of the Markov chain at the end of the second operation by $(Y',W')$.
  We showed above that when $(Y_0, W_0)$ is drawn from $\Pi$, $(Y_0, \R_{Y_0} W_0)$ is also distributed according to $\Pi$.
  Due to the fact that the second operation satisfies detailed balance with respect to $\Pi$, we see that for measurable subsets $A$, $B$ of $\mathbb X$,
  \[\begin{split}
  \P[ Y_0 \in A,~ Y'\in B]  &= \P[ (Y_0, \R_{Y_0} W_0) \in A \sr\times \mathbb V, ~ (Y', W') \in B \sr\times \mathbb V ]\\
  &=\P [ (Y_0, \R_{Y_0} W_0) \in B \sr\times \mathbb V, ~ (Y', W') \in A \sr\times \mathbb V] \\
  &=\P [ Y_0 \in B, Y' \in A ].
  \end{split}\]
  This shows that the Markov chain $\left(X^{(i)})\right)_{i\in1:M}$ constructed by Algorithm~\ref{alg:sppd} is reversible with respect to $\bar\pi$, which is the marginal distribution of $\Pi$ for the $x$-component.
\end{proof}

\section{Proofs of detailed balance for sequential-proposal No-U-Turn samplers (spNUTS1 and spNUTS2)}\label{sec:proof_spNUTS}
\renewcommand{\theprop}{\ref{prop:detailedbalance_spNUTS1}}
We prove that both spNUTS1 and spNUTS2 algorithms (Algorithms~\ref{alg:spNUTS1} and \ref{alg:spNUTS2}) construct reversible Markov chains with respect to the target distribution $\bar\pi$.
\begin{prop} 
  The Markov chain $\left(X^{(i)}\right)_{i\in1:M}$ constructed by the sequential-proposal No-U-Turn sampler of type 1 (spNUTS1, Algorithm~\ref{alg:spNUTS1}) is reversible with respect to the target distribution $\bar\pi$.
\end{prop}
\addtocounter{prop}{-1}
\renewcommand\theprop\originaltheprop
\begin{proof}
  We assume that the cosine value $c$ at which trajectory extensions in Algorithm~\ref{alg:spNUTS1} stop is fixed, as the general case readily follows by considering a mixture over $c$.
  The state of the Markov chain constructed by the algorithm in the current iteration is denoted by $Y_0$, and assumed to be distributed according to $\bar\pi$.
  The velocity drawn from $\psi_C$ at the start of the iteration is denoted by $W_0$.
  For $k\sr\geq 1$, the leapfrog trajectory starting from $(Y_{k-1},W_{k-1})$ stops at $(Y_k, W'_k)$, and the function that maps the initial position-velocity pair to the final pair will be denoted by $\S$, such that $(Y_k,W'_k) = \S(Y_{k-1},W_{k-1})$.
  We will show that for $n \sr \in 1\col N$ and for measurable subsets $A$ and $B$ of $\mathbb X\sr\times \mathbb V$,
  \begin{multline}
  \P[(Y_0,W_0) \sr\in A, ~(Y_n, -W'_n) \sr\in B, ~ Y_n\text{ is taken as the next state of the Markov chain}]\\
  = \P[(Y_0,W_0) \sr\in B, ~(Y_n, -W'_n) \sr\in A, ~ Y_n\text{ is taken as the next state of the Markov chain}].
  \label{eqn:spNUTS1proof_detailedbalance}
  \end{multline}
  Then, by considering the cases $A=A_0\sr\times\mathbb V$ and $B=B_0\sr\times\mathbb V$ for some $A_0, B_0 \subset \mathbb X$ and summing over $n\sr \in 1\col N$, we show that the Markov chain constructed by Algorithm~\ref{alg:spNUTS1} is reversible with respect to $\bar\pi$.

  When $(Y_k,W'_k)$ is rejected, the velocity is refreshed by drawing $U_k \sim \psi_C$ and setting
  \[
  W_k = U_k \frac{\Vert W'_k \Vert_C}{\Vert U_k \Vert_C}.
  \]
  In following equations, $y_{1:n}$, $w_{1:n-1}$, and $w'_{1:n}$ will denote functions of $y_0$, $w_0$, and $u_{1:n-1}$ defined recursively by
  \begin{gather}
    \S(y_{k-1},w_{k-1}) = (y_k, w'_k), \qquad\qquad w_k = u_k \frac{ \Vert w'_k \Vert_C } { \Vert u_k \Vert_C }.
    \label{eqn:spNUTS1_dummy}
  \end{gather}
  We also denote $\bar w_k := -w_k$ and $\bar w'_k := -w'_k$.
  The left hand side of \eqref{eqn:spNUTS1proof_detailedbalance} is then given by
  \begin{equation}\begin{split}
  &\int \1_A(y_0, w_0) \1_B(y_n,\bar w'_n) \frac{1}{Z} \pi(y_0) \psi_C(w_0) \1\left[\Lambda < \frac{\pi(y_n)\psi_C(w'_n)}{\pi(y_0)\psi_C(w_0)} \right] \1\left[\Lambda \geq \max_{k\in 1:n-1} \frac{\pi(y_k)\psi_C(w'_k)}{\pi(y_0)\psi_C(w_0)} \right] \\
  &\hspace{53ex}\cdot \1\left[0<\Lambda<1\right] \prod_{k=1}^{n-1} \psi_C(u_k) d\Lambda\, dy_0\, dw_0 \prod_{k=1}^{n-1} du_k\\
  &= \frac{1}{Z} \int \1_A(y_0,w_0) \1_B(y_n, \bar w'_n) \big[ \pi(y_n) \psi_C(w'_n) \land \pi(y_0) \psi_C(w_0) \\
  &\hspace{16ex} - \pi(y_n) \psi_C(w'_n) \land \pi(y_0) \psi_C(w_0) \land \max_{k\in 1:n-1} \pi(y_k) \psi_C(w'_k) \big] \prod_{k=1}^{n-1} \psi_C(u_k) dy_0 \, dw_0 \prod_{k=1}^{n-1} du_k.
  \end{split}
  \label{eqn:spNUTS1proof_prob1}\end{equation}
  To establish a symmetric relationship between $(Y_0, W_0)$ and $(Y_n, \bar W'_n)$, we define 
  \[
  u'_k := w'_k \frac{\Vert u_k \Vert_C}{\Vert w'_k \Vert_C}
  \]
  and write $\bar u'_k := -u'_k$ for $k \sr\in 1\col n{-}1$.
  Due to the symmetric nature of the stopping condition \eqref{eqn:stopping_cond_spNUTS1}, we have
  \[
  \S(y_k, \bar w'_k) = (y_{k-1}, \bar w_{k-1}).
  \]
  It is also readily observed from the definition of $u'_k$ that
  \[
  \bar w'_k = \bar u'_k \frac{\Vert \bar w_k \Vert_C}{\Vert \bar u'_k \Vert_C}.
  \]
  The two relations above form a counterpart to \eqref{eqn:spNUTS1_dummy} in a symmetric relationship between $(y_k, w_k)$ and $(y_{n-k}, \bar w'_{n-k})$ for $k \sr \in 0\col n$.
  Furthermore, since $\psi_C(v)$ is a function of $\Vert v\Vert_C$, we have $\psi_C(u_k) = \psi_C(\bar u'_k)$ and $\psi_C(w'_k) = \psi_C(\bar w_k)$.
  Finally, we use the equation
  \[
  \left| \frac{\partial (w_k, u'_k)}{\partial (w'_k, u_k)} \right| = 1,
  \]
  which is stated as Lemma~\ref{lem:spNUTS1_volelement} and proved below.
  This, together with the fact that leapfrog jumps preserve the volume element, leads to
  \begin{gather*}
    dy_0 \, dw_0 \prod_{k=1}^{n-1} du_k = dy_1 \, dw'_1 \,du_1 \prod_{k=2}^{n-1} du_k = dy_1 \, dw_1 \, du'_1 \, \prod_{k=2}^{n-1} du_k = dy_2 \, dw'_2 \, du'_1 \, du_2\, \prod_{k=3}^{n-1} du_k \\
    = dy_2 \, dw_2 \, du'_1 \, du'_2 \prod_{k=3}^{n-1} du_k = \cdots = dy_n \, dw_n \prod_{k=1}^{n-1} du'_k.
  \end{gather*}
  Thus, \eqref{eqn:spNUTS1proof_prob1} is equal to
  \[\begin{split}
  &\frac{1}{Z} \int \1_A(y_0,w_0) \1_B(y_n, \bar w'_n) \big[ \pi(y_n) \psi_C(\bar w'_n) \land \pi(y_0) \psi_C(\bar w_0) \\
  &\hspace{16ex} - \pi(y_n) \psi_C(\bar w'_n) \land \pi(y_0) \psi_C(\bar w_0) \land \max_{k\in 1:n-1} \pi(y_k) \psi_C(\bar w_k) \big] \prod_{k=1}^{n-1} \psi_C(\bar u'_k) dy_n \, d\bar w_n \prod_{k=1}^{n-1} d\bar u'_k,
  \end{split}\]
  which, under the change of notation of dummy variables $(y_0, w_0) \gets (y_n, \bar w'_n)$, becomes
  \[\begin{split}
  &\frac{1}{Z} \int \1_A(y_n,\bar w'_n) \1_B(y_0, w_0) \big[ \pi(y_0) \psi_C(w_0) \land \pi(y_n) \psi_C(w'_n) \\
  &\hspace{10ex} - \pi(y_0) \psi_C(w_0) \land \pi(y_n) \psi_C(w'_n) \land \max_{k\in 1:n-1} \pi(y_{n-k}) \psi_C(w'_{n-k}) \big] \prod_{k=1}^{n-1} \psi_C(u_{n-k}) dy_0 \, d w_0 \prod_{k=1}^{n-1} du_k.
  \end{split}\]
  Since the above expression equals the right hand side of \eqref{eqn:spNUTS1proof_detailedbalance}, the claim of detailed balance is proved.
\end{proof}

\begin{lemma}\label{lem:spNUTS1_volelement}
  Let $C$ be a positive definite symmetric matrix in $\mathbb R^{d\times d}$.
  Given $(w', u)$ in $\mathbb R^{2d}$, define $u' := w' \frac{\Vert u\Vert_C}{\Vert w' \Vert_C}$ and $w := u \frac{\Vert w' \Vert_C}{\Vert u\Vert_C}$.
  Then we have
  \[
  \left| \frac{\partial(u', w)}{\partial(w', u)} \right| \equiv 1.
  \]
\end{lemma}
\begin{proof}
  It is sufficient to prove the claim for $C=I$, the identity matrix.
  To see this, we denote $\tilde w := C^{-1/2}w$, $\tilde w' := C^{-1/2}w'$, $\tilde u := C^{-1/2}u$, $\tilde u' := C^{-1/2}u'$, and denote the Euclidean norm as $\Vert v \Vert := \sqrt{v^T v}$.
  Then since $\Vert \tilde w \Vert^2 = w^T C^{-1} w = \Vert w \Vert_C^2$ and the same kind of relation holds for the other three variables, we have
  \[
  \tilde u' = C^{-1/2} u' = C^{-1/2} w' \frac{\Vert \tilde u \Vert}{\Vert \tilde w' \Vert} = \tilde w' \frac{\Vert \tilde u \Vert}{\Vert \tilde w \Vert}
  \]
  and similarly $\tilde w = \tilde u \frac{\Vert \tilde w'\Vert}{\Vert \tilde u \Vert}$.
  But we also have
  \[
  \left| \frac{\partial (u',w)}{\partial (w',u)} \right|
  = \left| \frac{\partial (C^{1/2}\tilde u', C^{1/2}\tilde w)}{\partial (C^{1/2}\tilde w', C^{12}\tilde u)} \right|
  = \left| \left( \begin{array}{cc} C^{1/2} & 0 \\ 0 & C^{1/2} \end{array} \right) \frac{\partial (\tilde u',\tilde w)}{\partial (\tilde w',\tilde u)} \left( \begin{array}{cc} C^{1/2} & 0 \\ 0 & C^{1/2} \end{array} \right)^{-1} \right|
  = \left| \frac{\partial (\tilde u',\tilde w)}{\partial (\tilde w',\tilde u)} \right|,
  \]
  from which we see that it is sufficient to prove that the rightmost term is equal to unity.
  
  Now we will assume $C=I$.
  Computing partial derivatives yields
  \[
  \frac{\partial(u',w)}{\partial(w',u)}
  = \left( \begin{array}{cc} \frac{\Vert u \Vert}{\Vert w'\Vert}I - \frac{w'w'^T\Vert u\Vert}{\Vert w'\Vert^3} & \frac{w' u^T}{\Vert w'\Vert \Vert u \Vert}\\
    \frac{u w'^T}{\Vert u \Vert \Vert w'\Vert} & \frac{\Vert w'\Vert}{\Vert u \Vert} I - \frac{u u^T \Vert w'\Vert}{\Vert u \Vert^3} \end{array} \right).
  \]
  We carry out elementary column and row operations as follows to obtain
  \[\begin{split}
  &\left( \begin{array}{cccccc} 1&&&\frac{w'_1}{u_1}\frac{\Vert u\Vert^2}{\Vert w'\Vert^2}&&\\
    &\ddots&&&\ddots&\\ &&1&&& \frac{w'_d}{u_d}\frac{\Vert u\Vert^2}{\Vert w'\Vert^2}\\
    &&&1&&\\ &&&&\ddots& \\ &&&&&1 \end{array} \right)
  \frac{\partial (u',w)}{\partial (w',u)}
  \left( \begin{array}{cccccc} 1&&&\frac{u_1}{w'_1}\frac{\Vert w'\Vert^2}{\Vert u \Vert^2}&&\\
    &\ddots&&&\ddots&\\ &&1&&& \frac{u_d}{w'_d}\frac{\Vert w'\Vert^2}{\Vert u\Vert^2}\\
    &&&1&&\\ &&&&\ddots& \\ &&&&&1 \end{array} \right)\\
  &\hspace{20ex} \cdot \left( \begin{array}{ccccccc} 1& -\frac{w'_2}{w'_1}& \cdots& -\frac{w'_d}{w'_1}&&&\\ &1&&&&&\\ &&\ddots&&&&\\ &&&1&&&\\ &&&&1&&\\ &&&&&\ddots& \\ &&&&&&1\end{array}\right) 
    \left( \begin{array}{cccccc} 1&&&&&\\ &\ddots&&&&\\ &&1&&&\\ -\frac{u_1w'_1}{\Vert w'\Vert^2}&&&1&&\\\vdots&&&&\ddots&\\ -\frac{u_d w'_1}{\Vert w'\Vert^2}&&&&&1\end{array} \right)
  \end{split}\]
  \[
  = \left( \begin{array}{cccccccc} \frac{\Vert u \Vert}{\Vert w' \Vert}\left(1-\frac{u_1^2}{\Vert u\Vert^2}-\frac{{w'_1}^2}{\Vert w'\Vert^2}\right) & -\frac{\Vert u \Vert}{\Vert w'\Vert} \frac{w'_2}{w'_1} & \cdots & -\frac{\Vert u \Vert}{\Vert w'\Vert} \frac{w'_d}{w'_1}&\frac{u_1}{w'_1}\frac{\Vert w'\Vert}{\Vert u \Vert} + \frac{w'_1}{u_1}\frac{\Vert u \Vert}{\Vert w'\Vert}&&&\\
   -\frac{\Vert u \Vert}{\Vert w'\Vert}\frac{w'_1}{w'_2}\left(\frac{u_2^2}{\Vert u\Vert^2}+\frac{{w'_2}^2}{\Vert w'\Vert^2}\right) & \frac{\Vert u \Vert}{\Vert w'\Vert}&&&&\ddots&&\\
   \vdots&&\ddots&&&&\ddots&\\
   -\frac{\Vert u \Vert}{\Vert w'\Vert}\frac{w'_1}{w'_d}\left(\frac{u_d^2}{\Vert u\Vert^2}+\frac{{w'_d}^2}{\Vert w'\Vert^2}\right) &&& \frac{\Vert u \Vert}{\Vert w'\Vert}&&&&\frac{u_d}{w'_d}\frac{\Vert w'\Vert}{\Vert u\Vert}+\frac{w'_d}{u_d}\frac{\Vert u\Vert}{\Vert w'\Vert}\\
   &&&&\frac{\Vert w'\Vert}{\Vert u \Vert} &&&\\
   &&&&&\ddots&&\\ &&&&&&\ddots&\\ &&&&&&&\frac{\Vert w'\Vert}{\Vert u \Vert}
      \end{array} \right)
  \]
  The absolute value of the determinant of the above matrix can be directly computed as
  \begin{multline*}
    \left| \left(1-\frac{u_1^2}{\Vert u\Vert^2}-\frac{{w'_1}^2}{\Vert w'\Vert^2}\right) - \left(\frac{u_2^2}{\Vert u\Vert^2}+\frac{{w'_2}^2}{\Vert w'\Vert^2}\right) - \cdots - \left(\frac{u_d^2}{\Vert u\Vert^2} + \frac{{w'_d}^2}{\Vert w'\Vert^2}\right) \right|\\
    =\left| 1 - \left(\frac{u_1^2}{\Vert u\Vert^2}+\cdots +\frac{u_d^2}{\Vert u\Vert^2}\right) - \left(\frac{{w'_1}^2}{\Vert w'\Vert^2} + \cdots + \frac{{w'_d}^2}{\Vert w'\Vert^2}\right) \right| = 1.
  \end{multline*}
\end{proof}

\renewcommand{\theprop}{\ref{prop:detailedbalance_spNUTS2}}
\begin{prop} 
  The Markov chain $\left(X^{(i)}\right)_{i\in1:M}$ constructed by the sequential-proposal No-U-Turn sampler of type 2 (spNUTS2, Algorithm~\ref{alg:spNUTS2}) is reversible with respect to the target distribution $\bar\pi(x)$.
\end{prop}
\addtocounter{prop}{-1}
\renewcommand\theprop\originaltheprop
\begin{figure}[t]
\centering
\includegraphics[width=.45\textwidth]{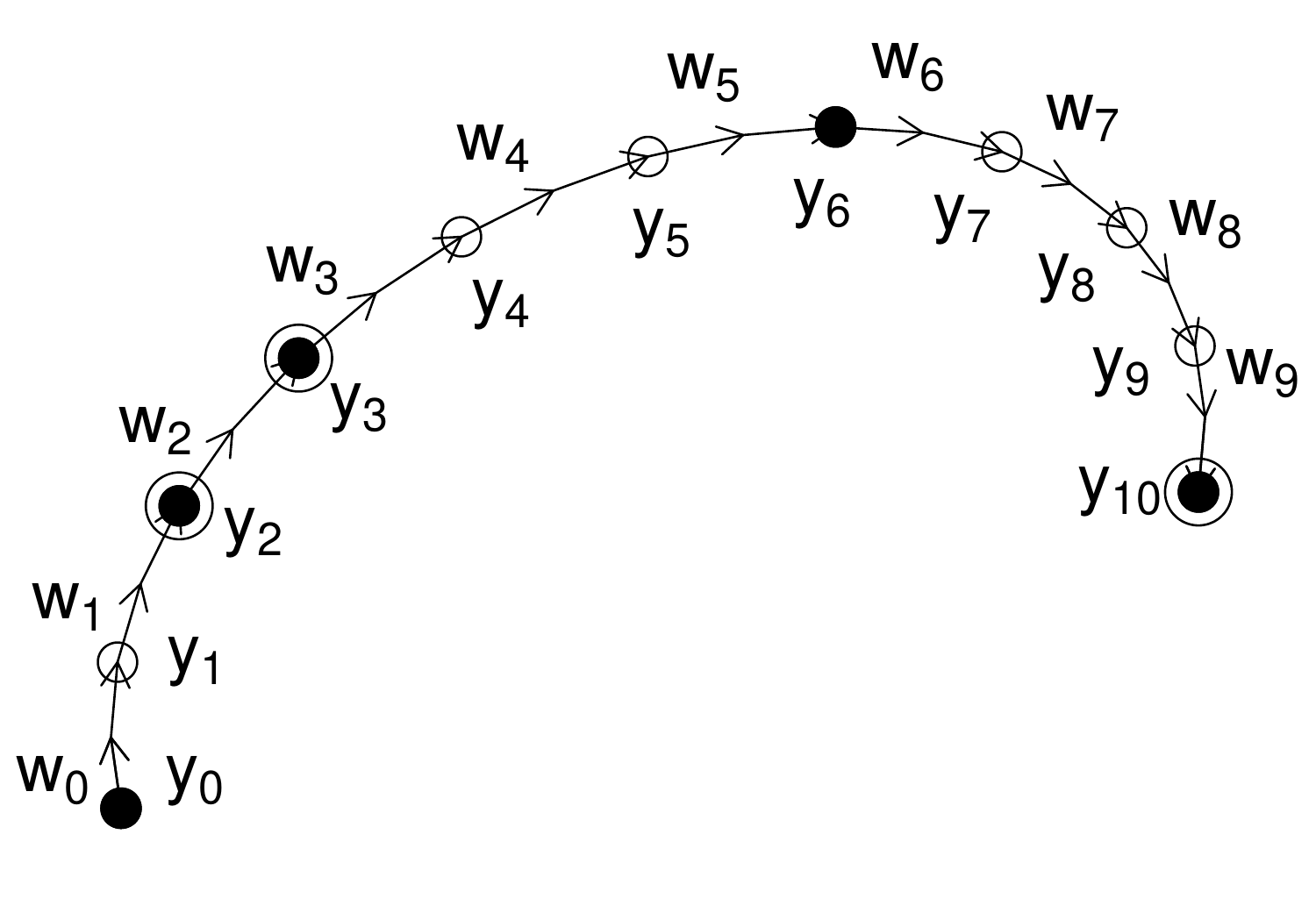}
\caption{A schematic diagram showing the variables defined in the proof of Proposition~\ref{prop:detailedbalance_spNUTS2}.
  Here $(y_k,w_k)$ for $k\sr\geq1$ are obtained by making two leapfrog jumps from $(y_{k-1},w_{k-1})$ (i.e., $l\sr=2$ in Algorithm~\ref{alg:spNUTS2}).
  Acceptable states are denoted by filled circles, and the states at which the stopping condition is checked are additionally marked by an encompassing larger circle.
  The diagram illustrates the case where $b_j \sr= 2^{j-1}$.
  This trajectory stops at $y_{10}$, which is the fourth ($b_3\sr=4$) acceptable state and the first acceptable state after making a U-turn.
}
\label{fig:spNUTS2proof}
\end{figure}
\begin{proof}
We assume the cosine angle $c$ is fixed.
Let $Y_0$ denote the current state of the Markov chain constructed by Algorithm~\ref{alg:spNUTS2} and $W_0$ the velocity drawn from $\psi_C$ at the start of the current iteration.
We recursively let $(Y_k,W_k)$ denote the state reached after $l$ leapfrog jumps starting from $(Y_{k-1},W_{k-1})$, for $k \sr\geq 1$.
We consider the case where the trajectory stops at the $b_j$-th acceptable state which is equal to $(Y_{k^*}, W_{k^*})$ for some $k^*$.
We will consider the probability
\begin{multline*}
  \P[ (Y_0,W_0) \in A,~ (Y_{k^*},-W_{k^*})\in B, \\
    (Y_{k^*},W_{k^*})\text{ is the $b_j$-th acceptable state and taken as the next state of the Markov chain}]
\end{multline*}
for measurable subsets $A$ and $B$ of $\mathbb X \sr\times\mathbb V$.
Let $\mathcal K$ be a subset of $0\col k^*$ and write $\mathcal K^c = (0\col k^*)\sr\setminus \mathcal K$.
Given the values $(Y_k,W_k)=(y_k,w_k)$, $k\sr\in 0\col k^*$ and the uniform$(0,1)$ random variable $\Lambda$ drawn at the start of the current iteration, the states $\{(y_k,w_k) \giventh k\sr\in \mathcal K\}$ are deemed acceptable and $\{(y_k,w_k) \giventh k\in\mathcal K^c\}$ not acceptable if and only if the quantity
\[
\1\left[ \Lambda \geq \max_{k\in \mathcal K^c} \frac{\pi(y_k)\psi_C(w_k)}{\pi(y_0)\psi_C(w_0)} \right]
\cdot \1 \left[ \Lambda < \min_{k\in \mathcal K\sr\setminus \{0\}} \frac{\pi(y_k)\psi_C(w_k)}{\pi(y_0) \psi_C(w_0)} \right]
\cdot \1 \left[ 0 < \Lambda < 1\right]
\]
is equal to unity.
Note that we consider $(y_0,w_0)$ as an acceptable state here.
Integrating the above quantity over $\Lambda$, we see that the probability of finding $\{(y_k,w_k)\giventh k\in \mathcal K\}$ acceptable and $\{(y_k,w_k) \giventh k\in \mathcal K^c\}$ not acceptable is given by
\[
\left[ 1\land \min_{k\in\mathcal K\sr\setminus\{0\}}\frac{\pi(y_k)\psi_C(w_k)}{\pi(y_0)\psi_C(w_0)}\right]
- \left[ 1\land \min_{k\in\mathcal K\sr\setminus\{0\}}\frac{\pi(y_k)\psi_C(w_k)}{\pi(y_0)\psi_C(w_0)} \land \max_{k\in\mathcal K^c}\frac{\pi(y_k)\psi_C(w_k)}{\pi(y_0)\psi_C(w_0)} \right].
\]
We will consider $\mathcal K$ that satisfies the following three conditions: $\{0,k^*\} \subset \mathcal K \subset 0\col k^*$, $|\mathcal K| = b_j\sr+1$, and $\mathcal K(i) - \mathcal K(i{\,-\,}1) \leq N$ for all $i\in 1\col b_j$, where the elements of $\mathcal K$ are ordered as $0= \mathcal K(0) < \mathcal K(1) < \cdots < \mathcal K(b_j)=k^*$.
The last condition is related to the fact that spNUTS2 tries at most $N$ consecutive states to find each new acceptable state.

Given $\{(y_k,w_k)\giventh k\sr\in 0\col k^*\}$ and $\mathcal K$, we let $\Upsilon\big(\{(y_k,w_k) \giventh k \sr\in 0\col k^*\}, \mathcal K \big)$ denote the indicator function that takes the value of unity if and only if the pair $(y_{k^*},w_{k^*})$ is the first position-velocity pair satisfying the U-turn condition among $\{(y_{\mathcal K(b_{j'})}, w_{\mathcal K(b_{j'})}) \giventh {j'} \sr\in 0\col j\}$ and the stopped trajectory satisfies the symmetry condition.
That is, we define
\begin{multline*}
  \Upsilon\big(\{(y_k,w_k) \giventh k\in 0\col k^*\}, \mathcal K\big)
  := \1\left[ \cosangle(y_{k^*}\sr-y_{0}, w_0) \leq c \text{ or } \cosangle(y_{k^*}\sr-y_{0}, w_{k^*}) \leq c \right] \\
  \cdot \prod_{{j'}=0}^{j-1} \Big\{ \1\left[ \cosangle(y_{\mathcal K(b_{j'})}\sr-y_0, w_0) >c\right] \cdot \1 \left[\cosangle(y_{\mathcal K(b_{j'})}\sr-y_0, w_{\mathcal K(b_{j'})}) >0 \right] \\
  \cdot \1 \left[ \cosangle(y_{k^*} \sr- y_{\mathcal K(b_j-b_{j'})}, w_{\mathcal K(b_j-b_{j'})} ) >c \right] \cdot \1 \left[\cosangle(y_{k^*} \sr- y_{\mathcal K(b_j-b_{j'})}, w_{k^*}) > c \right] \Big\},
\end{multline*}
where the dependence of $\cosangle$ on $C$ is suppressed.
Thus the probability of drawing $(y_0,w_0)$ and taking the $b_j$-th acceptable pair $(y_{k^*},w_{k^*})$ as the next state of the Markov chain while finding $\{(y_k,w_k)\giventh k \sr\in \mathcal K\}$ acceptable and $\{(y_k,w_k)\giventh k \sr\in \mathcal K^c\}$ not acceptable is given by
\begin{multline*}
\frac{1}{Z}\pi(y_0)\psi_C(w_0) \left( \left[ 1\land \min_{k\in\mathcal K\sr\setminus\{0\}}\frac{\pi(y_k)\psi_C(w_k)}{\pi(y_0)\psi_C(w_0)}\right] \right. \\
\left. - \left[ 1\land \min_{k\in\mathcal K\sr\setminus\{0\}}\frac{\pi(y_k)\psi_C(w_k)}{\pi(y_0)\psi_C(w_0)} \land \max_{k\in\mathcal K^c}\frac{\pi(y_k)\psi_C(w_k)}{\pi(y_0)\psi_C(w_0)} \right] \right)
\cdot \Upsilon\big(\{(y_k,w_k)\giventh k\in 0\col k^*\}, \mathcal K\big) dy_0 dw_0
\end{multline*}
\vspace{-3ex}
\begin{multline}
=\frac{1}{Z} \left( \min_{k\in\mathcal K}\{\pi(y_k)\psi_C(w_k)\} - \left[\min_{k\in\mathcal K}\{\pi(y_k)\psi_C(w_k)\} \land \max_{k\in\mathcal K^c}\{\pi(y_k)\psi_C(w_k)\} \right] \right)\\
\cdot \Upsilon\big(\{(y_k,w_k)\giventh k\in 0\col k^*\}, \mathcal K\big) dy_0 dw_0.
\label{eqn:spNUTS2_p}
\end{multline}
We now consider a reverse scenario where we draw $y_{k^*}$ from $\bar\pi$ and $-w_{k^*}$ from $\psi_C$ at the start of the current iteration of the algorithm.
The new leapfrog trajectory exactly reverses the original trajectory and is given by $\{(y_{k^*-k}, -w_{k^*-k}) \giventh k \in 0\col k^*\}$.
We denote $\bar y_k := y_{k^*-k}$ and $\bar w_k := -w_{k^*-k}$ for $k \in 0\col k^*$.
We also write $\bar{\mathcal K} := \{ k^*-k \giventh k\in\mathcal K\}$ and $\bar{\mathcal K}^c := (0\col k^*) \sr\setminus \bar{\mathcal K}$.
The elements of $\bar{\mathcal K}$ will also be denoted as $0=\bar{\mathcal K}(0) < \bar{\mathcal K}(1) < \cdots < \bar{\mathcal K}(b_j) = k^*$.
It is easy to see that $\bar{\mathcal K}(i) = k^* - \mathcal K(b_j-i)$ for $i\in 0\col b_j$.
It follows that
\[
\bar y_{\bar{\mathcal K}(b_{j'})} = y_{k^* - \bar{\mathcal K}(b_{j'})} = y_{\mathcal K(b_j-b_{j'})}, \qquad
\bar w_{\bar{\mathcal K}(b_{j'})} = -w_{\mathcal K(b_j-b_{j'})}.
\]
From this, the following equations hold:
\[\begin{split}
\cosangle(\bar y_{\bar{\mathcal K}(b_{j'})} \sr- \bar y_0, \bar w_{\bar{\mathcal K}(b_{j'})})
&= \cosangle(y_{k^*} - y_{\mathcal K(b_j-b_{j'})}, w_{\mathcal K(b_j-b_{j'})})\\
\cosangle(\bar y_{\bar{\mathcal K}(b_{j'})} \sr- \bar y_0, \bar w_0)
&= \cosangle(y_{k^*} \sr- y_{\mathcal K(b_j-b_{j'})}, w_{k^*})\\
\cosangle(\bar y_{k^*} \sr- \bar y_{\bar{\mathcal K}(b_j-b_{j'})}, \bar w_{\bar{\mathcal K}(b_j-b_{j'})})
&= \cosangle(y_{\mathcal K(b_{j'})} \sr- y_0, w_{\mathcal K(b_{j'})})\\
\cosangle(\bar y_{k^*} \sr- \bar y_{\bar{\mathcal K}(b_j-b_{j'})}, \bar w_{k^*})
&= \cosangle(y_{\mathcal K(b_{j'})} \sr- y_0, w_0).
\end{split}\]
Thus we see
\[
\Upsilon (\{(\bar y_k, \bar w_k) \giventh k \in 0\col k^*\}, \bar{\mathcal K})
= \Upsilon (\{(y_k, w_k) \giventh k\in 0\col k^*\}, \mathcal K).
\]
Finally, due to the measure-preserving property of leapfrog jumps, we have $dy_0\,dw_0 = dy_{k^*} \, dw_{k^*}$.
From these facts, we see that the probability of drawing $\bar y_0$ from $\bar\pi$, drawing $\bar w_0$ from $\psi_C$, finding the states $\{(\bar y_k, \bar w_k)\giventh k \in \bar {\mathcal K}\}$ acceptable and $\{(\bar y_k, \bar w_k) \giventh k \in \bar {\mathcal K}^c\}$ not acceptable, and taking $(\bar y_{k^*}, \bar w_{k^*})$ as the next state of the Markov chain equals
\begin{multline*}
\frac{1}{Z} \left( \min_{k\in\bar{\mathcal K}}\{\pi(\bar y_k)\psi_C(\bar w_k)\} - \left[\min_{k\in\bar{\mathcal K}}\{\pi(\bar y_k)\psi_C(\bar w_k)\} \land \max_{k\in\bar{\mathcal K}^c}\{\pi(\bar y_k)\psi_C(\bar w_k)\} \right] \right)\\
\cdot \Upsilon\big(\{(\bar y_k,\bar w_k)\giventh k\in 0\col k^*\}, \bar{\mathcal K}\big) d\bar y_0\, d\bar w_0
\end{multline*}
\vspace{-2em}
\begin{multline*}
= \frac{1}{Z} \left( \min_{k\in\bar{\mathcal K}}\{\pi(y_{k^*-k})\psi_C(-w_{k^*-k})\} - \left[\min_{k\in\bar{\mathcal K}}\{\pi(y_{k^*-k})\psi_C(-w_{k^*-k})\} \land \max_{k\in\bar{\mathcal K}^c}\{\pi(y_{k^*-k})\psi_C(-w_{k^*-k})\} \right] \right)\\
\cdot \Upsilon\big(\{(y_k,w_k)\giventh k\in 0\col k^*\}, \mathcal K\big) dy_{k^*} dw_{k^*}
\end{multline*}
\vspace{-2em}
\begin{multline}
= \frac{1}{Z} \left( \min_{k\in \mathcal K}\{\pi(y_k)\psi_C(w_k)\} - \left[\min_{k\in\mathcal K}\{\pi(y_k)\psi_C(w_k)\} \land \max_{k\in\mathcal K^c}\{\pi(y_k)\psi_C(w_k)\} \right] \right)\\
\cdot \Upsilon\big(\{(y_k,w_k)\giventh k\in 0\col k^*\}, \mathcal K\big) dy_0\, dw_0,
\label{eqn:spNUTS2_prev}
\end{multline}
which is the same as \eqref{eqn:spNUTS2_p}.
Recall that \eqref{eqn:spNUTS2_p} gives the probability corresponding to the case where the index set of the acceptable states is given by $\mathcal K$.
Let $\mathbb K(j,k^*)$ be the set of index sets $\mathcal K$ that satisfy the following three conditions: $\{0,k^*\} \subset \mathcal K \subset 0\col k^*$, $|\mathcal K| = b_j+1$, and $\mathcal K(i) - \mathcal K(i{\,-\,}1) \leq N$ for all $i\in 1\col b_j$.
Due to the symmetric nature of these conditions, we have $\mathcal K \in \mathbb K(j,k^*)$ if and only if $\bar{\mathcal K} \in \mathbb K(j,k^*)$.
Denoting 
\begin{multline*}
\Xi\big(\{(y_k,w_k)\giventh k\in 0\col k^*\}, \mathcal K\big)\\
:= \frac{1}{Z} \left( \min_{k\in \mathcal K}\{\pi(y_k)\psi_C(w_k)\} - \left[\min_{k\in\mathcal K}\{\pi(y_k)\psi_C(w_k)\} \land \max_{k\in\mathcal K^c}\{\pi(y_k)\psi_C(w_k)\} \right] \right)\\
\cdot \Upsilon\big(\{(y_k,w_k)\giventh k\in 0\col k^*\}, \mathcal K\big),
\end{multline*}
we have from \eqref{eqn:spNUTS2_prev},
\[
\Xi\big(\{(y_k,w_k)\giventh k\in 0\col k^*\}, \mathcal K\big) dy_0\, dw_0 = \Xi\big(\{(\bar y_k,\bar w_k)\giventh k\in 0\col k^*\}, \bar{\mathcal K}\big) d\bar y_0\, d\bar w_0.
\]
Thus
\begin{equation*}
  \begin{split}
    &\P[ (Y_0,W_0) \in A,~ (Y_{k^*},-W_{k^*})\in B, \\
    &\hspace{7ex}  (Y_{k^*},W_{k^*})\text{ is the $b_j$-th acceptable state and taken as the next state of the Markov chain}]\\
    &\hspace{0ex}= \int \1_A(y_0,w_0) \1_B(y_{k^*},-w_{k^*}) \sum_{\mathcal K \in \mathbb K(j,k^*)} \Xi\big(\{(y_k,w_k)\giventh k\in 0\col k^*\}, \mathcal K\big) dy_0 dw_0\\
    &\hspace{0ex}= \int \1_A(\bar y_{k^*}, -\bar w_{k^*}) \1_B(\bar y_0,\bar w_0)\sum_{\mathcal K \in \mathbb K(j,k^*)} \Xi\big(\{(\bar y_k,\bar w_k)\giventh k\in 0\col k^*\}, \bar{\mathcal K}\big) d\bar y_0 d\bar w_0\\
    &\hspace{0ex}= \int \1_A(\bar y_{k^*}, -\bar w_{k^*}) \1_B(\bar y_0,\bar w_0)\sum_{\bar{\mathcal K} \in \mathbb K(j,k^*)}  \Xi\big(\{(\bar y_k,\bar w_k)\giventh k\in 0\col k^*\}, \bar{\mathcal K}\big) d\bar y_0 d\bar w_0\\
    &\hspace{0ex}= \P[ (Y_0,W_0) \in B,~ (Y_{k^*},-W_{k^*})\in A, \\
    &\hspace{7ex}  (Y_{k^*},W_{k^*})\text{ is the $b_j$-th acceptable state and taken as the next state of the Markov chain}]\\
\end{split}
\end{equation*}
The third equality follows from the symmetry in $\mathbb K(j,k^*)$.
Summing the above equation for $j{\,\geq\,}0$ and $k^*{\,\geq\,}1$ and taking $A=A_0\sr\times \mathbb V$ and $B=B_0\sr\times \mathbb V$ for some measurable subsets $A_0$ and $B_0$ of $\mathbb X$, we establish detailed balance for the Markov chains $\left(X^{(i)}\right)_{i\in 1:M}$ constructed by Algorithm~\ref{alg:spNUTS2}.
\end{proof}

\section{Connection to the bouncy particle sampler}\label{sec:BPS}
Recently, a non-reversible, piecewise deterministic MCMC sampling method called the bouncy particle sampler (BPS) has been proposed \citep{peters2012rejection, bouchard2018bouncy}.
The BPS constructs a rejection free, continuous time Markov chain.
A key advantage of the BPS is that it allows for local updates of the target variables, meaning that the algorithm can update one subset of the target variables at a time while the rest of the variables evolve according to a flow that is easy to compute.
However, this algorithm has a limitation in terms of the target distributions it can be used for, because the user needs to be able to draw the arrival times of a non-homogeneous Poisson process which have a rate depending on the gradient of the target density.
In this section, we present a new, discrete time version of BPS, which is not rejection free.
This discrete time BPS is readily applicable to any target distribution with evaluable unnormalized density.
We note that there exists an alternative discrete time BPS, which was given in \citet[Algorithm~4]{vanetti2017piecewise}. 

We describe the discrete time BPS algorithm we propose within the framework of MCMC algorithms using deterministic kernels (see Section~\ref{sec:sppdDescrip}).
We assume that the sample space $\mathbb X$ and the velocity space $\mathbb V$ are both given by $\mathbb R^d$.
As in Section~\ref{sec:sppdDescrip}, the density of velocity distribution is denoted by $\psi(v\giventh x)$.
We suppose that for each $x \sr\in\mathbb X$, there exists a linear operator $\R_x\sr:\mathbb V \sr\to \mathbb V$, which satisfies $\R_x \circ \R_x = \id$ and $\psi(\R_x v\giventh x) = \psi(v\giventh x)$.
The time evolution map $\S_\tau$ is defined as
\[
\S_\tau(x,v) := (x-\R_x v \tau, -\R_x v).
\]
The self-inverse property of the linear operator $\R_x$ ensures that the absolute determinant of $\R_x$ when viewed as a matrix is equal to unity.
This translates into unit Jacobian determinant of the map $\R_x$, so the condition \eqref{eqn:psiR}, $\frac{\psi(\R_x v\giventh x)}{\psi(v\giventh x)} \left| \frac{\partial \R_x v}{\partial v} \right| = 1$, is satisfied.
It can also be readily checked that the condition \eqref{eqn:TSTS} is satisfied for $\S_\tau$ and that
\[
\left|\frac{\partial \S_\tau(x,v)}{\partial (x,v)}\right| = 1, \quad \forall(x,v)\sr\in \mathbb X\sr\times \mathbb V.
\]
A sequential-proposal discrete time BPS can be obtained as a specific case of Algorithm~\ref{alg:sppd} where $\S_\tau$ and $\R_x$ are given as above.
Algorithm~\ref{alg:spdBPS} gives a pseudocode for the sequential-proposal discrete time BPS.

\begin{figure}[t]
  \centering
  \scalebox{.89}{\begin{minipage}{\textwidth}
\begin{algorithm}[H]
  \SetKwInOut{Input}{Input}\SetKwInOut{Output}{Output}
  \Input{
    Distribution of the maximum number of proposals and the number of accepted proposals, $\nu(N,L)$\\
    Time step length distribution, $\mu(d\tau)$\\
    Velocity distribution density, $\psi(v\giventh x)$\\
    Reflection operators $\{\R_x\}$, $\{\R'_x\}$\\
    Velocity refreshment probability, $p^{\text{ref}}(x)$\\
    Number of iterations, $M$}
  \vspace{1ex}
  \Output{A draw of Markov chain, $\left(X^{(i)}\right)_{i=1,\dots,M}$}
  \vspace{1ex}
  \textbf{Initialize:} Set $X^{(0)}$ arbitrarily and draw $V^{(0)}\sim \psi(\,\cdot\giventh X^{(i+1)}).$
  
  \For {$i\gets 0\col M{-}1$}{
    Draw $N,L\sim \nu(\cdot, \cdot)$\\
    Draw $\tau \sim \mu(\cdot)$\\
    Draw $\Lambda\sim \text{unif}(0,1)$\\
    Set $X^{(i+1)} \gets X^{(i)}$ and $V^{(i+1)} \gets \R_{X^{(i)}}V^{(i)}$\\
    Set $n_a \gets 0$\\
    Set $(Y_0,W_0) \gets (X^{(i)}, V^{(i)})$\\
    \For {$n \gets 1\col N$} {
      Set $(Y_n, W_n) = (Y_{n-1} \sr- \R_{Y_{n-1}}W_{n-1}\tau,\, - \R_{Y_{n-1}}W_{n-1})$\\
      \textbf{if} {$\displaystyle \Lambda < \frac{\pi(Y_n)}{\pi(X)}$} \textbf{then} $n_a \gets n_a + 1$\\
      \If {$n_a = L$} {
        Set $(X^{(i+1)}, V^{(i+1)}) \gets (Y_n, W_n)$\\
        \texttt{break}
      }
    }
    With probability $p^{\text{ref}}(X^{(i+1)})$, refresh $V^{(i+1)}\sim \psi(\,\cdot\giventh X^{(i+1)})$ 
  }
  \caption{Sequential-proposal discrete time bouncy particle sampler}
\label{alg:spdBPS}
\end{algorithm}
  \end{minipage}}
  \end{figure}

A convenient choice for $\psi$ is a multivariate Gaussian density
\[
\psi(v\giventh x) = \psi_C(v) := \frac{1}{\sqrt{2\pi}^d \left| \det C \right|^{1/2}} \exp\{-v^T C^{-1} v\},
\]
where $C$ is a positive definite matrix.
In this case, the conditions \eqref{eqn:RR} and \eqref{eqn:psiR} hold if and only if
$$\R_x = C^{1/2} (I-2P) C^{-1/2}$$
for a symmetric projection matrix $P$, that is $P P {=} P$ and $P^T {=} P$.
A matrix $P$ is a symmetric projection matrix in $\mathbb{R}^d$ if and only if it is a projection onto the linear span of a subset of an orthonormal basis of $\mathbb{R}^d$, that is $P = \sum_{j\in A} e^j (e^j)^T$ for some $A \subseteq \{1,2,\dots,d\}$ and some orthonormal basis $(e^1, \dots, e^d)$.

A possible choice for $\R_x$ includes $-\id$, in which case the proposal map is given by $\S_\tau(x,v) = (x+v\tau, v)$.
Since the map $\S_\tau$ can be readily evaluated, this choice has a computational advantage when multiple sequential proposals are made.
Another sensible choice for the operator $\R_x$ is the reflection on the hyperplane perpendicular to the gradient of the log target density under the metric given by $C^{-1}$.
We write $U(x) := -\log \pi(x)$ and denote the velocity reflection operator by $\R_{\nabla U(x)}$.
This velocity reflection operator can be written as 
\begin{equation}
 \R_{\nabla U(x)} v := v - 2\frac{\langle \nabla U(x), v \rangle_C}{ \lVert \nabla U(x) \rVert_C^2} \nabla U(x),
  \label{eqn:R_BPS}\end{equation}
where $\langle u, w \rangle_C := u^T C^{-1} v$.
This $\R_{\nabla U(x)}$ is the reflection operator used by the original BPS algorithm of \citet{peters2012rejection}.
Since
\[
\S_\tau^2(x,v) = (x-\R_x v \tau + \R_{x-\R_x v} \R_x v \tau, \, \R_{x-\R_x v} \R_x v),
\]
we have, in the case where $\tau$ is small such that $\nabla U(x-\R_x v \tau) \approx \nabla U(x)$,
\[
\S_\tau^2(x,v) \approx \left(x + 2 \frac{\langle \nabla U(x), v \rangle_C}{\lVert \nabla U(x) \rVert_C^2} \nabla U(x), \, v\right).
\]
Therefore, repeated application of the map $\S_\tau$ has an approximate net effect of moving the particles along the gradient of the log target density.

\subsection{Numerical examples}
\begin{figure}[t]
  \centering
  \begin{subfigure}[b]{\textwidth}
    \centering
    \includegraphics[width=.3\textwidth]{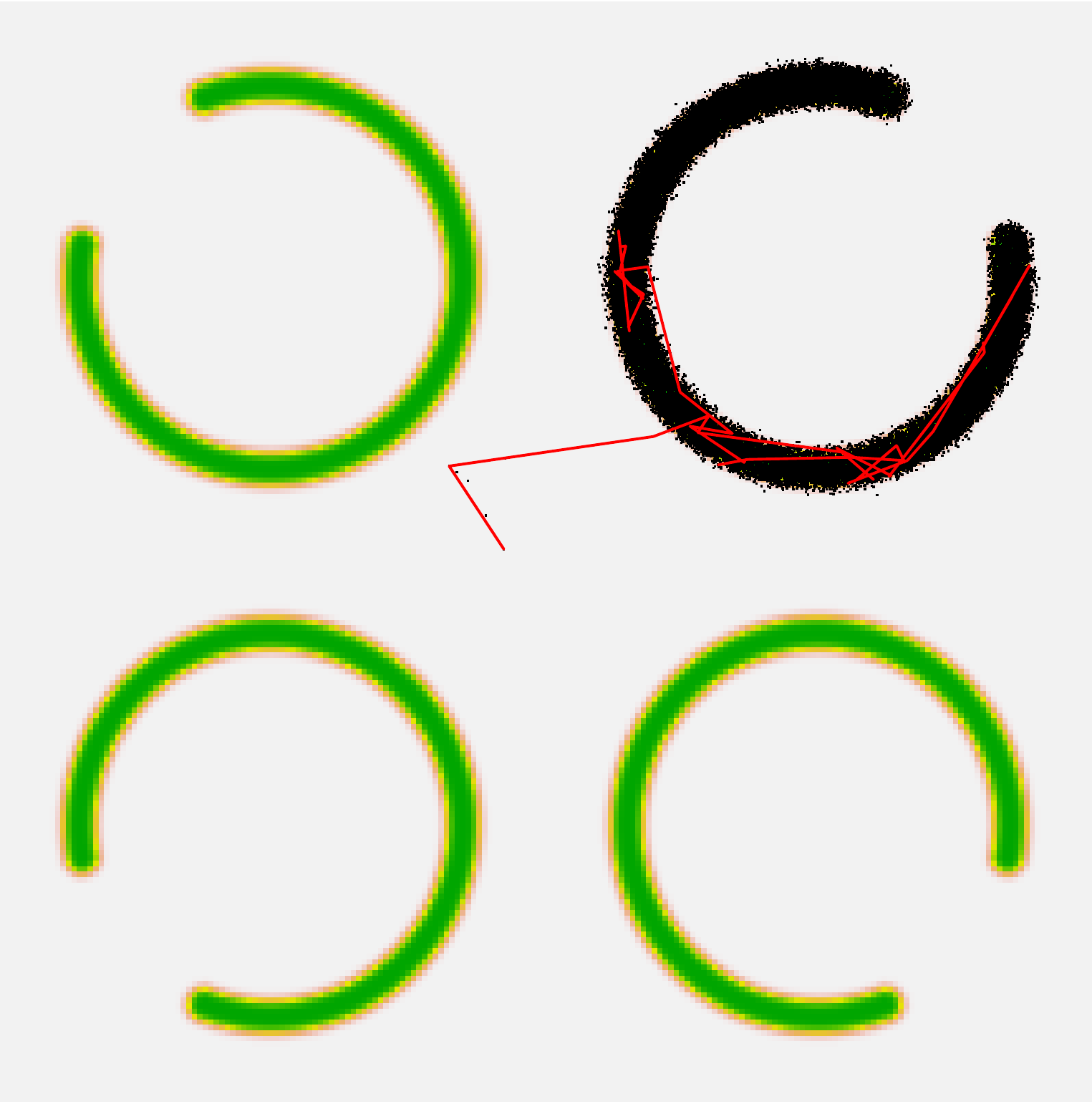} 
    \includegraphics[width=.3\textwidth]{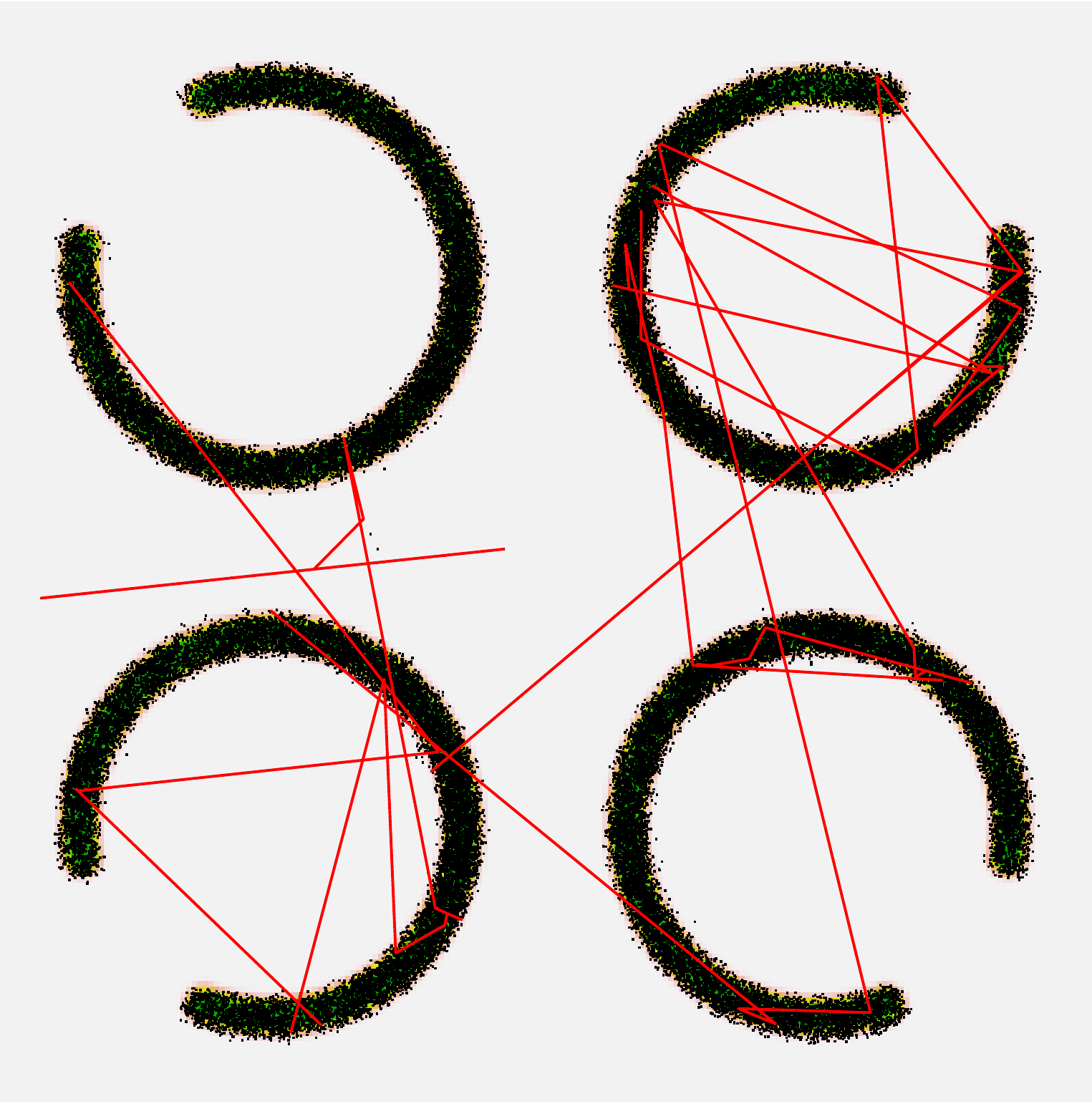} 
    \includegraphics[width=.3\textwidth]{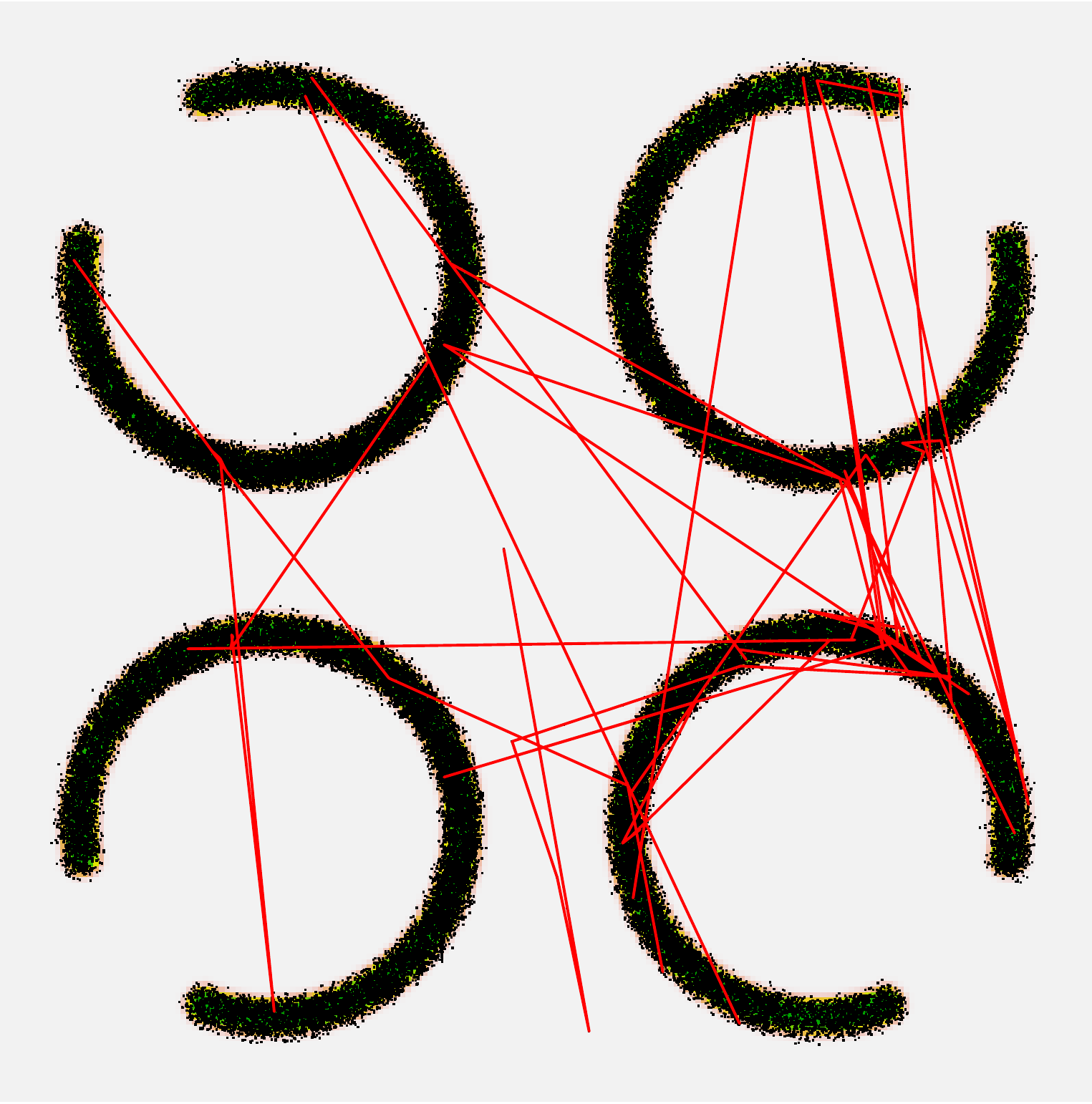} 
    \caption{The first 120,000 states in the constructed Markov chain are shown as black dots. The trajectory connecting every fourth point is shown by red segments up to one hundred points. Left, $N\sr=1$; middle, $N\sr=10$; right, $N\sr=20$.}
    \label{fig:4C_N_diagram}
  \end{subfigure}

  \vspace{1ex}
  \begin{subfigure}[b]{\textwidth}
    \centering
    \includegraphics[width=0.9\textwidth]{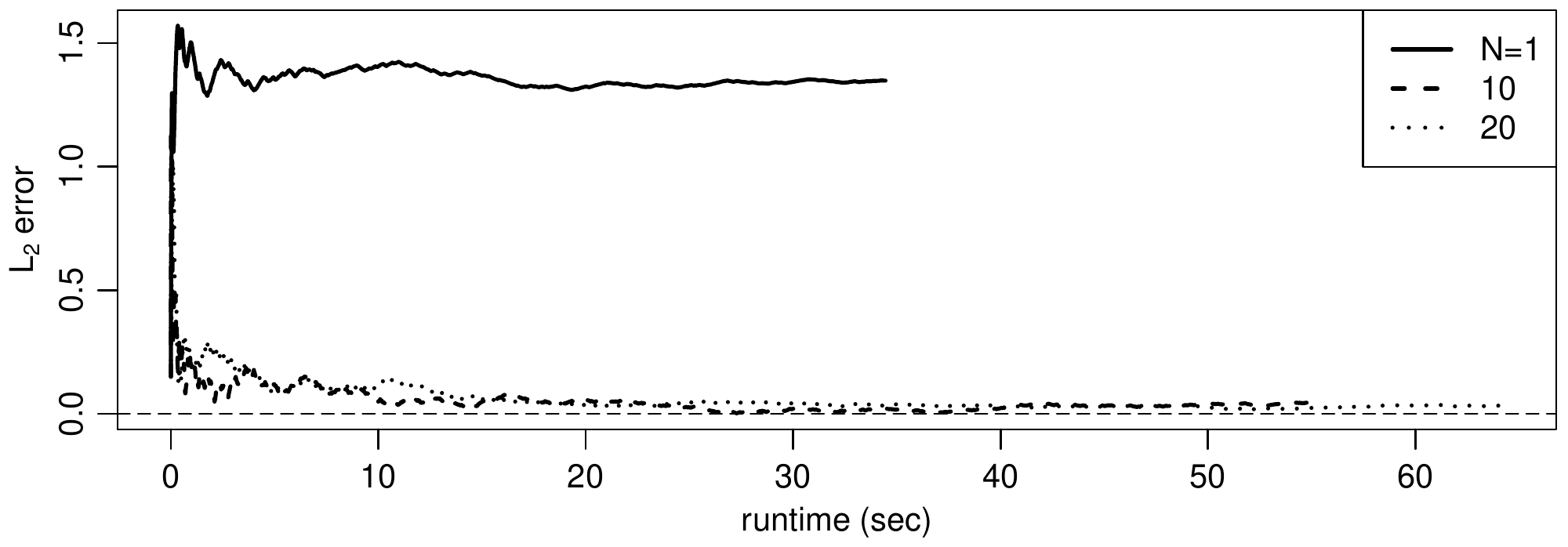}
    \caption{The distance between the sample mean and the center of the distribution as a function of runtime.}
    \label{fig:4C_N_convergence}
  \end{subfigure}
  \caption{Numerical results of the sequential-proposal discrete time BPS (Algorithm~\ref{alg:spdBPS}) for the model with four ``C''s.}
  \label{fig:4C_N}
\end{figure}
To graphically illustrate the performance of the sequential-proposal discrete time BPS (Algorithm~\ref{alg:spdBPS}), we created a target distribution defined on a unit square.
The regions of high likelihood density look like four open rings, or four rotated letters of ``C'', as shown in Figure~\ref{fig:4C_N}.
We applied the sequential-proposal discrete time BPS on this model with varying algorithmic parameters.
In every experiment, we ran the algorithm up to 120,000 iterations, where the number of acceptable proposals $L$ was fixed at one and the jump size $\tau$ at each iteration varied uniformly between $0.08$ and $0.12$.

Figure~\ref{fig:4C_N_diagram} shows the 120,000 sampled points as black dots.
The target density is represented by a color map on a green-white scale at the background.
Starting from the initial point, the trajectory of every fourth point is shown by red segments.
We varied the maximum number of proposals $N$ from one to ten and twenty.
We used the reflection operator $\R_x \sr= -\id$ and the velocity refreshment probability $p^\text{ref}\equiv0.1$ for this experiment.
In the case of $N\sr=1$, where no subsequent proposals were made if the first proposal was rejected, there was no jumps between ``C''s.
As we increased $N$ to ten and twenty, the jump between the ``C''s happened more frequently, and mixing happened faster.
Figure~\ref{fig:4C_N_convergence} shows the distances between the sample means of the constructed Markov chains and the center of the target distribution as a function of runtime in seconds.
The sample means clearly did not converge to the true mean when $N\sr=1$, but the sample means converged to the true mean with similar rates when $N\sr=10$ or $N\sr=20$.

\begin{figure}[t]
  \centering
  \begin{subfigure}[b]{\textwidth}
    \centering
    \includegraphics[width=.3\textwidth]{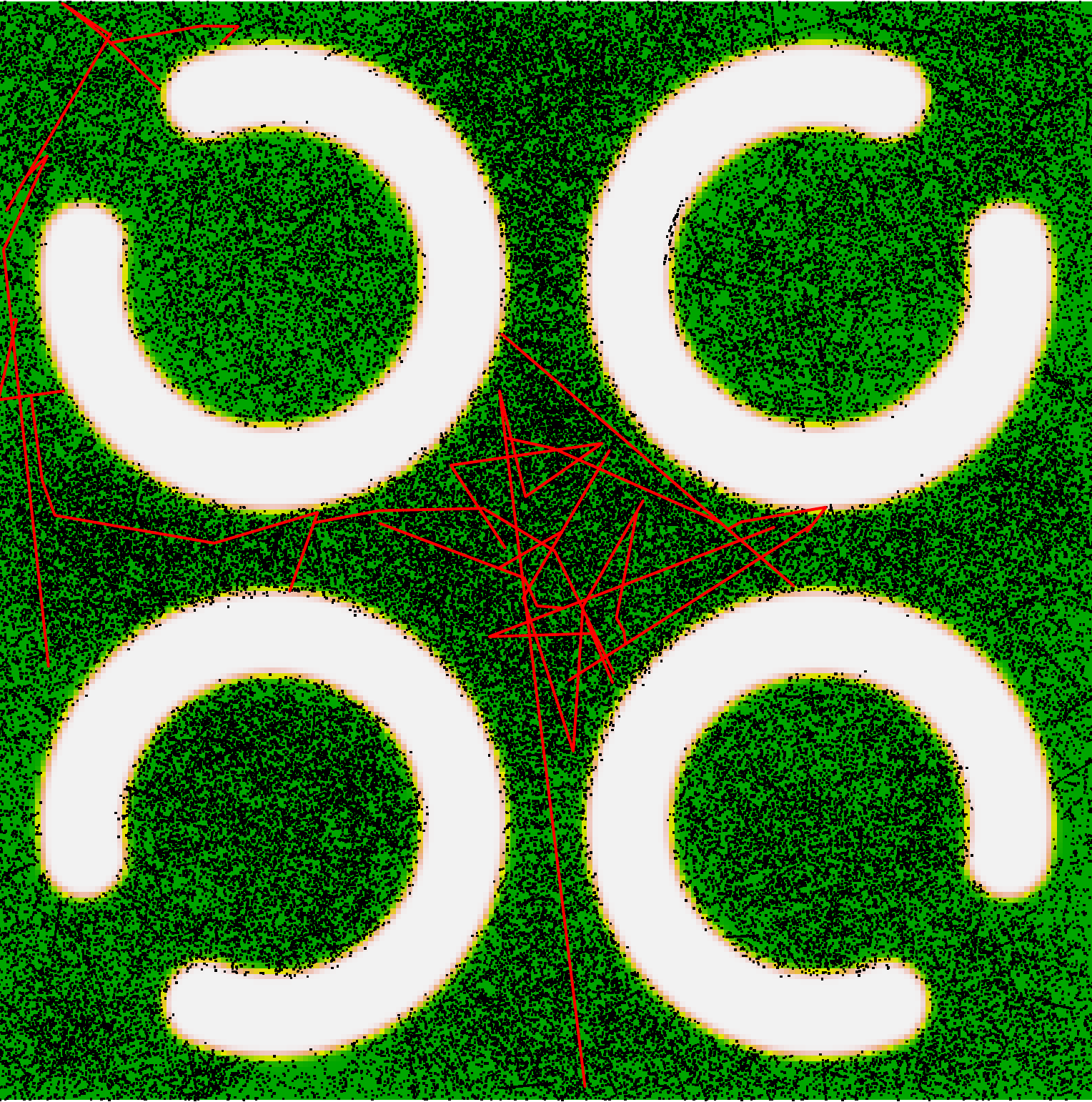} 
    \includegraphics[width=.3\textwidth]{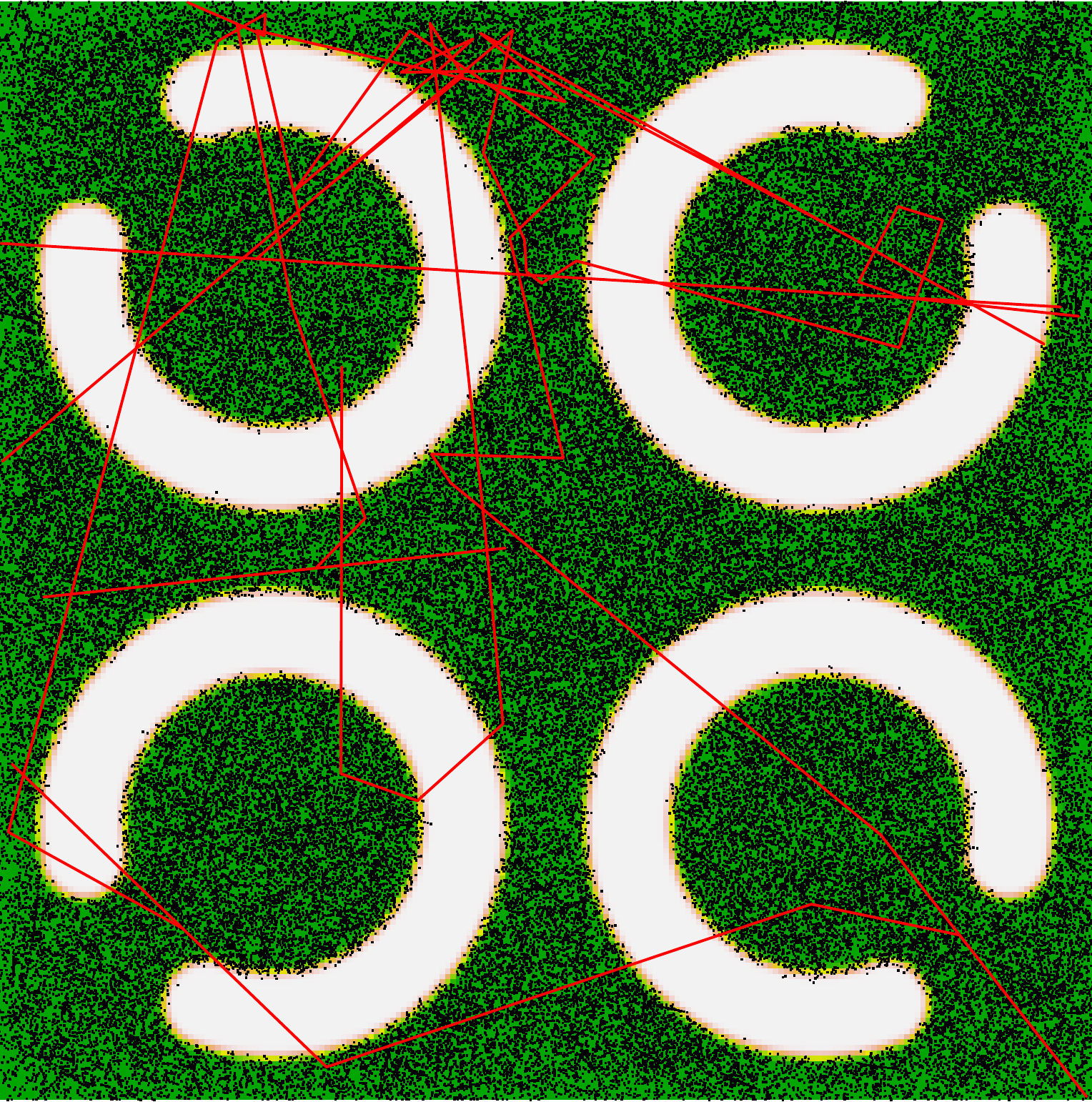} 
    \includegraphics[width=.3\textwidth]{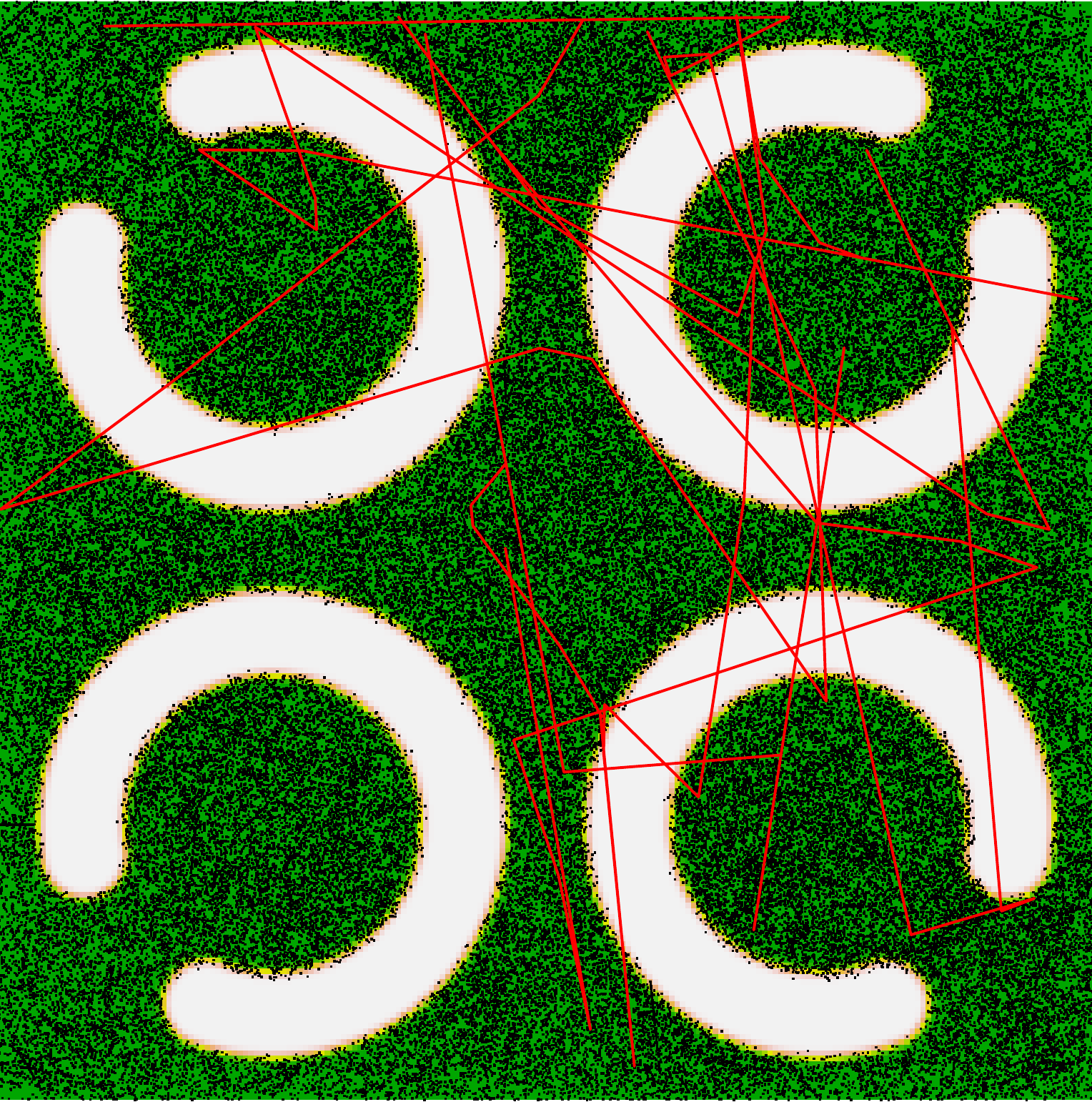} 
    \caption{The first 120,000 states in the constructed Markov chain are shown as black dots. The trajectory connecting every fourth point is shown by red segments up to one hundred points. Left, $N\sr=1$; middle, $N\sr=10$; right, $N\sr=20$.}
    \label{fig:inverted_4C_N_diagram}
  \end{subfigure}

  \vspace{1ex}
  \begin{subfigure}[b]{\textwidth}
    \centering
    \includegraphics[width=0.9\textwidth]{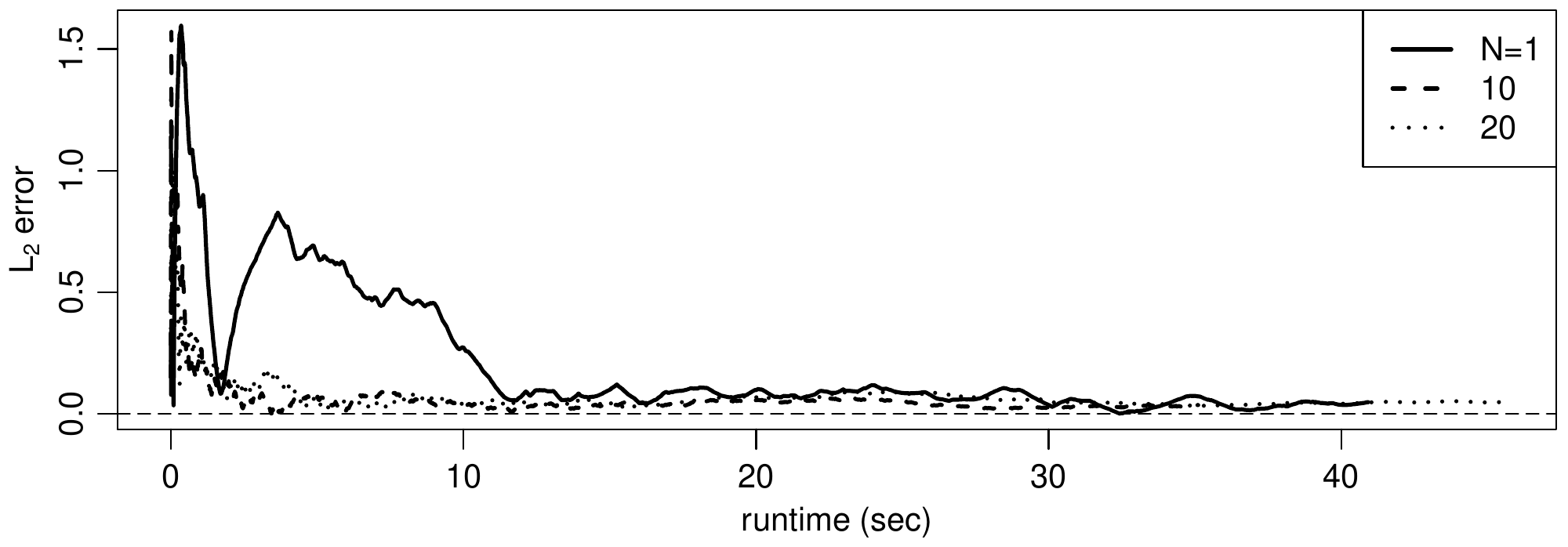} 
    \caption{The distance between the sample mean and the center of the distribution as a function of runtime.}
    \label{fig:inverted_4C_N_convergence}
  \end{subfigure}
  \caption{Numerical results of the sequential-proposal discrete time BPS (Algorithm~\ref{alg:spdBPS}) for the inverted model with four ``C''s.}
  \label{fig:inverted_4C_N}
\end{figure}
  
Figure~\ref{fig:inverted_4C_N} shows the same experiment, when the target density was inverted from the original model (i.e., the log target density was multiplied by $-1$).
The four ``C''s acted as barriers that were difficult for particles to pass through.
The velocity reflection operator $\R_x\sr= -\id$ was used, and $p^{\text{ref}}$ was 0.1.
For $N\sr=1$, there were no jumps across the barriers.
For $N\sr=10$ or $20$, however, the jumps across the barriers happened frequently.

\begin{figure}[t]
  \centering
  \begin{subfigure}[b]{\textwidth}
    \centering
    \begin{tabular}{lccc}
      & $\R_x = -\id$ &  $\R_x = \R_{\nabla U(x)}$ & mixture \\
      \parbox[b][16ex][t]{7ex}{$p^\text{ref}\sr=0$} &
      \includegraphics[width=.29\textwidth]{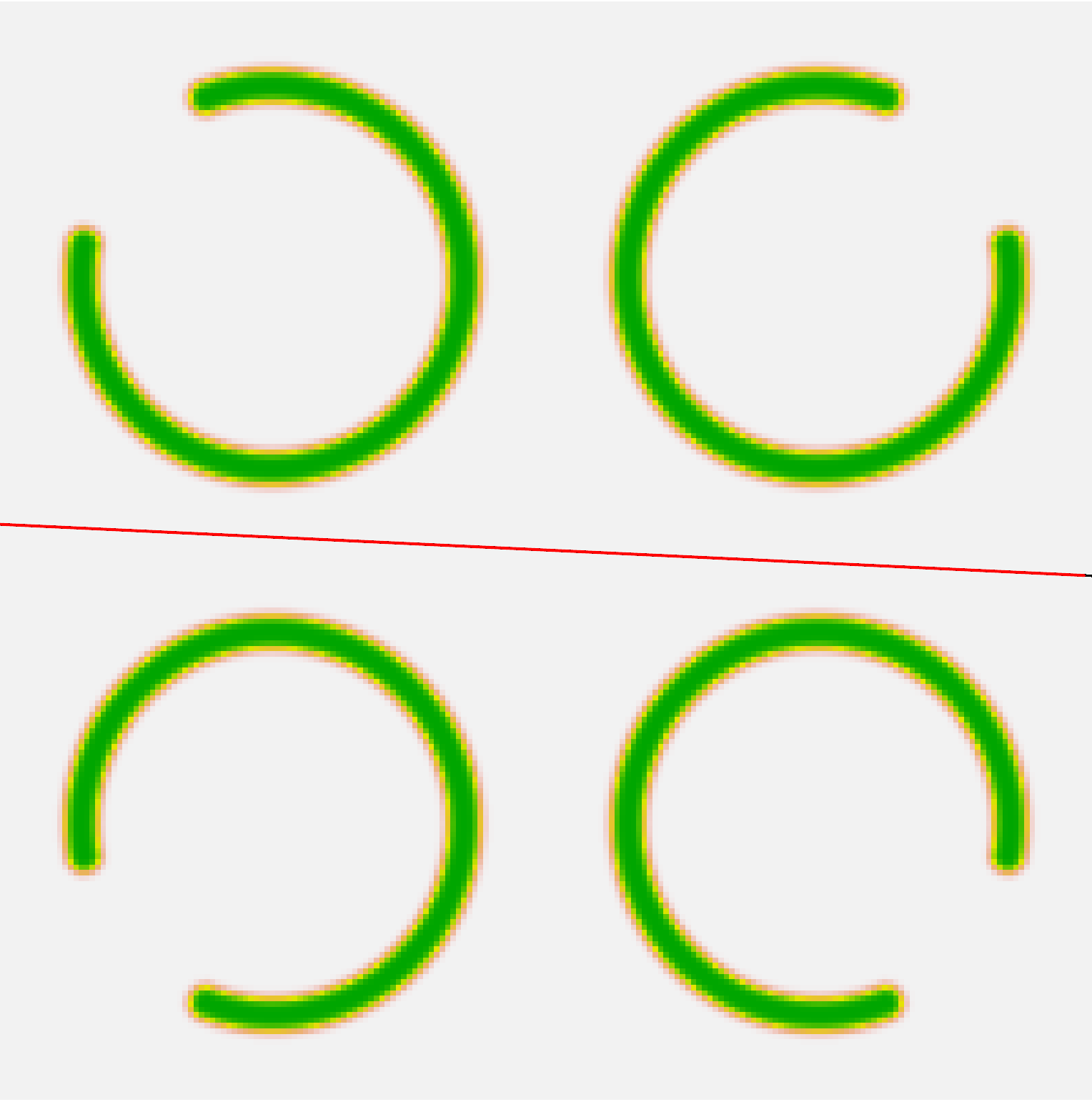}&
      \includegraphics[width=.29\textwidth]{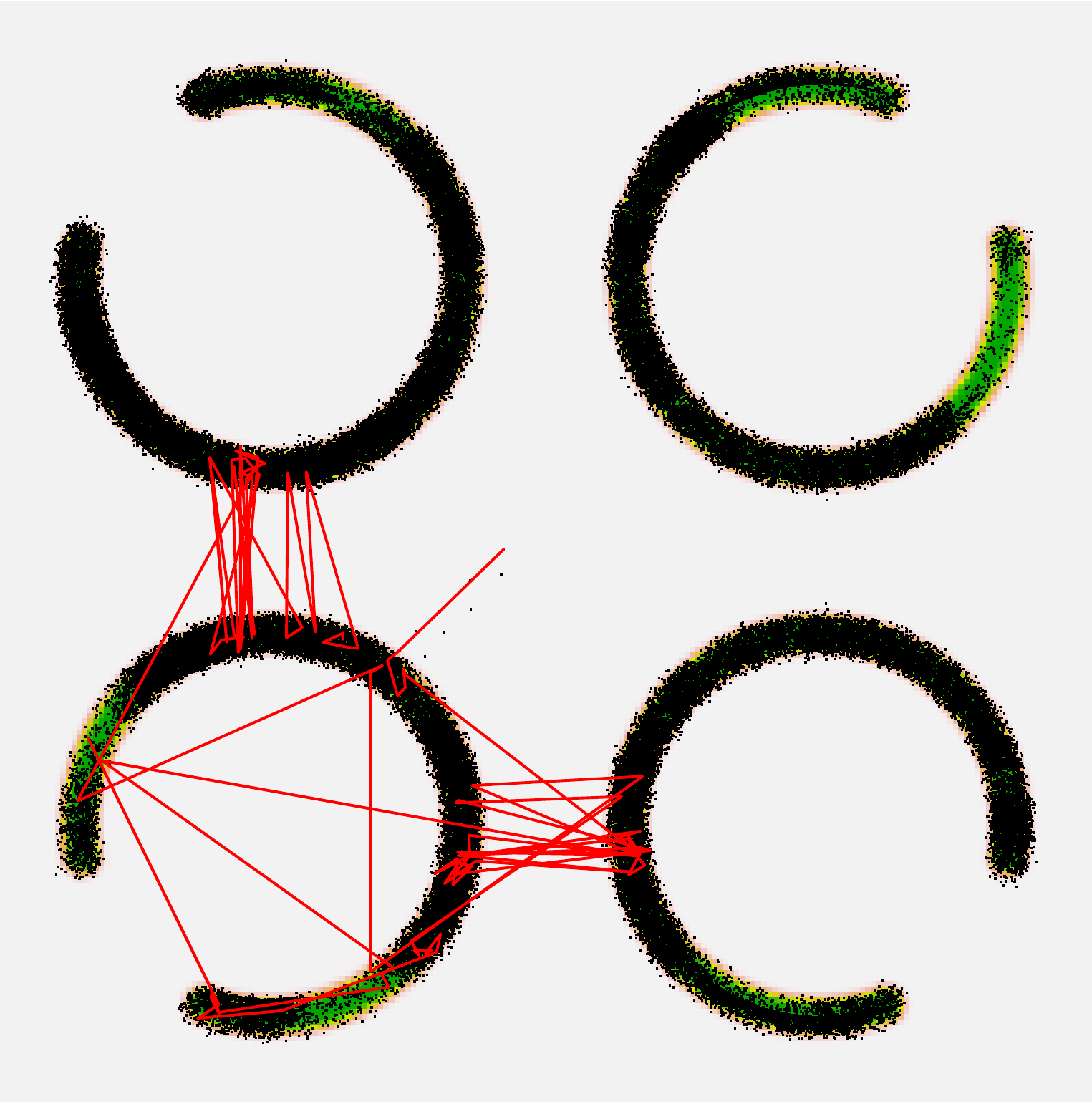}&
      \includegraphics[width=.29\textwidth]{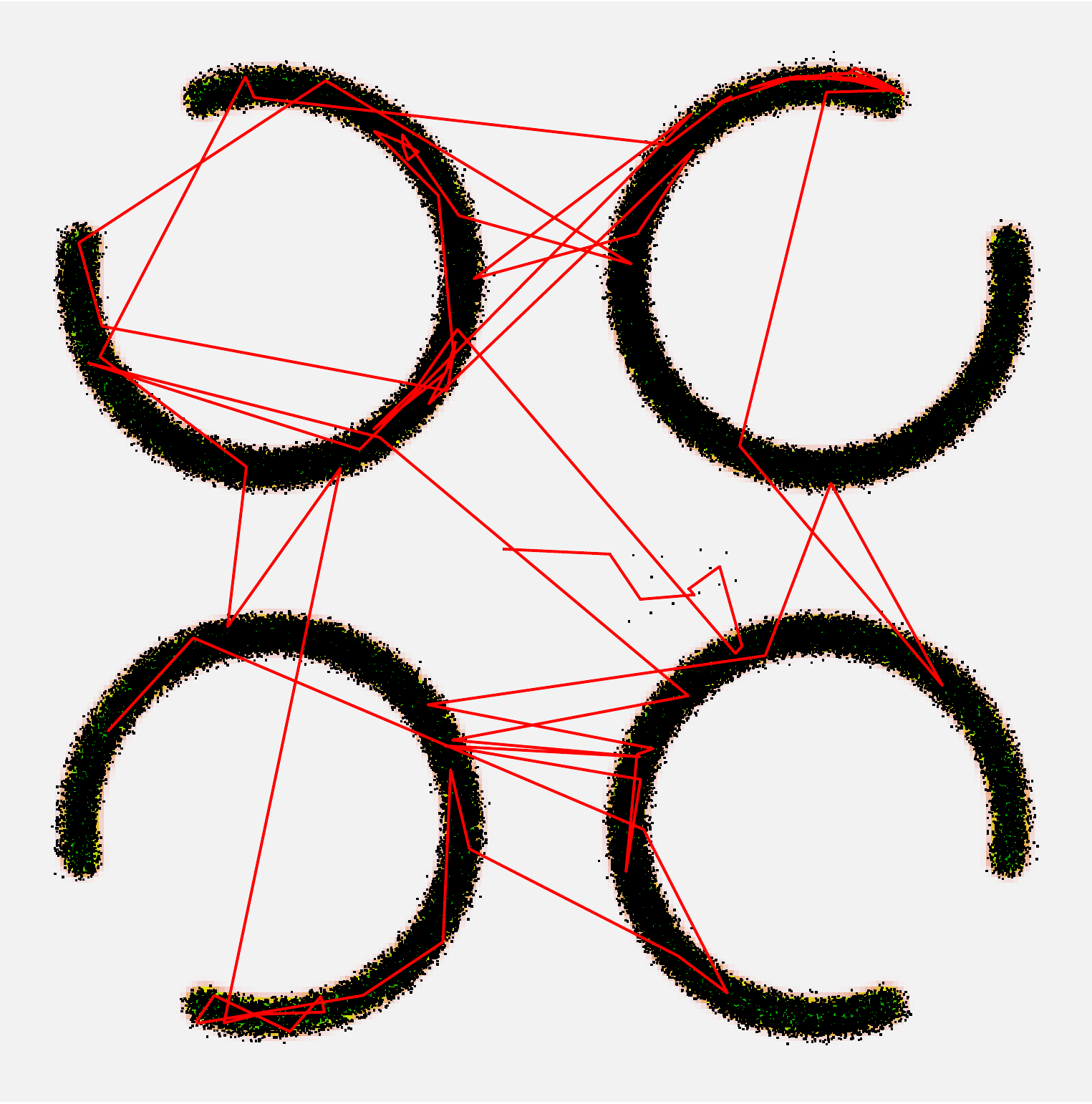}\\
      \parbox[b][16ex][t]{7ex}{$p^\text{ref}\sr=0.1$\\} &
      \includegraphics[width=.29\textwidth]{figures/C4model/{c4_in_BPS_N20L1pref0.1eps0.1Rtype_negId_M120000_0}.pdf}&
      \includegraphics[width=.29\textwidth]{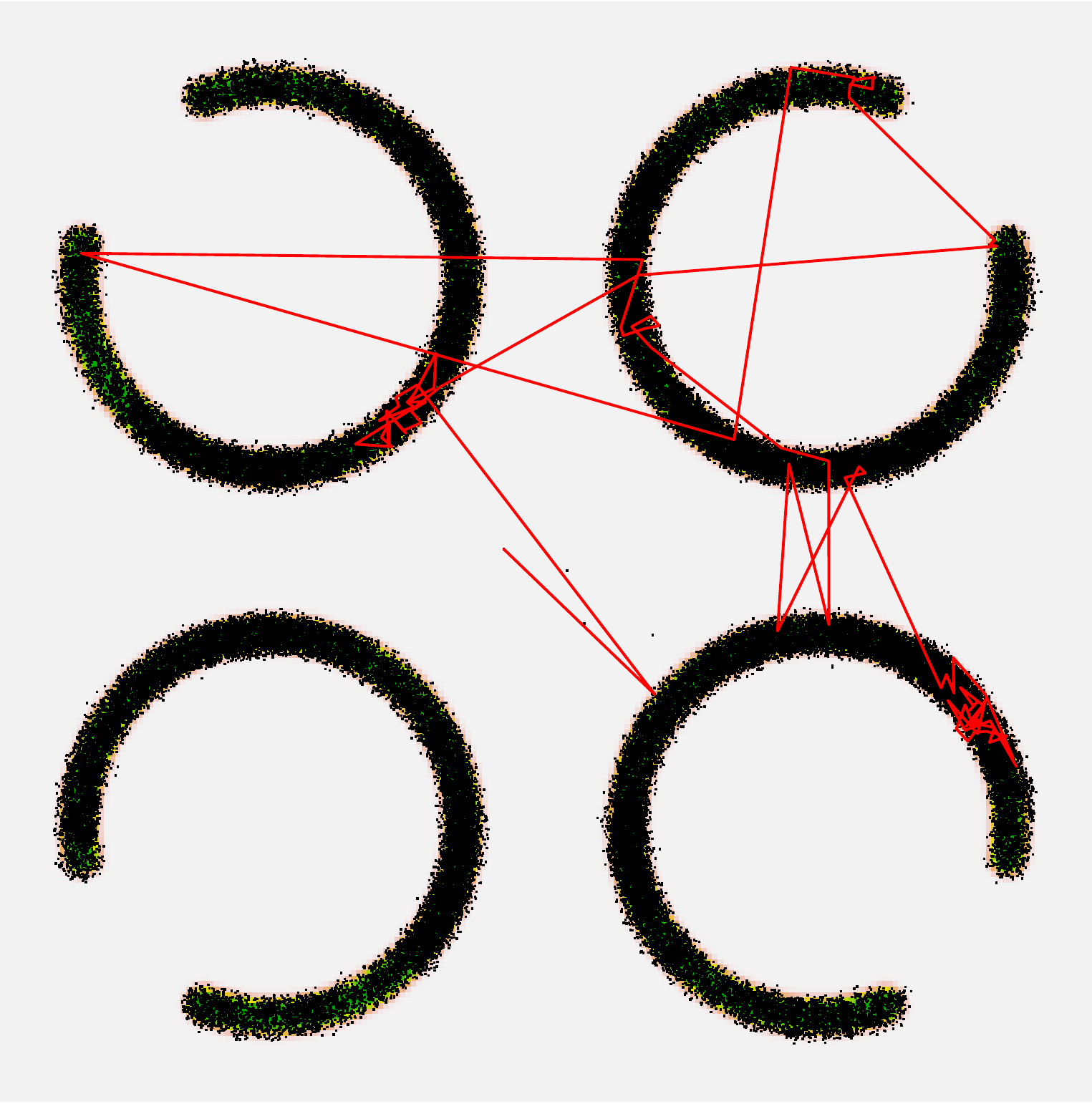}&
      \includegraphics[width=.29\textwidth]{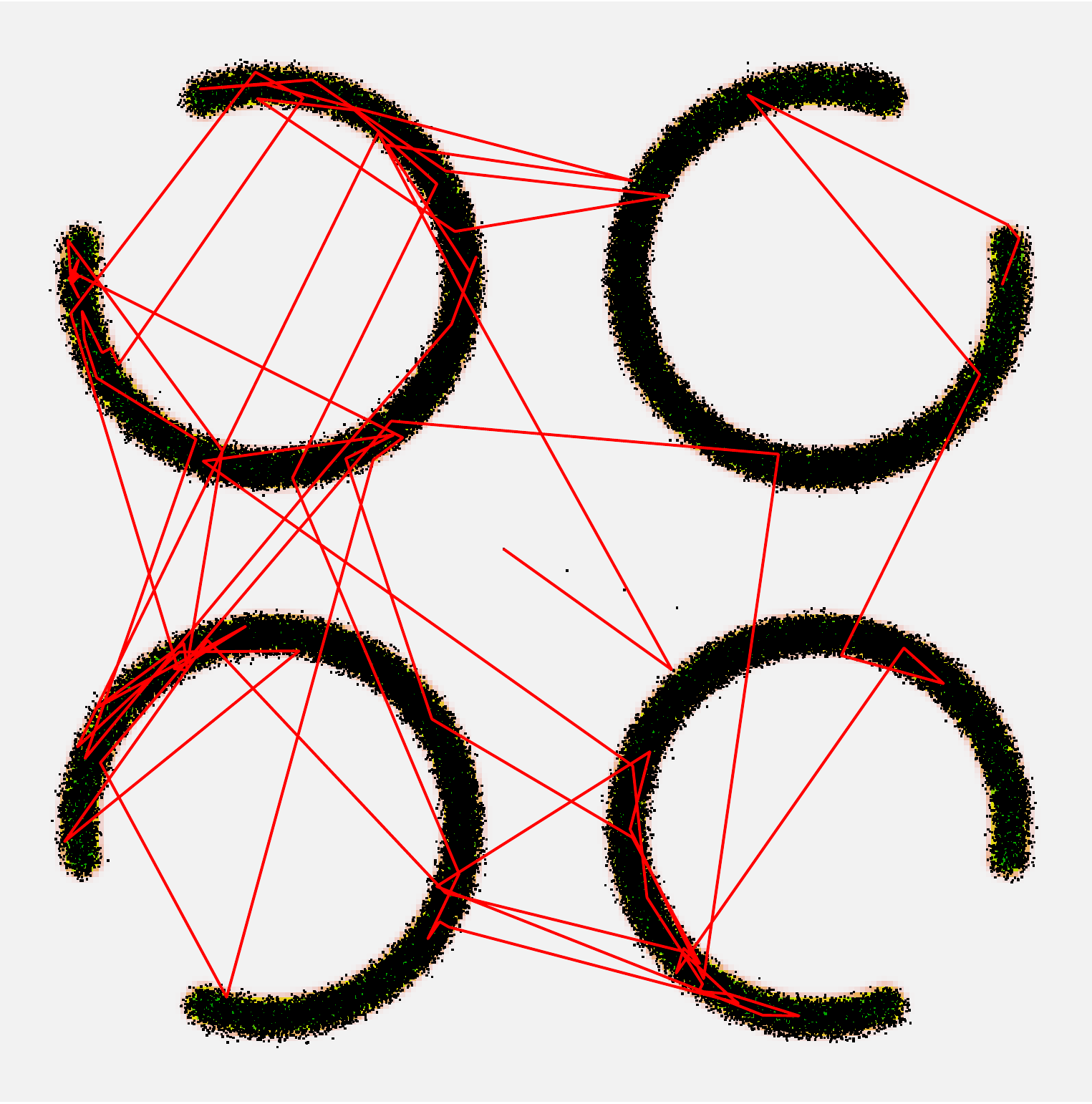}
    \end{tabular}
    \caption{The first 120,000 states in the constructed Markov chain are shown as black dots. The trajectory connecting every fourth point is shown by red segments up to one hundred points.}
    \label{fig:4C_refop_pref_diagram}
  \end{subfigure}

  \vspace{1ex}
  \begin{subfigure}[b]{\textwidth}
    \centering
    \includegraphics[width=0.9\textwidth]{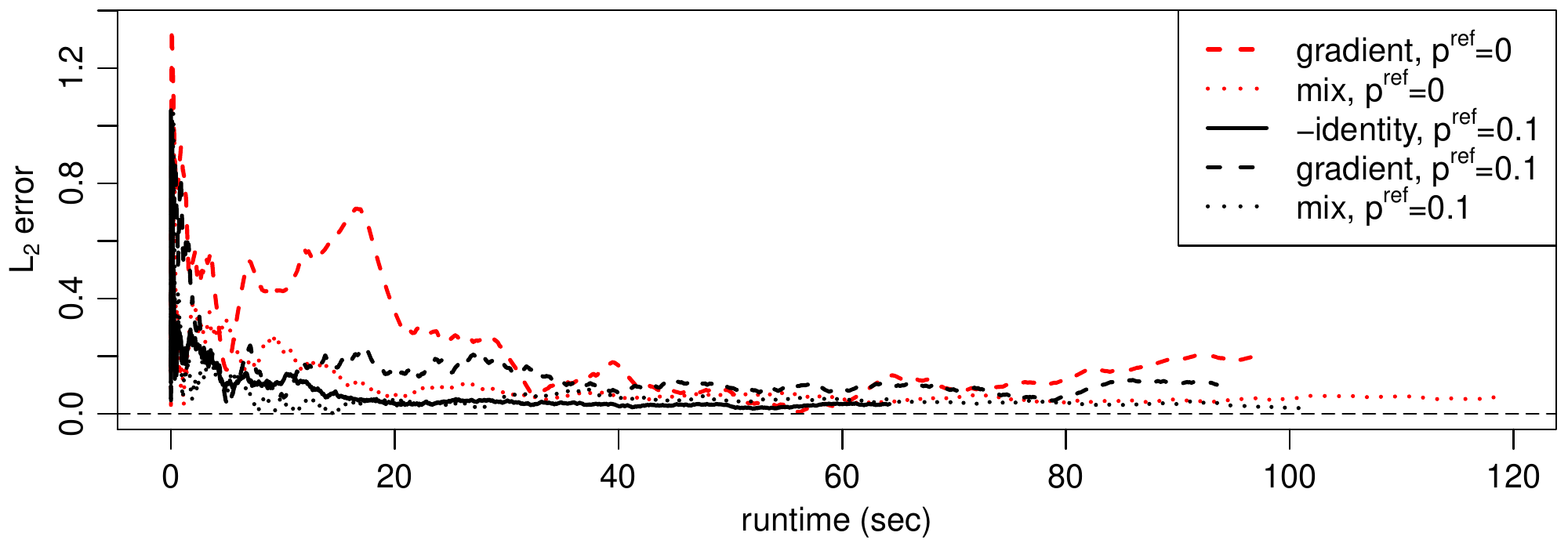}
    \caption{The distance between the sample mean and the center of the distribution as a function of runtime.
      In the legend, ``gradient'' means $\R_x \sr=\R_{\nabla U(x)}$, ``$-$identity'' means $\R_x \sr=-\id$, and ``mix'' indicates the case where these two operators are used with equal probability.}
    \label{fig:4C_refop_pref_convergence}
  \end{subfigure}
  \caption{Numerical results for the four ``C'' model with varying $p^\text{ref}$ and reflection operators.}
  \label{fig:4C_refop_pref}
\end{figure}
  
Figure~\ref{fig:4C_refop_pref} shows the numerical results for the original four ``C'' model for various choices of velocity reflection operator and velocity refreshment probability.
The maximum number of proposals $N$ was fixed at twenty.
The left column shows the results when the reflection operator $\R_x \sr= -\id$ was used.
In the middle column, the reflection operator $\R_{\nabla U(x)}$ defined in \eqref{eqn:R_BPS} was used.
In the right column, the reflection operator was randomly chosen between $-\id$ and $\R_{\nabla U(x)}$ with equal probability whenever the reflection operator was used by the algorithm.
The reflection operator $\R_x\sr= -\id$ resulted in a non-ergodic Markov chain when we did not refresh the velocity (top row, left).
When $p^\text{ref}=0.1$, the choice of $\R_{\nabla U(x)}$ resulted in a slower convergence to the target distribution compared to the case $\R_x \sr= -\id$.
The speed of convergence was improved if the algorithm used both $-\id$ and $\R_{\nabla U(x)}$ with equal probability.
From these results, we see that it is crucial to occasionally refresh the velocity for certain choices of $\R_x$ and that using a mixture of different velocity reflection operators can speed up the mixing of the Markov chain.


\clearpage
\begin{thebibliography}{49}
\expandafter\ifx\csname natexlab\endcsname\relax\def\natexlab#1{#1}\fi
\expandafter\ifx\csname url\endcsname\relax
  \def\url#1{\texttt{#1}}\fi
\expandafter\ifx\csname urlprefix\endcsname\relax\def\urlprefix{URL: }\fi

\bibitem[{Andrieu and Atchade(2007)}]{andrieu2007efficiency}
Andrieu, C. and Atchade, Y. (2007) On the efficiency of adaptive {MCMC}
  algorithms.
\newblock \textit{Electronic Communications in Probability}, \textbf{12},
  336--349.

\bibitem[{Andrieu and Livingstone(2019)}]{andrieu2019peskun}
Andrieu, C. and Livingstone, S. (2019) Peskun-{T}ierney ordering for {M}arkov
  chain and process {M}onte {C}arlo: beyond the reversible scenario.
\newblock \textit{arXiv preprint arXiv:1906.06197}.

\bibitem[{Andrieu and Moulines(2006)}]{andrieu2006ergodicity}
Andrieu, C. and Moulines, {\'E}. (2006) On the ergodicity properties of some
  adaptive {MCMC} algorithms.
\newblock \textit{The Annals of Applied Probability}, \textbf{16}, 1462--1505.

\bibitem[{Andrieu and Thoms(2008)}]{andrieu2008tutorial}
Andrieu, C. and Thoms, J. (2008) A tutorial on adaptive {MCMC}.
\newblock \textit{Statistics and Computing}, \textbf{18}, 343--373.

\bibitem[{Atchad{\'e} and Fort(2010)}]{atchade2010limit}
Atchad{\'e}, Y. and Fort, G. (2010) Limit theorems for some adaptive
  {M}{C}{M}{C} algorithms with subgeometric kernels.
\newblock \textit{Bernoulli}, \textbf{16}, 116--154.

\bibitem[{Atchade and Fort(2012)}]{atchade2012limit}
Atchade, Y.~F. and Fort, G. (2012) Limit theorems for some adaptive
  {M}{C}{M}{C} algorithms with subgeometric kernels: {P}art {II}.
\newblock \textit{Bernoulli}, \textbf{18}, 975--1001.

\bibitem[{Atchad{\'e} and Rosenthal(2005)}]{atchade2005adaptive}
Atchad{\'e}, Y.~F. and Rosenthal, J.~S. (2005) On adaptive {M}arkov chain
  {M}onte {C}arlo algorithms.
\newblock \textit{Bernoulli}, \textbf{11}, 815--828.

\bibitem[{Beskos et~al.(2013)Beskos, Pillai, Roberts, Sanz-Serna and
  Stuart}]{beskos2013optimal}
Beskos, A., Pillai, N., Roberts, G., Sanz-Serna, J.-M. and Stuart, A. (2013)
  Optimal tuning of the hybrid {M}onte {C}arlo algorithm.
\newblock \textit{Bernoulli}, \textbf{19}, 1501--1534.

\bibitem[{Bouchard-C{\^o}t{\'e} et~al.(2018)Bouchard-C{\^o}t{\'e}, Vollmer and
  Doucet}]{bouchard2018bouncy}
Bouchard-C{\^o}t{\'e}, A., Vollmer, S.~J. and Doucet, A. (2018) The bouncy
  particle sampler: A nonreversible rejection-free {M}arkov chain {M}onte
  {C}arlo method.
\newblock \textit{Journal of the American Statistical Association},
  \textbf{113}, 855--867.

\bibitem[{Calderhead(2014)}]{calderhead2014general}
Calderhead, B. (2014) A general construction for parallelizing
  {M}etropolis--{H}astings algorithms.
\newblock \textit{Proceedings of the National Academy of Sciences},
  \textbf{111}, 17408--17413.

\bibitem[{Campos and Sanz-Serna(2015)}]{campos2015extra}
Campos, C.~M. and Sanz-Serna, J. (2015) Extra chance generalized hybrid {M}onte
  {C}arlo.
\newblock \textit{Journal of Computational Physics}, \textbf{281}, 365--374.

\bibitem[{Dua and Graff(2017)}]{dua2019}
Dua, D. and Graff, C. (2017) {UCI} machine learning repository.
\newblock \urlprefix\url{http://archive.ics.uci.edu/ml}.

\bibitem[{Duane et~al.(1987)Duane, Kennedy, Pendleton and
  Roweth}]{duane1987hybrid}
Duane, S., Kennedy, A.~D., Pendleton, B.~J. and Roweth, D. (1987) Hybrid
  {M}onte {C}arlo.
\newblock \textit{Physics letters B}, \textbf{195}, 216--222.

\bibitem[{Fang et~al.(2014)Fang, Sanz-Serna and Skeel}]{fang2014compressible}
Fang, Y., Sanz-Serna, J.~M. and Skeel, R.~D. (2014) Compressible generalized
  hybrid {M}onte {C}arlo.
\newblock \textit{The Journal of chemical physics}, \textbf{140}, 174108.

\bibitem[{Geyer(1991)}]{geyer1991markov}
Geyer, C.~J. (1991) Markov chain {M}onte {C}arlo maximum likelihood.
\newblock \textit{Interface Foundation of North America}.
\newblock Retrieved from the University of Minnesota Digital Conservancy.

\bibitem[{Goodman and Weare(2010)}]{goodman2010ensemble}
Goodman, J. and Weare, J. (2010) Ensemble samplers with affine invariance.
\newblock \textit{Communications in applied mathematics and computational
  science}, \textbf{5}, 65--80.

\bibitem[{Green and Mira(2001)}]{green2001delayed}
Green, P.~J. and Mira, A. (2001) Delayed rejection in reversible jump
  {M}etropolis--{H}astings.
\newblock \textit{Biometrika}, \textbf{88}, 1035--1053.

\bibitem[{Gupta et~al.(1990)Gupta, Irb{\"a}c, Karsch and
  Petersson}]{gupta1990acceptance}
Gupta, S., Irb{\"a}c, A., Karsch, F. and Petersson, B. (1990) The acceptance
  probability in the hybrid {M}onte {C}arlo method.
\newblock \textit{Physics Letters B}, \textbf{242}, 437--443.

\bibitem[{Haario et~al.(2001)Haario, Saksman and Tamminen}]{haario2001adaptive}
Haario, H., Saksman, E. and Tamminen, J. (2001) An adaptive {M}etropolis
  algorithm.
\newblock \textit{Bernoulli}, \textbf{7}, 223--242.

\bibitem[{Haario et~al.(2005)Haario, Saksman and
  Tamminen}]{haario2005componentwise}
--- (2005) Componentwise adaptation for high dimensional {MCMC}.
\newblock \textit{Computational Statistics}, \textbf{20}, 265--273.

\bibitem[{Hastings(1970)}]{hastings1970monte}
Hastings, W.~K. (1970) {M}onte {C}arlo sampling methods using {M}arkov chains
  and their applications.
\newblock \textit{Biometrika}, \textbf{57}, 97--109.

\bibitem[{Hoffman and Gelman(2014)}]{hoffman2014no}
Hoffman, M.~D. and Gelman, A. (2014) The {N}o-{U}-turn sampler: adaptively
  setting path lengths in {H}amiltonian {M}onte {C}arlo.
\newblock \textit{Journal of Machine Learning Research}, \textbf{15},
  1593--1623.

\bibitem[{Horowitz(1991)}]{horowitz1991generalized}
Horowitz, A.~M. (1991) A generalized guided {M}onte {C}arlo algorithm.
\newblock \textit{Physics Letters B}, \textbf{268}, 247--252.

\bibitem[{Hukushima and Nemoto(1996)}]{hukushima1996exchange}
Hukushima, K. and Nemoto, K. (1996) Exchange {M}onte {C}arlo method and
  application to spin glass simulations.
\newblock \textit{Journal of the Physical Society of Japan}, \textbf{65},
  1604--1608.

\bibitem[{Kou et~al.(2006)Kou, Zhou, Wong et~al.}]{kou2006equi}
Kou, S., Zhou, Q., Wong, W.~H. et~al. (2006) Equi-energy sampler with
  applications in statistical inference and statistical mechanics.
\newblock \textit{The Annals of Statistics}, \textbf{34}, 1581--1619.

\bibitem[{Leimkuhler and Reich(2004)}]{leimkuhler2004simulating}
Leimkuhler, B. and Reich, S. (2004) \textit{Simulating {H}amiltonian dynamics},
  vol.~14.
\newblock Cambridge university press.

\bibitem[{Liouville(1838)}]{liouville1838note}
Liouville, J. (1838) Note on the theory of the variation of arbitrary
  constants.
\newblock \textit{Journal de Math\'{e}matiques Pures et Appliqu\'{e}es},
  \textbf{3}, 342--349.

\bibitem[{Liu et~al.(2000)Liu, Liang and Wong}]{liu2000multiple}
Liu, J.~S., Liang, F. and Wong, W.~H. (2000) The multiple-try method and local
  optimization in {M}etropolis sampling.
\newblock \textit{Journal of the American Statistical Association},
  \textbf{95}, 121--134.

\bibitem[{Marinari and Parisi(1992)}]{marinari1992simulated}
Marinari, E. and Parisi, G. (1992) Simulated tempering: a new {M}onte {C}arlo
  scheme.
\newblock \textit{EPL (Europhysics Letters)}, \textbf{19}, 451.

\bibitem[{Metropolis et~al.(1953)Metropolis, Rosenbluth, Rosenbluth, Teller and
  Teller}]{metropolis1953equation}
Metropolis, N., Rosenbluth, A.~W., Rosenbluth, M.~N., Teller, A.~H. and Teller,
  E. (1953) Equation of state calculations by fast computing machines.
\newblock \textit{Journal of Chemical Physics}, \textbf{21}, 1087--1092.

\bibitem[{Mira et~al.(2001{\natexlab{a}})Mira, M{\o}ller and
  Roberts}]{mira2001perfect}
Mira, A., M{\o}ller, J. and Roberts, G.~O. (2001{\natexlab{a}}) Perfect slice
  samplers.
\newblock \textit{Journal of the Royal Statistical Society: Series B
  (Statistical Methodology)}, \textbf{63}, 593--606.

\bibitem[{Mira et~al.(2001{\natexlab{b}})}]{mira2001metropolis}
Mira, A. et~al. (2001{\natexlab{b}}) On metropolis-hastings algorithms with
  delayed rejection.
\newblock \textit{Metron}, \textbf{59}, 231--241.

\bibitem[{Neal(2011)}]{neal2011mcmc}
Neal, R. (2011) {M}{C}{M}{C} using {H}amiltonian dynamics.
\newblock \textit{Handbook of {M}arkov Chain {M}onte {C}arlo}, \textbf{2}.

\bibitem[{Neal(1994)}]{neal1994improved}
Neal, R.~M. (1994) An improved acceptance procedure for the hybrid {M}onte
  {C}arlo algorithm.
\newblock \textit{Journal of Computational Physics}, \textbf{111}, 194--203.

\bibitem[{Neal et~al.(2003)}]{neal2003slice}
Neal, R.~M. et~al. (2003) Slice sampling.
\newblock \textit{The Annals of Statistics}, \textbf{31}, 705--767.

\bibitem[{Peskun(1973)}]{peskun1973optimum}
Peskun, P.~H. (1973) Optimum {M}onte-{C}arlo sampling using markov chains.
\newblock \textit{Biometrika}, \textbf{60}, 607--612.

\bibitem[{Peters et~al.(2012)}]{peters2012rejection}
Peters, E.~A. et~al. (2012) Rejection-free {M}onte {C}arlo sampling for general
  potentials.
\newblock \textit{Physical Review E}, \textbf{85}, 026703.

\bibitem[{Plummer et~al.(2006)Plummer, Best, Cowles and
  Vines}]{plummer2006coda}
Plummer, M., Best, N., Cowles, K. and Vines, K. (2006) {CODA}: convergence
  diagnosis and output analysis for {MCMC}.
\newblock \textit{R news}, \textbf{6}, 7--11.

\bibitem[{{R Core Team}(2018)}]{R}
{R Core Team} (2018) \textit{R: A Language and Environment for Statistical
  Computing}.
\newblock R Foundation for Statistical Computing, Vienna, Austria.
\newblock \urlprefix\url{https://www.R-project.org/}.

\bibitem[{Roberts et~al.(1997)Roberts, Gelman and Gilks}]{roberts1997weak}
Roberts, G.~O., Gelman, A. and Gilks, W.~R. (1997) Weak convergence and optimal
  scaling of random walk {M}etropolis algorithms.
\newblock \textit{The Annals of Applied Probability}, \textbf{7}, 110--120.

\bibitem[{Roberts and Rosenthal(1998)}]{roberts1998optimal}
Roberts, G.~O. and Rosenthal, J.~S. (1998) Optimal scaling of discrete
  approximations to {L}angevin diffusions.
\newblock \textit{Journal of the Royal Statistical Society: Series B
  (Statistical Methodology)}, \textbf{60}, 255--268.

\bibitem[{Roberts and Rosenthal(1999)}]{roberts1999convergence}
--- (1999) Convergence of slice sampler {M}arkov chains.
\newblock \textit{Journal of the Royal Statistical Society: Series B
  (Statistical Methodology)}, \textbf{61}, 643--660.

\bibitem[{Roberts and Rosenthal(2007)}]{roberts2007coupling}
--- (2007) Coupling and ergodicity of adaptive {M}arkov chain {M}onte {C}arlo
  algorithms.
\newblock \textit{Journal of Applied Probability}, \textbf{44}, 458--475.

\bibitem[{Sexton and Weingarten(1992)}]{sexton1992hamiltonian}
Sexton, J. and Weingarten, D. (1992) Hamiltonian evolution for the hybrid
  {M}onte {C}arlo algorithm.
\newblock \textit{Nuclear Physics B}, \textbf{380}, 665--677.

\bibitem[{Sherlock et~al.(2010)Sherlock, Fearnhead and
  Roberts}]{sherlock2010random}
Sherlock, C., Fearnhead, P. and Roberts, G.~O. (2010) The random walk
  {M}etropolis: Linking theory and practice through a case study.
\newblock \textit{Statistical Science}, \textbf{25}, 172--190.

\bibitem[{Sohl-Dickstein et~al.(2014)Sohl-Dickstein, Mudigonda and
  DeWeese}]{sohl2014hamiltonian}
Sohl-Dickstein, J., Mudigonda, M. and DeWeese, M.~R. (2014) {H}amiltonian
  {M}onte {C}arlo without detailed balance.
\newblock \textit{arXiv preprint arXiv:1409.5191}.

\bibitem[{Tierney(1998)}]{tierney1998note}
Tierney, L. (1998) A note on {M}etropolis-{H}astings kernels for general state
  spaces.
\newblock \textit{The Annals of Applied Probability}, \textbf{8}, 1--9.

\bibitem[{Tierney and Mira(1999)}]{tierney1999some}
Tierney, L. and Mira, A. (1999) Some adaptive {M}onte {C}arlo methods for
  {B}ayesian inference.
\newblock \textit{Statistics in Medicine}, \textbf{18}, 2507--2515.

\bibitem[{Vanetti et~al.(2017)Vanetti, Bouchard-C{\^o}t{\'e}, Deligiannidis and
  Doucet}]{vanetti2017piecewise}
Vanetti, P., Bouchard-C{\^o}t{\'e}, A., Deligiannidis, G. and Doucet, A. (2017)
  Piecewise deterministic {M}arkov chain {M}onte {C}arlo.
\newblock \textit{arXiv preprint arXiv:1707.05296}.

\end{thebibliography}
\end{document}